\documentclass[a4paper,11pt]{article}

\usepackage{aligned-overset}
\usepackage{authblk}
\usepackage{amsmath, amsthm}
\usepackage{tensor}
\usepackage[style=numeric-comp,maxbibnames=99,bibencoding=utf8,firstinits=true]{biblatex}
\usepackage[hmargin={26mm,26mm},vmargin={30mm,35mm}]{geometry}
\usepackage{nicefrac}
\usepackage{cancel}
\usepackage{colortbl}
\usepackage[colorlinks,allcolors={blue}]{hyperref}
\usepackage[compact,small]{titlesec}
\usepackage{tikz-cd}
\usepackage{subcaption}
\usepackage{pgfplots}
\usepackage{loglogslopetriangle}
\usepackage{newtxtext}
\usepackage{newtxmath}
\usepackage{booktabs, tabularx}

\bibliography{df-einstein}
\AtBeginBibliography{\footnotesize}

\newcommand{\email}[1]{\href{mailto:#1}{#1}}

\numberwithin{equation}{section}

\newtheorem{theorem}{Theorem}
\newtheorem{proposition}[theorem]{Proposition}

\theoremstyle{remark}
\newtheorem{remark}[theorem]{Remark}
\theoremstyle{definition}

\newcommand{\Real}{\mathbb{R}}

\DeclareMathOperator{\sign}{sign}

\DeclareRobustCommand{\bvec}[1]{\boldsymbol{#1}}
\DeclareMathOperator{\CURL}{\bf curl}
\DeclareMathOperator{\DIV}{div}

\newcommand{\kform}[1][k]{\Lambda^{#1}}
\newcommand{\tkform}[1][k]{\kform[#1]_{\perp}}
\newcommand{\ed}{{\rm d}}        % exterior derivative
\newcommand{\sed}{\hat{\rm d}}   % spatial exterior derivative
\newcommand{\ip}[1]{i_{#1}}  % interior product
\newcommand{\ipn}{\ip{\bvec{n}}}
\newcommand{\ipm}{\ip{\bvec{m}}}
\newcommand{\ep}[1]{j_{#1}}      % exterior product
\newcommand{\epn}{\ep{\bvec{n}}}
\newcommand{\hs}{\star}
\newcommand{\shs}{\hat{\hs}}     % spatial hodge star
\newcommand{\chs}{{\hs}_c}     % constant hodge star
\newcommand{\ld}[1]{{\mathscr{L}}_{#1}}    % Lie derivative
\newcommand{\ldn}{\ld{\bvec{n}}}
\newcommand{\ldm}{\ld{\bvec{m}}}
\newcommand{\inc}{\iota}

\newcommand{\Poly}[2][]{\mathcal{P}_{#2}^{#1}}

\newcommand{\ded}[2]{\ed^{#1}_{#2}}
\newcommand{\ued}[2]{\underline{\ed}^{#1}_{#2}}
\newcommand{\Xec}[2]{\underline{X}^{#1}_{#2}}
\newcommand{\Lpdf}[3]{L^{#1}\Lambda^{#2}(#3)}

\newcommand{\Pkdf}[3]{\Poly{#1}\Lambda^{#2}(#3)}
\newcommand{\tPkdf}[3]{\Poly[-]{#1}\Lambda^{#2}(#3)}
\newcommand{\Pec}[2]{P^{#1}_{#2}}
\newcommand{\Iec}[2]{\underline{I}^{#1}_{#2}}
\newcommand{\tr}{\mathrm{tr}}

\newcommand{\ddt}{\frac{\partial}{\partial t}}

\newcommand{\Mh}[1][h]{\mathcal{M}_{#1}}
\newcommand{\Th}[1][h]{\mathcal{T}_{#1}}
\newcommand{\Fh}[1][h]{\mathcal{F}_{#1}}
\newcommand{\Eh}[1][h]{\mathcal{E}_{#1}}
\newcommand{\Vh}{\mathcal{V}_h}

\newcommand{\Hcurl}[1]{\bvec{H}(\CURL;#1)}
\newcommand{\Hdiv}[1]{\bvec{H}(\DIV;#1)}

\newcommand{\ul}[2][]{\underline{#2}_{#1}}
\newcommand{\ulf}[1]{\ul[f]{#1}}
\newcommand{\ulh}[1]{\ul[h]{#1}}

\renewcommand{\exp}{{\rm e}}

\begin{document}

\title{A polytopal discrete de Rham scheme for the exterior calculus Einstein's equations}

\author[1]{Todd A. Oliynyk}
\author[1]{Jia Jia Qian}

\affil[1]{School of Mathematics, Monash University, Melbourne, Australia, \email{todd.oliynyk@monash.edu}, \email{jia.qian@monash.edu}}

\maketitle
%\tableofcontents

%------------------------------------------------------------------------------%
\abstract

In this work, based on the $3+1$ decomposition in \cite{Fecko:97, Olivares.Peshkov.ea:22}, we present a fully exterior calculus breakdown of spacetime and Einstein's equations. Links to the orthonormal frame approach \cite{Van-Elst.Uggla:97} are drawn to help understand the variables in this context. Two formulations are derived, discretised and tested using the exterior calculus discrete de Rham complex \cite{Bonaldi.Di-Pietro.ea:25}, and some discrete quantities are shown to be conserved in one of the cases.

\section{Introduction}

Numerical relativity is the field of solving Einstein's field equations using numerical techniques, enabling the modelling of general relativistic phenomena such as binary black holes, that is then followed by real-life observations such as the first detection of gravitational waves \cite{Abbott.Abbott.ea:16} by the Laser Interferometer Gravitational-Wave Observatory (LIGO). Since the proposal of the Arnowitt–Deser–Misner (ADM) formulation \cite{Arnowitt.Deser.ea:08, York:79} to the breakthrough of the first stable, long term evolution of black holes using generalised harmonic coordinates \cite{Pretorius:05}, the development and refinement of methods continues as we increase our understanding of these equations.

The design of numerical methods for partial differential equations (PDEs) generally begins by first recasting the equations as an initial value problem. The most common approach for Einstein's equations is to consider the $4$-dimensional spacetime as a family of $3$-dimensional spacelike hypersurfaces paramatrised by time, called the $3+1$ formalism of general relativity (GR) \cite{Gourgoulhon:12}. This decomposes Einstein's equations into a set of evolution and constraint equations. In the continuous setting, the constraints propagate; the evolution makes it so that they stay true for all time if they are true initially. In practice, this property often fails to hold, in particular for schemes based on the famous strongly hyperbolic Baumgarte-Shapiro-Shibata-Nakamura-Oohara-Kojima (BSSNOK) \cite{Nakamura.Oohara.ea:87, Shibata.Nakamura:95, Baumgarte.Shapiro:98, Beyer.Sarbach:04} system. It has been shown that in certain cases, the growing constraint violations can be catastrophic, impacting the stability of the simulation, e.g.\cite{Brodbeck.Frittelli.ea:99, Alic.Bona-Casas.ea:12}. To deal with this, there are procedures that employ constraint dampening, such as those based on the Z4/CCZ4/Z4c formalism \cite{Bona.Ledvinka.ea:03, Alic.Bona-Casas.ea:12, Alic.Kastaun.ea:13, Bernuzzi.Hilditch:10}, that use Lagrange multipliers in order to control the deviation of the constraints \cite{Gundlach.Calabrese.ea:05}. These methods do not attempt at all to achieve exact constraint preservation, that is nevertheless an innate quality of the equations.

The philosophy of compatible discretisations is to reproduce properties of the continuous equations by using numerical methods that mimic the underlying geometric structures. In computational electromagnetism, it is widely known that the divergence free condition of the magnetic field $\DIV B=0$ can be naturally obtained by respecting the property $\DIV\CURL=0$. This is done by using discrete fields and operators that form a discrete complex. Examples include the standard finite element methods (FEM), or its generalisation the finite element exterior calculus (FEEC) \cite{Arnold.Falk.ea:10}, discrete exterior calculus \cite{Desbrun.Hirani.ea:05}, compatible discrete operators \cite{Arnold.Bochev.ea:07}, and finally the discrete de Rham (DDR) method \cite{Di-Pietro.Droniou.ea:20, Di-Pietro.Droniou:21}. Some work has been done in \cite{Quenneville-Belair:15, Hu.Liang.ea:22} that applies certain discrete complexes to the Hodge wave formulation of the linearised Einstein--Bianchi system \cite{Friedrich:96, Anderson.Choquet-Bruhat.ea:97}, a formalism based on a decomposition of the Weyl tensor that resembles the Maxwell equations, to get naturally the linearised Bianchi identity. The application of these techniques to Einstein is still underexplored, partially due to the need to reformulate GR into suitable systems that can properly take advantage of these approaches. A recent development in this direction is presented in \cite{Olivares.Peshkov.ea:22}, where a $3+1$ decomposition of the exterior calculus Einstein's equations (see also \cite{Frauendiener:06}) leads to a set of equations with a Maxwell-like structure. In this article, we follow a similar approach as \cite{Olivares.Peshkov.ea:22}, but insist on preserving the coordinate free aspect where possible in the equations by introducing spatial differential form operators using the $3+1$ theory from \cite{Fecko:97}. The interest in such a presentation is to make it more accessible for the broader numerical community that may not be familiar with the component form of these operators. Furthermore, for certain discretisations methods, where there is no direct discrete analogy of a component and the exterior derivative is handled through a single operator, it is preferable to take the coordinate free viewpoint where possible.

The exterior calculus discrete de Rham (ECDDR) complex \cite{Bonaldi.Di-Pietro.ea:25} is a fully discrete complex, replacing both the differential forms and exterior derivative by discrete constructions, while reproducing key homological properties. The main features include an arbitrary-order of accuracy and the handling of general polytopal meshes, that enables to better capture the behaviour of the solution. For complex, higher-order systems, the workload is easily parallelisable, and cost reduction techniques such as static condensation and serendipity processes can also be applied. Such polytopal methods have seen use for the Yang--Mills equations \cite{Droniou.Oliynyk.ea:23}, the Navier--Stokes equations \cite{Di-Pietro.Droniou.ea:24}, where in some cases, certain discrete versions of the constraints is shown to propagate exactly. To our knowledge we present here the first application of such an approach to the full set of GR equations; the vast majority of codes use finite differencing in space (e.g. Einstein toolkit \cite{Loffler.Faber.ea:12}, BAM \cite{Brugmann.Gonzalez.ea:08}, LazEv \cite{Campanelli.Lousto.ea:06}) with more rarely psudospectral \cite{Boyle.Brown.ea:07} or discontinuous Galerkin \cite{Teukolsky:16, Dumbser.Zanotti.ea:24}. 

As a first exploratory work that applies polytopal methods to numerical relativity, we attempt to provide as much detail as possible on the numerical relativity setting, including elements that may be considered standard knowledge, so that the construction is accessible to those that are not familiar with Einstein's equations. This paper is ordered as follows. Section \ref{sec:setting} fixes some definitions and conventions that we use for the $3+1$ decomposition of spacetime. Section \ref{sec:ec.decomp} presents the $3+1$ decomposition of exterior calculus following the separation of spacetime, including the split of the four dimensional exterior derivative and Hodge star operator. Section \ref{sec:einstein} introduces the spacetime Einstein equations in terms of differential forms, before using the theory of the preceding sections to write it as a initial value problem. The detailed calculations are collected in the appendix \ref{app:calc}. A brief overview of ECDDR  is given in Section \ref{sec:disc}, then the discrete schemes are shown, and we discuss how the nonlinearity and exact constraint preservation is handled (see also appendix \ref{sec:disc.formulas}). Finally, Section \ref{sec:tests} contains the numerical results.

\section{Setting}\label{sec:setting}

In this work we use the following index conventions: greek letters denote the spacetime indices ranging from 0 to 3, while lowercase latin letters the spatial indices (i, j, k, etc.) of the orthonormal frame, and capital latin letters the spatial indices of the canonical basis. We also use the Einstein summation: in any term, indices that appear once in superscript and once in subscript are implicitly summed over. For example,
\begin{equation*}
  T\indices{^\mu_{jk}}v^jw^k=\sum_{j=1}^3\sum_{k=1}^3T\indices{^\mu_{jk}}v^jw^k,\qquad \mu=0,1,2,3.
\end{equation*}

We follow the setting in \cite[Section 4]{Gourgoulhon:12}. Let $M$ be a $4$-dimensional Lorentzian manifold with metric $g$ with signature $(-,+,+,+)$. Let $(M,g)$ be a also a globally hyperbolic spacetime; i.e. it admits a spacelike 3-dimensional hypersurface $\Sigma$ that intersects each timelike or null curve exactly once ($M\cong \Sigma\times\Real$). Then we take a foliation of $(M,g)$; a smooth, regular (non-vanishing gradient) scalar field $t$ on $M$ such that the level sets
\begin{equation*}
  \Sigma_{\hat{t}}\coloneq\{p\in M; t(p)=\hat{t}\},\qquad\forall \hat{t}\in\Real,
\end{equation*}
are spacelike and $(\Sigma_{\hat{t}})_{\hat{t}\in\Real}$ is a partition of $M$.

Denote by $\flat:TM\to T^*M$ and $\sharp:T^*M\to TM$ the isomorphisms between vector fields and $1$-forms induced by the metric. The future directed unit normal vector (field) is
\begin{equation}\label{def:n}
\bvec{n}\coloneq-N(\ed t)^\sharp,
\end{equation}
where $N\coloneq(-g((\ed t)^\sharp,(\ed t)^\sharp))^{-\frac12}$ is called the lapse function. This is future directed in the sense $g(\bvec{n},\bvec{n})=-1$. 
The normal evolution vector is the scaled normal vector 
\begin{equation}\label{def:m}
  \bvec{m}\coloneq N\bvec{n}.
\end{equation}

Given a vector field $X$ on $M$ satsifying $\ed t(X)=1$, we define the shift vector field $\beta$ relative to $X$ as
\begin{equation*}
\beta \coloneq -m + X.
\end{equation*}
By definition, $\beta$ is the purely spatial part of $X$, since the calculation $\ed t(\beta) = -\ed t(m) + \ed t(X)=-1+1=0$ is always true. So in fact the choice of $X$ is completely parametrised by the spatial vector field $\beta$, with the trivial choice of $\beta=0$ being equivalent to setting $X=\bvec{m}$.

Given spatial coordinates $(\tilde{x}^i)$ on a single (spatial) slice $\Sigma_{\hat{t}_0}$ of the hypersurface, we can extend to a global consistent coordinate system $(t, x^i)$ on $M$ by solving the initial value problem
\begin{equation*}
\ld{X} x^i = 0, \quad x^i|_{\Sigma_{\hat{t}_0}} = \tilde{x}^i.
\end{equation*}
In these coordinates, we have that $X=\frac{\partial}{\partial t}$, and the expression $\frac{\partial}{\partial t}=\bvec{m}+\beta$.

The induced $3$-dimensional metric on $\Sigma_{\hat{t}}$ is 
\begin{equation}\label{eq:h}
h\coloneq g+\bvec{n}^\flat\otimes\bvec{n}^\flat.
\end{equation}

\section{$3+1$ exterior calculus decomposition}\label{sec:ec.decomp}

In this section, we split the $4$-dimensional exterior calculus structure on $M$ into functions and operators on the $3$-dimensional slices, based on the work in \cite{Fecko:97} that uses a more generalised congruence method to split the spacetime.

\subsection{Vector fields}

For every $t\in\Real$, $p\in \Sigma_t$, we denote the inclusion map by $\inc_t:\Sigma_t\to M$, and identify $T_p\Sigma_t$ with a subspace of $T_pM$ using the push-forward $(\inc_t)_*:T_p\Sigma_t\to T_pM$. The normal vector then splits the tangent space into
\begin{equation}\label{eq:TpM}
  T_pM=T_p\Sigma_t\oplus \textrm{span}(\bvec{n}),
\end{equation}
and we define the orthogonal projector 
\begin{equation*}
  \pi_t:T_pM\to T_p\Sigma_t, \qquad \pi_tX\coloneq X+g(\bvec{n},X)\bvec{n},
\end{equation*}
which naturally induces a map
\begin{equation*}
  \pi^*_t:T^*_p\Sigma_t\to T^*_pM, \qquad (\pi_t^*\omega)X\coloneq \omega(\pi_t X)\quad \forall X\in T_pM.
\end{equation*}

\subsection{Differential forms}

Let $X$ be a vector field on $M$. Define the interior product $\ip{X}:\kform(M)\to\kform[k-1](M)$ by $\ip{X}\omega\coloneq\omega(X,\cdots)$, and exterior product $\ep{X}:\kform(M)\to\kform[k+1](M)$ by $\ep{X}\omega\coloneq X^\flat\wedge\omega$. Then for any $k$-form $\omega$,
\[
(\ep{X}\ip{X}+\ip{X}\ep{X})\omega=\cancel{\ep{X}\ip{X}\omega}+(\ip{X}X^\flat)\wedge\omega-\cancel{X^\flat\wedge(\ip{X}\omega)}=g(X,X)\omega
\] 
and so with $X=\bvec{n}$, we get the decomposition 
\begin{equation}\label{eq:decomp}
\omega=-(\epn\ipn+\ipn\epn)\omega=-\bvec{n}^\flat\wedge\ipn\omega-\ipn\epn\omega\eqcolon-\bvec{n}^\flat\wedge \bar{\omega}+\hat{\omega},
\end{equation}
where $\bar{\omega}\coloneq\ipn\omega$ and $\hat{\omega}\coloneq-\ipn\epn\omega$. The operations $\epn\ipn$ and $\ipn\epn$ are orthogonal projections satisfying $(\epn\ipn)(\ipn\epn)=0=(\ipn\epn)(\epn\ipn)$ and $-(\epn\ipn+\ipn\epn)=1$. So this is indeed a decomposition of $\omega$ that splits the form into a part containing $\bvec{n}^\flat$ and a purely spatial part.

This leads to the definition of spatial $k$-forms as the subspace
\begin{equation}\label{def:sp.kform}
  \tkform(M)\coloneq\{\omega\in\kform(M);\omega=-\ipn\epn\omega\Leftrightarrow\ipn\omega=0\}.
\end{equation} 

These spaces with the wedge operator form a graded algebra $(\bigoplus_{k=0}^4\tkform(M),\wedge)$, which is easily seen using the second characterisation in \eqref{def:sp.kform}: For any $\omega\in\tkform(M)$ and $\mu\in\tkform[l](M)$, 
\begin{align*}
  \ipn(\omega\wedge\mu)=(\ipn\omega)\wedge\mu+(-1)^k\omega\wedge(\ipn\mu)=0,
\end{align*}
thus $\omega\wedge\mu\in\tkform[l](M)$.

\subsection{Exterior derivative}

With the same idea as for differential forms, we can examine the action of the spacetime $\ed$ on \eqref{eq:decomp} and look for a decomposition of the $\ed$ operator. A necessary formula when dealing with exterior derivatives and interior products is Cartan's magic formula: For any vector field $X$, $\ld{X}$ acts on differential forms as
\begin{equation}\label{eq:magic}
  \ld{X}=\ed\ip{X}+\ip{X}\ed.
\end{equation}
It is important to note that $\ldn,\ldm:\tkform(M)\to\tkform(M)$; i.e. they map spatial forms to spatial forms. This is easily seen using \eqref{eq:magic} and calculating $\ipn\ldn$ (resp. $\bvec{m}$),
\begin{equation*}
  \ipn\ldn=\cancel{\ipn\ed\ipn}+\cancel{\ipn\ipn\ed}=0,
\end{equation*}
since $\ipn$ is $0$ on spatial forms, and $\ipn\ipn=0$.

Define the spatial exterior derivative on spatial forms as $\sed:\tkform(M)\to\tkform[k+1](M), \sed\coloneq-\ipn\epn\ed$. This is equivalent to the exterior derivative on each $\Sigma_t$, which we can see by
\begin{equation*}
  -\inc^*_t\ipn\epn\ed\omega=-\inc^*_t(\ipn\bvec{n}^\flat\wedge\ed\omega-\bvec{n}^\flat\wedge\ipn\omega)=\inc^*_t\ed\omega=\ed\inc^*_t\omega,
\end{equation*}
using that $\inc^*_t\bvec{n}^\flat=0$ since $\bvec{n}^\flat$ disappears on $T\Sigma_t$, and that the pull-back commutes with the exterior derivative. This operator naturally satisfies the spatial identity $\sed\sed=0$.

We first deal with the action of $\ed$ on $\bvec{n}^\flat$ using the definition \eqref{def:n}; the decomposition of $\ed\bvec{n}^\flat$ can be calculated as
\begin{equation}\label{eq:dn}
\ed\bvec{n}^\flat=-\ed N\wedge\ed t=\ed N\wedge\frac1N\bvec{n}^\flat=-\bvec{n}^\flat\wedge(\frac1N\sed N),
\end{equation}
where the $\sed$ in the last term comes from \eqref{eq:decomp}, replacing $\omega$ by $\ed N$. More generally for a spatial form $\mu$, we can use \eqref{eq:decomp} to write
\begin{align}
  \ed\mu&=-\bvec{n}^\flat\wedge\ipn\ed\mu+\sed\mu \nonumber\\
        &=-\bvec{n}^\flat\wedge(\ipn\ed+\ed\ipn)\mu+\sed\mu \nonumber\\
        &=-\bvec{n}^\flat\wedge\ldn\mu+\sed\mu, \label{eq:spat.d}
\end{align}
where $\ed\ipn$ can be inserted in the second line because it vanishes on spatial forms, and \eqref{eq:magic} for the last line. 
Applying $\ed$ to \eqref{eq:spat.d} yields
\begin{align*}
  0&=-\ed\bvec{n}^\flat\wedge\ldn\mu+\bvec{n}^\flat\wedge\ed\ldn\mu+\ed\sed\mu \\
   &=-\bvec{n}^\flat\wedge(-\frac{1}{N}\sed N\wedge\ldn\mu-\sed\ldn\mu+\ldn\sed\mu)+\sed\sed\mu, \\
   &=-\bvec{n}^\flat\wedge(-\sed\ldm\mu+\ldm\sed\mu)+\sed\sed\mu,
\end{align*}
which shows, in particular, that $\sed\sed=0$ and $[\sed,\ldm]=0$.

Then taking $\ed$ of \eqref{eq:decomp}, using \eqref{eq:dn} and \eqref{eq:spat.d} since $\bar{\omega}$ and $\hat{\omega}$ are spatial, we get
\begin{align}
\ed\omega&{}=-\ed\bvec{n}^\flat\wedge \bar{\omega}+\bvec{n}^\flat\wedge\ed \bar{\omega}+\ed \hat{\omega} \nonumber\\
&=\bvec{n}^\flat\wedge(\frac1N\sed N\wedge\bar{\omega}+\sed \bar{\omega}-\ld{\bvec{n}}\hat{\omega})+\sed \hat{\omega} \nonumber \\
&=-\bvec{n}^\flat\wedge \frac1N(\ld{\bvec{m}}\hat{\omega}-\sed(N\bar{\omega}))+\sed \hat{\omega}, \label{eq:full.d}
\end{align}
where in order to get the final line, we use $\sed N\wedge\bar{\omega}+N\wedge\sed\bar{\omega}=\sed(N\bar{\omega})$, and that $N\ld{\bvec{n}}=\ld{\bvec{m}}$ on spatial forms. This last property is easy to see using \eqref{eq:magic}, since both $\ipn$ and $\ipm$ vanish on spatial forms
\begin{equation*}
  N\ld{\bvec{n}}=\underbrace{N\ed\ipn}_{=0}+N\ipn\ed=\underbrace{\ed\ipm}_{=0}+\ipm\ed=\ldm.
\end{equation*}

\subsection{Hodge star}\label{sec:3+1hs}
In the following, we let $(e_\alpha)_{\alpha\in[0,3]}$ be a right-handed $g$-orthonormal basis of $TM$ such that $e_0=\bvec{n}$, and $(\theta^\alpha)_{\alpha\in[0,3]}$ the dual basis. In this basis, the frame components of the metric and its inverse are given by
\begin{equation*}
g_{\alpha\beta}:=g(e_\alpha,e_\beta)= -\delta_{\alpha}^0\delta_{\beta}^0 +\delta_\alpha^i \delta_\beta^j \delta_{ij}
\quad \text{and} \quad 
g^{\alpha\beta}:=g(\theta^\alpha,\theta^\beta)= -\delta^{\alpha}_0\delta^{\beta}_0 +\delta^\alpha_i \delta^\beta_j \delta^{ij},
\end{equation*}
respectively, and the metric can be expresses as
\begin{equation*}
g=g_{\alpha\beta}\theta^\alpha \otimes \theta^\beta.
\end{equation*}
Throughout this article, we will use without comment the frame metric components $g_{\alpha\beta}$ and its inverse to lower and raise spacetime frame indices, e.g.~$\alpha,\beta,\gamma$. Similarly, we will raise and lower spatial frame indices, e.g.~$i,j,k$, with the spatial frame metric components $h_{ij}=\delta_{ij}$.

From this choice of frame, we have $(e_i)_{i\in[1,3]}$ is a spatial basis of $\Sigma_t$, and $\theta^0=-\bvec{n}^\flat$. 
The induced volume form on $M$ has the presentation $\varepsilon_g=\theta^0\wedge\theta^1\wedge\theta^2\wedge\theta^3$, and from the definition of $h$ \eqref{eq:h}, the spatial volume form is $\varepsilon_h=\theta^1\wedge\theta^2\wedge\theta^3$. It is easy to see in this basis that $\ipn(\varepsilon_g)=\ip{e_0}(\varepsilon_g)=\varepsilon_h$ because $\ip{e_0}\theta^i=\delta^i_0$, and the calculation
\begin{equation}\label{eq:vol.form}
  \ip{e_0}(\varepsilon_g)=\ip{e_0}(\theta^0)\wedge\theta^1\wedge\theta^2\wedge\theta^3=\varepsilon_h.
\end{equation}
In orthonormal basis, this is equivalent to $\varepsilon_{0ijk}=\varepsilon_{ijk}$, and note that in index notation we omit the subscript $g$ and $h$ as it should be clear from context which is which.

Let the Hodge star operator on $(M,g)$ be $\star$, defined as in \ref{def:hodge.star}. In the same way, we let the $3$-dimensional Hodge star on the spatial slices $(\Sigma_t,h)$ be $\shs $. Then the action of the general Hodge star on any basis $k$-form can be split into two cases using $\shs $:
\begin{enumerate}
\item Basis $k$-form contains $\theta^0$, $i_1\cdots i_{k-1}$ spatial indices:
      \begin{equation*}
        \star(\theta^0\wedge\theta^{i_1}\wedge\cdots\wedge\theta^{i_{k-1}})=\frac{1}{(4-k)!}\varepsilon\indices{^{0i_1\cdots i_{k-1}}_{i_{k}\cdots i_3}}\theta^{i_k}\wedge\cdots\wedge\theta^{i_3},
      \end{equation*}
      and by \eqref{eq:vol.form}, we replace $\varepsilon\indices{^{0i_1\cdots i_{k-1}}_{i_{k}\cdots i_3}}=-\varepsilon\indices{_0^{i_1\cdots i_{k-1}}_{i_{k}\cdots i_3}}$ and $4-k=3-(k-1)$ to see the appearance of $\shs $ ($i_1\cdots i_3$ spatial indices), and thus
      \begin{align*}
        \star(\theta^0\wedge\theta^{i_1}\wedge\cdots\wedge\theta^{i_{k-1}})&=-\frac{1}{(3-(k-1))!}\varepsilon\indices{_0^{i_1\cdots i_{k-1}}_{i_{k}\cdots i_3}}\theta^{i_k}\wedge\cdots\wedge\theta^{i_3}\\
        &=-\shs (\theta^{i_1}\wedge\cdots\wedge\theta^{i_{k-1}}).
      \end{align*}
      In the special case that $k=1$, i.e. $\star\theta^0$, we see the above calculation leads to
      \begin{align*}
        \star\theta^0&=-\frac{1}{3!}\varepsilon\indices{_{0i_{1}\cdots i_3}}\theta^{i_1}\wedge\cdots\wedge\theta^{i_3}\\
        &=-\theta^1\wedge\theta^2\wedge\theta^3=-\shs (1).
      \end{align*}
\item Basis $k$-form is spatial:
      \begin{align*}
        \star(\theta^{i_1}\wedge\cdots\wedge\theta^{i_{k}})&=\frac{1}{(4-k)!}\varepsilon\indices{^{i_1\cdots i_{k}}_{i_{k+1}\cdots i_4}}\theta^{i_{k+1}}\wedge\cdots\wedge\theta^{i_4}\\
        &=\frac{4-k}{(4-k)!}\varepsilon\indices{^{i_1\cdots i_{k}}_{0i_{k+1}\cdots i_3}}\theta^{0}\wedge\theta^{i_{k+1}}\wedge\cdots\wedge\theta^{i_3}\\
        &=\frac{(-1)^k}{(3-k)!}\varepsilon\indices{_0^{i_1\cdots i_{k}}_{i_{k+1}\cdots i_3}}\theta^{0}\wedge\theta^{i_{k+1}}\wedge\cdots\wedge\theta^{i_3}\\
        &=(-1)^k\theta^0\wedge\shs (\theta^{i_{1}}\wedge\cdots\wedge\theta^{i_k}),
      \end{align*}
      where in the second line, we can fix one of $i_{k+1}\cdots i_{4}$ as $0$, since $i_{1}\cdots i_{k}$ are all spatial, and collect the sums by antisymmetry of $\varepsilon$ and $\wedge$, in the third we use only the antisymmetry of $\varepsilon$ to pick up $(-1)^k$, and finally recombine in the last line to see $\shs $.
\end{enumerate}

Applying these formulas to the decomposition of a $k$-form $\omega$ as in \eqref{eq:decomp}, recalling that $\bar{\omega}$ and $\hat{\omega}$ are spatial and $-\bvec{n}^\flat=\theta^0$ to get
\begin{equation}\label{eq:star.decomp}
  \star\omega=\star(-\bvec{n}^\flat\wedge \bar{\omega})+\star \hat{\omega}=-(-1)^k\bvec{n}^\flat\wedge\shs \hat{\omega}-\shs \bar{\omega}.
\end{equation}

\section{Einstein's equations}\label{sec:einstein}

Letting $\nabla$ denote the Levi-Civita connection of $g$, the connection coefficients $\omega^\gamma{}_{\beta\alpha}{}$
associated to the orthonormal frame $(e_\alpha)_{\alpha\in[0,3]}$ are defined by
\begin{equation*}
\nabla_{e_\alpha} e_{\beta}=\omega\indices{^\gamma_{\beta\alpha}}e_\gamma \quad \Longleftrightarrow \quad
\theta^\gamma(\nabla_{e_\alpha}e_\beta)=\omega\indices{^\gamma_{\beta\alpha}},
\end{equation*}
which in turn, we use to define the connection 1-forms $\omega\indices{^\gamma_\beta}$ via
\begin{equation*}
\omega\indices{^\gamma_\beta} =\omega\indices{^\gamma_{\beta\alpha}}\theta^\alpha.
\end{equation*} 
Since the connection $\nabla$ is metric, the torsion vanishes and the following Cartan structure equations and metric compatibility conditions hold (see \cite[Section 7.8]{Nakahara:03} for details on the Cartan frame formalism):
\begin{equation}
\ed\theta^\alpha+\omega\indices{^\alpha_\beta}\wedge\theta^\beta=0\label{eq:torsion}
\end{equation}
and
\begin{equation}\label{eq:metric.comp}
 \omega_{\alpha\beta}+\omega_{\beta\alpha}=0.
\end{equation}
The curvature $2$-form $\Omega\indices{^\alpha_\beta}$ is then defined by
\begin{equation}
  \Omega\indices{^\alpha_\beta}=\ed\omega\indices{^\alpha_\beta}+\omega\indices{^\alpha_\gamma}\wedge\omega\indices{^\gamma_\beta}.\label{eq:curvature}
\end{equation}
The standard Riemann tensor of the metric in this basis are the components of the curvature $2$-forms
\begin{align}\label{eq:Riem}
  \Omega\indices{^\alpha_\beta}
  ={}& R\indices{^\alpha_{\beta\mu\nu}}\theta^\mu\wedge\theta^\nu \\
  ={}&\Bigl(e_\mu \bigl(\omega\indices{^\alpha_{\beta\nu}}\bigr)-e_\nu\bigl(\omega\indices{^\alpha_{\beta\mu}}\bigr)+\omega\indices{^\alpha_{\gamma\mu}}\omega\indices{^\gamma_{\beta\nu}}-\omega\indices{^\alpha_{\gamma\nu}}\omega\indices{^\gamma_{\beta\mu}}
-\omega\indices{^\alpha_{\beta\gamma}}\omega\indices{^\gamma_{\nu\mu}}+\omega\indices{^\alpha_{\beta\gamma}}\omega\indices{^\gamma_{\mu\nu}}\Bigr)\theta^\mu\wedge\theta^\nu. \nonumber
\end{align}

The collections of forms in these equations can be interperated together as sections of a principal bundle that acts on the tangent bundle, in a similar way to the Yang--Mills equations. The principal difference however in Einstein's equations is the interaction of these Lie algebra indices with the indices of the differential forms.

\subsection{Differential form formulation}

For convenience we introduce the Hodge star of the wedge of 1-forms as
\begin{equation}\label{def:hsBasis}
  \Sigma^{\alpha_1\cdots\alpha_k}\coloneq \star\bigl(\theta^{\alpha_{1}}\wedge \cdots \wedge \theta^{\alpha_{k}}\bigr).
\end{equation}
Note that this is a $(n-k)$-form (although there are $k$ indices, they are not the components of some tensor). Lowering the indices using the metric, we get what are called ``hypersurface forms'' in \cite{Frauendiener:06, Olivares.Peshkov.ea:22}, i.e. the forms
\begin{equation}\label{eq:hs.forms}
\Sigma_{\alpha_1\cdots\alpha_k}:= g_{\alpha_1\beta_1}\cdots g_{\alpha_k\beta_k} \Sigma^{\beta_1\cdots\beta_k}=\frac{1}{(4-k)!}\varepsilon_{\alpha_1\cdots\alpha_k\alpha_{k+1}\cdots\alpha_4}\theta^{\alpha_{k+1}}\wedge\cdots\wedge\theta^{\alpha_4}.
\end{equation}
By \eqref{eq:double.hodge}, we have the relation
\begin{equation*}
  \star\Sigma^{\alpha_1\cdots\alpha_k}=(-1)^{k(n-k)+1}\bigl(\theta^{\alpha_{1}}\wedge \cdots \wedge \theta^{\alpha_{k}}\bigr),
\end{equation*}
or equally,
\begin{equation}\label{eq:wedge.hs}
  \theta^{\alpha_{1}}\wedge\cdots\wedge\theta^{\alpha_{k}}
  =-\frac{1}{(4-k)!}\varepsilon^{\alpha_{k+1}\cdots\alpha_4\alpha_1\cdots\alpha_k}\Sigma_{\alpha_{k+1}\cdots\alpha_4}.
\end{equation}

The standard tensor formulation of the vacuum Einstein equations is
\begin{equation}\label{eq:Einstein.tensor}
  G_{\mu\nu}\coloneq R_{\mu\nu}-\frac{1}{2}Rg_{\mu\nu}=0,
\end{equation}
where $G_{\mu\nu}$ is the Einstein tensor, $R_{\mu\nu}$ the Ricci tensor \eqref{def:ricci} and $R$ the scalar curvature \eqref{def:sc}. The equivalent differential form version of Einstein equations is obtained as follows. Define the $2$-forms $\mathcal{L}_\alpha$ as
\begin{equation}\label{def:L}
  \mathcal{L}_\alpha=-\frac{1}{2}\varepsilon_{\alpha\beta\mu\nu}\omega^{\beta\mu}\wedge\theta^\nu=-\frac{1}{2}\omega^{\beta\mu}\wedge\Sigma_{\alpha\beta\mu}.
\end{equation}
Then applying the exterior derivative, using the fact that the components of $\varepsilon$ and $g$ are constant in this basis and Cartan's structure equations \eqref{eq:torsion} and \eqref{eq:curvature}, we get the equation
\begin{align}\nonumber
  \ed \mathcal{L}_\alpha
  &=-\frac{1}{2}(\ed\omega^{\beta\mu}\wedge\Sigma_{\alpha\beta\mu}-\varepsilon_{\alpha\beta\mu\nu}\omega^{\beta\mu}\wedge\ed\theta^\nu)\\
  &=\underbrace{-\frac{1}{2}\Omega^{\beta\mu}\wedge\Sigma_{\alpha\beta\mu}}_{\mathcal{E}_\alpha}
   +\underbrace{\frac{1}{2}\omega\indices{^\beta_\gamma}\wedge\omega^{\gamma\mu}\wedge\Sigma_{\alpha\beta\mu}-\frac{1}{2}\varepsilon_{\alpha\beta\mu\nu}\omega^{\beta\mu}\wedge\omega\indices{^\nu_\gamma}\wedge\theta^\gamma}_{\mathcal{S}_\alpha}, \label{eq:sparling}
\end{align}
where we split the result into the Einstein form $\mathcal{E}_\alpha$ and the Sparling form $\mathcal{S}_\alpha$. Expanding the Einstein form by \eqref{eq:Riem} and \eqref{eq:hs.forms}, then the formulas \eqref{eq:wedge.hs} for the second line, \eqref{eq:kron.delta} for the third and the symmetries of the Riemann curvature tensor \eqref{eq:reim.sym.1} in the fourth, 
\begin{align}\nonumber
  \mathcal{E}_\alpha=-\frac{1}{2}\Omega^{\beta\mu}\wedge\Sigma_{\alpha\beta\mu}&=-\frac{1}{2}R\indices{^{\beta\mu}_{\gamma\nu}}\varepsilon_{\alpha\beta\mu\rho}\theta^\gamma\wedge\theta^\nu\wedge\theta^\rho \\\nonumber
  &=\frac{1}{2}R\indices{^{\beta\mu}_{\gamma\nu}}\varepsilon_{\alpha\beta\mu\rho}\varepsilon^{\sigma\gamma\nu\rho}\Sigma_{\sigma} \\\nonumber
  &=-3R\indices{^{\beta\mu}_{\gamma\nu}}\delta_{[\alpha}^{\sigma}\delta_{\beta}^{\gamma}\delta_{\mu]}^{\nu}\Sigma_{\sigma} \\\nonumber
  &=-(R\indices{^{\beta\mu}_{\beta\mu}}\delta^\sigma_\alpha+R\indices{^{\sigma\mu}_{\mu\alpha}}+R\indices{^{\beta\sigma}_{\alpha\beta}})\Sigma_{\sigma} \\\nonumber
  &=2G\indices{^\sigma_\alpha}\Sigma_{\sigma},\\
  \label{eq:einstein}
  &= 2 G_{\sigma\alpha}\star \theta^\sigma,
\end{align}
and we see that it contains exactly the Einstein tensor. The vacuum Einstein equations \eqref{eq:Einstein.tensor} are therefore equivalent to setting $\mathcal{E}_\alpha=0$ and imposing the differential form equation
\begin{equation}\label{eq:dL.S}
  \ed \mathcal{L}_\alpha=\mathcal{S}_\alpha.
\end{equation}
Equation \eqref{eq:dL.S} naturally implies the equation
\begin{equation}\label{eq:dS.0}
  \ed \mathcal{S}_\alpha = 0,
\end{equation}
through $\ed\ed=0$. Thus equation \eqref{eq:dS.0} is in fact equivalent to the twice contracted Bianchi identity.

\subsection{Closing equations}

What is missing from \eqref{eq:dL.S} and \eqref{eq:dS.0} are the relations of the basis $\theta^\alpha$ with the $1$-forms $\omega\indices{^{\alpha}_{\beta}}$. They are incorporated in \cite{Olivares.Peshkov.ea:22} by introducing the collection of $2$-forms $\mathcal{C}^\alpha$ so that
\begin{equation}\label{eq:dtheta.C}
  \ed\theta^\alpha=\mathcal{C}^\alpha.
\end{equation}
The $\mathcal{C}^\alpha$ contain the commutators $[e_\mu,e_\nu]=\mathcal{C}\indices{^\alpha_\nu_\mu}e_\alpha$ which determine the connection $\omega\indices{^\alpha_\beta}$; we see by \eqref{eq:torsion} that $C\indices{^\alpha_{\beta\gamma}}=2\omega\indices{^\alpha_{[\beta\gamma]}}$ because
\begin{equation}\label{eq:comm.coeff}
  \mathcal{C}^\alpha=\frac{1}{2}\mathcal{C}\indices{^\alpha_{\beta\gamma}}\theta^\beta\wedge\theta^\gamma=-\omega\indices{^\alpha_\beta}\wedge\theta^\beta=\frac{1}{2}\bigl( \omega\indices{^\alpha_{\beta\gamma}}-\omega\indices{^\alpha_{\gamma\beta}}\bigr)\theta^\beta\wedge\theta^\gamma=\omega\indices{^\alpha_{[\beta\gamma]}}\theta^\beta\wedge\theta^\gamma,
\end{equation}
then using the standard trick in proving the metric formula for Christoffel symbols, we use metric compatibility \eqref{eq:metric.comp} to write and sum the following equations
\begin{align*}
  \omega\indices{_{\alpha\beta\mu}}+\omega\indices{_{\beta\alpha\mu}}&=0, \\
  \omega\indices{_{\alpha\mu\beta}}+\omega\indices{_{\mu\alpha\beta}}&=0, \\
  -\omega_{\mu\beta\alpha}-\omega\indices{_{\beta\mu\alpha}}&=0,
\end{align*}
where again we use the metric to lower the first (non form) index of $\omega$.
Grouping terms so that the antisymmetry \eqref{eq:comm.coeff} applies, we get
\begin{equation*}
  2\omega\indices{_{\alpha\beta\mu}}
  +\mathcal{C}\indices{_{\beta\alpha\mu}}
  -\mathcal{C}\indices{_{\alpha\beta\mu}}
  +\mathcal{C}\indices{_{\mu\alpha\beta}}
  =0
\end{equation*}
so
\begin{equation}\label{eq:omega.C}
  \omega\indices{_{\alpha\beta\mu}}
  =\frac{1}{2}
  (-\mathcal{C}\indices{_{\beta\alpha\mu}}
  +\mathcal{C}\indices{_{\alpha\beta\mu}}
  -\mathcal{C}\indices{_{\mu\alpha\beta}}).
\end{equation}
Hence the equation \eqref{eq:dtheta.C} and its exterior derivative
\begin{equation}\label{eq:dC.0}
  \ed \mathcal{C}^\alpha=0
\end{equation}
are another pair of differential form equations, this time capturing information on the metric $g$, which $\theta^\alpha$ and $\omega\indices{^\alpha_\beta}$ depend on.

\subsection{$3+1$ Einstein's equations}\label{sec:equations}
The spacetime equations
\begin{align*}
  \ed \mathcal{L}_\alpha&=\mathcal{S}_\alpha,       &  \ed\theta^\alpha&=\mathcal{C}^\alpha, \\
  \ed \mathcal{S}_\alpha&=0,              &  \ed \mathcal{C}^\alpha&=0,
\end{align*}
can be decomposed, using the theory of Section \ref{sec:ec.decomp}, by first decomposing $\mathcal{L}_\alpha$, $\mathcal{S}_\alpha$, $\theta^\alpha$ and $\mathcal{C}^\alpha$ into their normal and tangential components, as follows\\
\begin{subequations}\label{eq:LSTC}
\begin{minipage}{0.40\textwidth}
\begin{align}
  \mathcal{L}_\alpha&=-\bvec{n}^\flat\wedge H_\alpha+D_\alpha,   \label{eq:L}\\
  \mathcal{S}_\alpha&=-\bvec{n}^\flat\wedge U_\alpha+V_\alpha,  \label{eq:S}
\end{align}
\end{minipage}
\begin{minipage}{0.45\textwidth}
\begin{align}
  \theta^0&=-\bvec{n}^\flat\wedge 1,\quad \theta^i=-\bvec{n}^\flat\wedge 0+\theta^i, \\
  \mathcal{C}^\alpha&=-\bvec{n}^\flat\wedge E^\alpha+B^\alpha, \label{eq:C}
\end{align}
\end{minipage}
\end{subequations}\\[0.3cm]
recalling that $\theta^\alpha$ already splits nicely since $\theta^0=-\bvec{n}^\flat$ in the orthonormal basis.
Then the $3+1$ evolution and constraint equations are
\begin{subequations}\label{eq:3+1}
\begin{align}\label{eq:3+1.1}
  \ldm& D_\alpha-\sed(NH_\alpha)=NU_\alpha,      & \sed D_\alpha&=V_\alpha, \\
  \label{eq:3+1.2}
  \ldm &V_\alpha-\sed(NU_\alpha)=0,              & \sed V_\alpha&= 0,  \\
  \ldm &\theta^i=NE^i,                           & \sed \theta^i&= B^i,\quad \sed N=-NE^0,\quad B^0=0, \label{eq:3+1.3}\\
  \ldm &B^\alpha-\sed(NE^\alpha)=0,              & \sed B^\alpha&=0.
  \label{eq:3+1.4}
\end{align}
\end{subequations}

\begin{remark}\label{rm:orth.frame}
The geometric meaning of these variables in \eqref{eq:LSTC} can be understood by relating them to those from the $1+3$ orthonormal frame approach \cite{Van-Elst.Uggla:97}. In particular, the decomposition of the commutator gives a straightforward comparison. We have from \cite{Van-Elst.Uggla:97} the equations (note that \cite{Van-Elst.Uggla:97} uses Greek letters for spatial frame indices, we have replaced them here with Latin letters to match our convention)
\begin{align*}
	[e_0,e_i]&=\dot{u}_{i}e_0 - \big[\frac{1}{3}\Theta\delta\indices{^j_i}+\sigma\indices{^j_i}+\epsilon\indices{^j_i_k}(\omega^k-\Omega^k)\big]e_j, \\
	[e_i,e_j]&=-2\epsilon\indices{_i_j_k}\omega^ke_0 - \big[2a_{[i}\delta\indices{^k_j_]} + \epsilon\indices{_i_j_l}n\indices{^l^k}\big]e_k,
\end{align*}
where $\dot{u}^{i}$ is the acceleration vector, $\Theta$ is the expansion scalar, $\sigma_{ij}$ is the shear tensor, $\omega^i$ is the vorticity vector, $\Omega^i$ is the Fermi-rotation of the frame, and the pair $\{n^{lk},a_i\}$ parametrize the commutator coefficients $\theta^{k}([e_i,e_j])$. Comparing with our decomposition of the commutator coefficients in $C^\alpha$ \eqref{eq:C}, it is clear that $E^0$ is the $4$-acceleration, $B^0$ captures the vorticity of $\bvec{n}$ (which vanishes since $e_0=\bvec{n}$ is hypersurface orthogonal \cite{Wald:84}), and $H^0$ is related to the Fermi-rotation. The remaining non-gauge variables come from the decompositions of $E\indices{^i}$ and $B\indices{^i}$ into the trace/trace-free and symmetric/antisymmetric components (this symmetry involves the non-differential form index), which is remarkably similar to the relations $\eqref{eq:rel.Bij}$ and $\eqref{eq:rel.Ek}$. One relation of particular interest is the extrinsic curvature $K_{ij}$, which turns out to be the symmetric part of $E$; indeed starting from the usual formula $\ldm h = -2NK$ ($K$ is the extrinsic curvature tensor in the following),
\begin{align*}
	K&=-\frac{1}{2N}\ldm h=-\frac{1}{2N}\ldm(\delta_{ij}\theta^i\otimes\theta^j)=-\frac{1}{2N}\delta_{ij}(\ldm\theta^i\otimes\theta^j+\theta^i\otimes\ldm\theta^j)\\
	&=-\frac{1}{2}\delta_{ij}(E^i\otimes\theta^j+\theta^i\otimes E^j)=-\frac{1}{2}(E_i\otimes\theta^i+\theta^j\otimes E_j)=-E_{(ij)}\theta^i\otimes\theta^j,
\end{align*}
where the link to $E$ comes from \eqref{eq:3+1.3}.
\end{remark}

\subsection{Non differential relations}\label{sec:non.diff.rel}
The final piece of the puzzle is the relations between the spatial field in \eqref{eq:LSTC}, since equations \eqref{eq:3+1} only tell the evolution of the tangential fields 
\begin{equation*}
D_{\alpha}=\frac12 D_{\alpha ij}\theta^i\wedge\theta^j,\quad
B_\alpha = \frac12 B_\alpha{}_{ij}\theta^i\wedge\theta^j, \quad \theta^i
\end{equation*} 
as well as relations for the pieces 
\begin{equation*}
V_{\alpha}=\frac{1}{3!} V_{\alpha ijk} \theta^i\wedge \theta^j\wedge\theta^k, \quad
U_{\alpha}=\frac12 U_{\alpha ij}\theta^i \wedge \theta^j,
\end{equation*}
of the Sparling form \eqref{eq:S}, which are nonlinear in the tangential variables. The detailed derivation can be found in Appendix \ref{sec:nondiff.relations} and \ref{sec:nonlinear.relations}, so we list here the concise forms that are obtained with the help of $\shs$:
\begin{subequations}\label{eq:rel}
\begin{align}
	D_{0ij}&=2\shs B_{[ij]}, \label{eq:rel.D0}\\
	B_{0ij}&=-2\shs D_{[ij]},\label{eq:rel.B0}\\
	\shs D_{kl} &=-E_{(kl)}+\delta_{kl}E\indices{_i^i}-\frac{1}{2}B_{0kl}, \label{eq:rel.Di}\\
	\shs B_{kl}&=-\big(-H_{(kl)}+\delta_{kl}H\indices{_i^i}-\frac{1}{2}D_{0kl}\big), \label{eq:rel.Bij} \\
	\shs  H_{0ij} &=E_{[ij]}+\frac12 B_{0ij}, \label{eq:rel.H0}\\
	\shs E_{0ij}&= -\big(H_{[ij]} + \frac12 D_{0ij}\big), \label{eq:rel.E0}\\
	H_{kl} &=\shs B_{lk}-\frac{1}{2}\delta_{kl}\shs B\indices{_i^i}-\shs E_{0kl}, \label{eq:rel.Hij}\\
	E_{kl} &=-\big(\shs D_{lk}-\frac{1}{2}\delta_{kl}\shs D\indices{_i^i}-\shs H_{0kl} \big), \label{eq:rel.Ek}
\end{align}
\end{subequations}
\begin{gather*}
H_\alpha = H_{\alpha i}\theta^i, \quad
\shs H_{\alpha}=\frac12\shs H_{\alpha ij}\theta^i\wedge \theta^j,\quad
\shs D_{\alpha}=\shs D_{\alpha i}\theta^i,\\
\shs B_{\alpha}= \shs B_{\alpha i}\theta^i,\quad
E_{\alpha}=E_{\alpha i}\theta^i, \quad
\shs E_{\alpha}=\frac12\shs E_{\alpha ij}\theta^i\wedge \theta^j.
\end{gather*}
The nonlinear relations are:
\begin{subequations}
\begin{align}
    \shs  V_0
    ={}&\frac12(E\indices{^\beta_i}\shs D\indices{_\beta^i}-H\indices{^\beta_i}\shs B\indices{_\beta^i}), \\
  \shs  V_k
    ={}&-B\indices{^\beta_i_k}\shs D\indices{_\beta^i}, \\
  \shs U\indices{_0^j}
    ={}&E\indices{^\beta_i}\shs H\indices{_\beta^i^j},\\
  \shs  U\indices{_i^j}
    ={}&\delta^j_i\shs V_0 -E\indices{^\beta_i}\shs D\indices{_\beta^j}+H\indices{^\beta_i}\shs B\indices{_\beta^j}, \label{eq:rel.Ui}
\end{align}
\end{subequations}
where
\begin{gather*}
\shs U_{\alpha}=\shs U_{\alpha i}\theta^i.
\end{gather*}
\subsection{Two formulations}\label{sec:formulation}

For simplicity we choose $H^0=0$, corresponding to a Fermi-Walker transported frame (see Remark \ref{rm:orth.frame}), and lapse and shift functions $N\in\kform[0](M)$ and $\bvec{\beta}\in\mathfrak{X}(M)$. Recall that $\theta^0=-\bvec{n}^\flat$ by our choice of tetrad, which fixes $E^0$, $B^0$ to be
\begin{equation}\label{eq:E0B0}
E^0=-\frac{1}{N}\sed N,\qquad B^0=0. 
\end{equation}

A two-field formulation of Einstein's equations that is first-order in time and second-order in space is: given initial data $(D^i(0), \theta^i(0))_{i=1}^3\in\tkform[2](M)\times\tkform[1](M)$, find $(D^i, \theta^i)_{i=1}^3\in\tkform[2](M)\times\tkform[1](M)$ such that
\begin{subequations}\label{eq:form.ev1}
\begin{align}
  \ldm D^i-{}&\sed(NH^i)=NU^i,   \label{eq:form.ev1.1}\\
  \ldm &\theta^i=NE^i,        \label{eq:form.ev1.2}
\end{align}
\end{subequations}
where $H^i$, $E^i$, $U^i$ are defined as functions of $D^i$, $\theta^i$ and $B^i = \ed\theta^i$, by the relations
\begin{align*}
	H_{kl} &=\shs B_{lk}-\frac{1}{2}\delta_{kl}\shs B\indices{_i^i}-\shs E_{0kl}, \\
	E_{kl} &=-\big(\shs D_{lk}-\frac{1}{2}\delta_{kl}\shs D\indices{_i^i}-\shs H_{0kl} \big), \\
	\shs  U\indices{_i^j}&=\frac12\delta^j_i(E\indices{^\beta_i}\shs D\indices{_\beta^i}-H\indices{^\beta_i}\shs B\indices{_\beta^i}) -E\indices{^\beta_i}\shs D\indices{_\beta^j}+H\indices{^\beta_i}\shs B\indices{_\beta^j},
\end{align*}
from Section \ref{sec:equations}. Another possibility is to include the evolution of $B$ in the system to get a three-field formulation that is first-order in space and time: given initial data $(D^i(0), \theta^i(0), B^i(0))_{i=1}^3\in\tkform[2](M)\times\tkform[1](M)\times\tkform[2](M)$, find $(D^i, \theta^i, B^i)_{i=1}^3\in\tkform[2](M)\times\tkform[1](M)\times\tkform[2](M)$ such that
\begin{subequations}\label{eq:form.ev2}
\begin{align}
  \ldm D^i-{}&\sed(NH^i)=NU^i,   \label{eq:form.ev2.1}\\
  \ldm &\theta^i=NE^i,        \label{eq:form.ev2.2} \\
  \ldm B^i-{}&\sed(NE^i)=0.        \label{eq:form.ev2.3}
\end{align}
\end{subequations}
Although \eqref{eq:form.ev2.3} is just a consequence of taking the exterior derivative of \eqref{eq:form.ev2.2}, these redundant equations are often necessary to prove some property of the problem, and it is therefore worth investigating the impact on the discretisation. The analysis of the hyperbolicity of these formulations and the proof of their well-posedness will be addressed in an upcoming work.

In both cases, the Einstein constraints
\begin{subequations}\label{eq:form.cst}
\begin{align}
\sed D_\alpha&=V_\alpha, \label{eq:form.cst.1}\\
\sed \theta^i&= B^i,     \label{eq:form.cst.2}\\
\sed B^\alpha&=0,        \label{eq:form.cst.3}
\end{align}
\end{subequations}
are propagated if they are satisfied initially.
\begin{theorem}\label{thm:einstein1}
If $(D^i, \theta^i)\in C^1\tkform[2](M)\times C^2\tkform[1](M)$ solves \eqref{eq:form.ev1} where the initial conditions $D^i(0)$, $\theta^i(0)$ satisfy the Einstein constraints \eqref{eq:form.cst}, then they generate a consistent solution to Einstein's equations.
\end{theorem}
\begin{proof}
Let $D^i$, $\theta^i$ follow the conditions of the theorem where $B^i\coloneq\sed\theta^i$. Then the auxiliary equation \eqref{eq:dtheta.C} is satsified since the normal and tangential parts \eqref{eq:form.ev1.2}, \eqref{eq:form.cst.2} hold, where the $\alpha=0$ equation follows from the gauge \eqref{eq:E0B0}.

Define $D^0$, $B^0$, $H^k$  and $E^k$ from $D^i$, $B^i$ by the relations \eqref{eq:rel.D0}, \eqref{eq:rel.B0}, \eqref{eq:rel.Hij}, \eqref{eq:rel.Ek} and gauge conditions \eqref{eq:E0B0}. Since the set of relations \eqref{eq:rel} are self consistent, we can construct the connection 1-forms $\omega\indices{^\alpha_\beta}$ through either \eqref{eq:omegaEB} or \eqref{eq:omegaDH}. This constructed $\omega\indices{^\alpha_\beta}$ is indeed the unique connection $1$-form that is torsion-free \eqref{eq:torsion} and metric compatible \eqref{eq:metric.comp}. The torsion freeness is true due to \eqref{eq:dtheta.C}, and that \eqref{eq:comm.coeff} is true due to definition \eqref{eq:omegaEB} of $\omega\indices{^\alpha_\beta}$ being true. The metric compatibility is easy to see straight from either definition; both \eqref{eq:omegaEB} and \eqref{eq:omegaDH} are antisymmetric in $\mu$ and $\nu$.

  We now prove the propogation of the Hamiltonian constraint \eqref{eq:form.cst.1}. From the identity \eqref{eq:sparling}, the fulfillment of the temporal (evolution) part of Einstein's equation \eqref{eq:form.ev1.1} is equivalent to the vanishing of the temporal part of the Einstein form $\bar{\mathcal{E}_i}=0$. By \eqref{eq:einstein} and the decompositions of Section \ref{sec:3+1hs}, the Einstein form splits as
  \begin{align*}
    \mathcal{E}_\alpha
    &= 2G_{0\alpha}\star \theta^0 + 2G_{i\alpha}\star \theta^i \\
    &=-2G_{0\alpha}\theta^1\wedge\theta^2\wedge\theta^3 - 2G_{i\alpha}\theta^0\wedge\shs \theta^i \\
    &=-\bvec{n}^\flat\wedge(\underbrace{-2G_{i\alpha}\shs \theta^i}_{\bar{\mathcal{E}}_\alpha})\underbrace{-2G_{0\alpha}\theta^1\wedge\theta^2\wedge\theta^3}_{\hat{\mathcal{E}}_\alpha}
  \end{align*}
  where we use $\theta^0=-\bvec{n}^\flat$ for the last line. Thus \eqref{eq:form.ev1.1} is the same as
  \begin{equation*}
    \bar{\mathcal{E}_i}=-2G_{ji}\shs \theta^j=0,
  \end{equation*}
  which by the injectivity of $\shs $ and symmetry of $G_{\alpha\beta}$, implies that $G_{ij}=G_{ji}=0$.

  The propogation of the remaining constraint is a result of the twice contracted Bianchi identity
  \begin{equation*}
    \nabla^{\alpha}G_{\alpha\beta}=0,
  \end{equation*}
  that expands in this basis as
  \begin{equation}\label{eq:exp.bianchi}
    g^{\alpha\gamma}(e_\gamma(G_{\alpha\beta})-\omega\indices{^\lambda_{\alpha\gamma}}G_{\lambda\beta}-\omega\indices{^\lambda_{\beta\gamma}}G_{\alpha\lambda}).
  \end{equation}
  With $\beta = i$ a spatial index, we can simplify \eqref{eq:exp.bianchi} to
  \begin{equation}\label{eq:bianchi.spat}
    -e_0(G_{0 i})- \omega\indices{^{0\gamma}_{\gamma}}G_{0 i}+\omega\indices{^\lambda_{i 0}}G_{0\lambda}-g^{jk}\omega\indices{^0_{i}^j}G_{j0},
  \end{equation}
  then $\beta = 0$,
  \begin{equation}\label{eq:bianchi.0}
    g^{\alpha\gamma}e_\gamma(G_{\alpha 0})-\omega\indices{^{\lambda\gamma}_{\gamma}}G_{\lambda 0}-\omega\indices{^\lambda_{0}^0}G_{0\lambda}.
  \end{equation}
  Setting $G_{0 0}=0$ and $G_{i 0}=G_{0i}=0$ on the initial slice, \eqref{eq:bianchi.spat} implies that
  \begin{equation}\label{eq:ev.G0i}
    e_0(G_{0 i})=0,
  \end{equation}
  so $G_{i 0}=0$ is preserved;  and \eqref{eq:bianchi.0} implies that 
  \begin{equation}\label{eq:ev.G00}
    e_0(G_{00})-g^{ij}e_j(G_{i 0})=0,
  \end{equation}
  which is a transport equation with unique solution $G_{00}=0$. 
\end{proof}
\begin{theorem}\label{thm:einstein2}
If $(D^i, \theta^i, B^i)\in C^1\tkform[2](M)\times C^2\tkform[1](M) \times C^1\tkform[1](M)$ solves \eqref{eq:form.ev2} where the initial conditions $D^i(0)$, $\theta^i(0)$, $B^i(0)$ satisfy the Einstein constraints \eqref{eq:form.cst}, then they generate a consistent solution to Einstein's equations.
\end{theorem}
\begin{proof}
	The proof is identical to that of Theorem \ref{thm:einstein1} if we can show that \eqref{eq:form.cst.2} (and as consequence \eqref{eq:form.cst.3}) is preserved by these equations. 
	
	Let $D^i, \theta^i, B^i$ be as in the theorem. Taking the exterior derivative of \eqref{eq:form.ev2.2} and subtracting from that \eqref{eq:form.ev2.3}, we get $\ldm(\sed\theta^i-B^i)=0$. The existence and uniquess of a solution to this equation comes from the fact that it is a transport equation on $\sed\theta^i-B^i$. In adapted coordinates, for a generic $3$-form $T$, $\ldm T$ has components, 
\begin{equation*}
  \ldm T_{\alpha_1\alpha_2\alpha_3}=\frac{\partial}{\partial t} T_{\alpha_1\alpha_2\alpha_3}-m^\beta\frac{\partial}{\partial x^\beta}T_{\alpha_1\alpha_2\alpha_3}-T_{\beta\alpha_2\alpha_3}\frac{\partial m^\beta}{\partial x^{\alpha_1}}-T_{\alpha_1\beta\alpha_3}\frac{\partial m^\beta}{\partial x^{\alpha_2}}-T_{\alpha_1\alpha_2\beta}\frac{\partial m^\beta}{\partial x^{\alpha_3}},
\end{equation*}
which, given initial conditions, can be solved using e.g. the method of characteristics. Since \eqref{eq:form.cst.2} is true initially, it is propogated in time. The remainder follows the proof of Theorem \ref{thm:einstein1}. 
\end{proof}

\section{The exterior calculus discrete de Rham complex}\label{sec:ecddr}

The exterior calculus discrete de Rham (ECDDR) complex is a discretisation of the chain complex of differential forms
\begin{equation*}
  \begin{tikzcd}
    \{0\}\arrow{r}{}
    & \kform[0](\Omega)\arrow{r}{\ed^0}
    & \kform[1](\Omega) \arrow{r}{\ed^1}
    & \kform[2](\Omega) \arrow{r}{\ed^2}
    & \kform[3](\Omega) \arrow{r}{} & \{0\}
  \end{tikzcd}
\end{equation*}
reproducing certain important geometric properties at the discrete level. The foundation on which the ECDDR complex is constructed is the Stokes' formula, which, on a smooth $n$-dimensional manifold $\Omega$, reads
\begin{equation*}
  \int_\Omega \ed\omega\wedge\mu = (-1)^{(k+1)}\int_{\Omega}\omega\wedge\ed\mu + \int_{\partial\Omega}\tr_{\partial\Omega}\;\omega\wedge\tr_{\partial\Omega}\;\mu \qquad\forall(\omega,\mu)\in\kform(\overline{\Omega})\times\kform[n-k-1](\overline{\Omega}),
\end{equation*}
where $\tr_{\partial\Omega}$ is the trace map (the pullback of the inclusion of the boundary). From this formula, and proper choices of ``trimmed'' polynomial spaces, one can construct a completely discrete sequence of spaces and operators
\begin{equation*}
  \begin{tikzcd}
    \{0\}\arrow{r}{}
    & \Xec{0}{r,h}\arrow{r}{\ued{0}{r,h}}
    & \Xec{1}{r,h}\arrow{r}{\ued{1}{r,h}}
    & \Xec{2}{r,h}\arrow{r}{\ued{2}{r,h}}
    & \Xec{3}{r,h}\arrow{r}{} & \{0\}
  \end{tikzcd}
\end{equation*}
that replicates the geometric identity $\ed\ed=0$, and the exactness when $\Omega$ is contractible. We only introduce here the important notions, and the motivation behind certain constructions. For full details of the ECDDR complex and its properties, see \cite{Bonaldi.Di-Pietro.ea:25}.

\subsection{Notations}

The mesh $\Mh$ of a polytopal domain $\Omega\in\Real^3$ is a partition of $\Omega$ into polytopes of dimension $d\in[0,3]$. The set of polytopes of dimension $d$ is denoted $\Delta_d(\Mh)$; intuitively in $\Real^3$, these are the collections of vertices $\Vh\coloneq\Delta_0(\Mh)$, edges $\Eh\coloneq\Delta_1(\Mh)$, faces $\Fh\coloneq\Delta_2(\Mh)$, and elements $\Th\coloneq\Delta_3(\Mh)$. The $h$ denotes the diameter of the mesh, defined by the diameter of the largest polytope in the mesh, and decreases as the mesh is refined.

The definition of full/trimmed polynomial spaces is identic to that of FEEC \cite{Arnold.Falk.ea:10}. We denote the space of $r$-th order polynomial $k$-forms on $\Real^3$ by $\Pkdf{r}{k}{\Real^3}$. These spaces contain the $k$-forms with components (with respect to the canonical basis on $\Real^3$) that are polynomials of degree at most $r$. For any mesh entity $f\in\Delta_d(\Mh)$, we have the local polynomial space $\Pkdf{r}{k}{f} = \tr_f\Pkdf{r}{k}{\Real^3}$. The trimmed polynomial subspaces $\tPkdf{r}{k}{f}\subseteq\Pkdf{r}{k}{f}$ are given by
\begin{align*}
	\tPkdf{r}{0}{f} &= \Pkdf{r}{0}{f} & k=0 \\
	\tPkdf{r}{k}{f} &= \ed\Pkdf{r}{k-1}{f}\oplus\ip{(\bvec{x}-\bvec{x}_f)}\Pkdf{r-1}{k+1}{f} & k \geq 1
\end{align*}
where $\bvec{x}_f$ is a fixed point inside $f$. The motive behind using trimmed subspaces is to increase the efficiency of the method while ensuring that discrete operators remain well-defined.

\subsection{Discrete spaces and operators}

The discrete space of $k$-forms is a Cartesian product of trimmed polynomials on mesh entities of various dimensions. In $3$ dimensions, there are the following discrete spaces replacing $\kform(\Omega)$, for $k\in[0,3]$:
\begin{align*}
	\Xec{0}{r,h} &= \Big(\bigtimes_{V\in\Vh}\Pkdf{r}{0}{V}\Big)\times\Big(\bigtimes_{E\in\Eh}\tPkdf{r}{1}{E}\Big)\times\Big(\bigtimes_{F\in\Fh}\tPkdf{r}{2}{F}\Big)\times\Big(\bigtimes_{T\in\Th}\tPkdf{r}{3}{T}\Big) \\
	\Xec{1}{r,h} &= \Big(\bigtimes_{E\in\Eh}\Pkdf{r}{0}{E}\Big)\times\Big(\bigtimes_{F\in\Fh}\tPkdf{r}{1}{F}\Big)\times\Big(\bigtimes_{T\in\Th}\tPkdf{r}{2}{T}\Big) \\
	\Xec{2}{r,h} &= \Big(\bigtimes_{F\in\Fh}\Pkdf{r}{0}{F}\Big)\times\Big(\bigtimes_{T\in\Th}\tPkdf{r}{1}{T}\Big) \\
	\Xec{3}{r,h} &= \bigtimes_{T\in\Th}\Pkdf{r}{0}{T}
\end{align*} 
which can be conveniently summarised as
\begin{equation}\label{def:Xec}
	\Xec{k}{r,h}\coloneq\bigtimes_{d=k}^3\bigtimes_{f\in\Delta_d(\Mh)}\tPkdf{r}{d-k}{f}.
\end{equation}
A discrete element is written as $\ulh{\omega}\in\Xec{k}{r,h}$, and its restriction to a particular cell $f\in\Delta_d(\Mh)$ is denoted by $\ulf{\omega}\in\Xec{k}{r,f}$, containing only the subset of polynomials associated to $f$ and its boundary $\partial f$.

These spaces are built so that we can define on any mesh element, when the degree of the discrete form is consistent with the dimension, two full polynomial constructions called the (local) potential reconstruction $\Pec{k}{r,f}:\Xec{k}{r,f}\to\Pkdf{r}{k}{f}$, and the (local) discrete exterior derivative $\ded{k}{r,f}:\Xec{k}{r,f}\to\Pkdf{r}{k+1}{f}$. The role of the potential reconstruction is to reconstruct, from the trimmed components in the discrete space, a full polynomial that plays the role of an approximation of the continuous field of interest. The definitions of these maps are connected, and form a hierarchical structure: for any discrete $k$-form $\ulh{\omega}\in\Xec{k}{r,h}$, and $d$-dimensional polytope $f\in\Delta_d(\Mh)$,
\begin{itemize}
	\item if $d=k$, then $\Pec{k}{r,f}\ulf{\omega}=\omega_f$; the reconstruction  is just the polynomial component in $\ulf{\omega}$ attached to $f$ (see that when $d=k$ in \eqref{def:Xec}, we have the space $\tPkdf{r}{0}{f}=\Pkdf{r}{0}{f}$, and thus $\omega_f$ is a full polynomial). The discrete exterior derivative of a discrete $k$-form does not exist, since $d<k+1$.
	\item if $d>k$, then we first define the discrete exterior derivative by mimicking the Stokes' formula: by Riesz representation theorem, there exists $\ded{k}{r,f}\ulf{\omega}\in\Pkdf{r}{k+1}{f}$ such that
\begin{equation}\label{eq:ded.stokes}
	\int_{f}\ded{k}{r,f}\ulf{\omega}\wedge\mu = (-1)^{k+1}\int_f\hs^{-1}_f\omega_f\wedge\ed\mu + \int_{\partial f} \Pec{k}{r,\partial f}\ul[\partial f]{\omega}\wedge\tr_{\partial f}\mu\qquad \forall\mu\in\Pkdf{r}{d-k-1}{f},
\end{equation}
where $\omega_f\in\tPkdf{r}{d-k}{f}$ is the trimmed polynomial component of $\ulf{\omega}$ that is attached to $f$, and $\Pec{k}{r,\partial f}$ is the potential reconstruction on the boundary of $f$ (e.g. on the edges of a face). In particular, for the smallest case $d=k+1$, the $\Pec{k}{r,\partial f}$ exists due to the first definition.
\item if $d>k$ and $\ded{k}{r,f}$ is defined, then the potential $\Pec{k}{r,f}$ is defined by again mimicking the Stokes' theorem as in \eqref{eq:ded.stokes}, except with the roles of $\int_f\ed\omega\wedge\mu$ and $\int_f\omega\wedge\ed\mu$ swapped (since now we want a linear form on $\omega$ instead of $\ed\omega$). The $\ed\omega$, now on the right-hand side, is approximated by the discrete exterior derivative $\ded{k}{r, h}\ulf{\omega}$ on $f$.
\end{itemize}
We see that starting from the base case, the existence of the potential reconstruction on all $k$-polytopes means that the discrete exterior derivative can be defined for all $(k+1)$-polytopes, which in turn, provides the information to define the potential on all $(k+1)$-polytopes, that can then be used to define the discrete exterior derivative on all $(k+2)$-polytopes and so on, until the highest dimension is reached. The global discrete exterior derivative $\ued{k}{r,f}:\Xec{k}{r,h}\to\Xec{k+1}{r,h}$ of the ECDDR complex is defined by collecting the local discrete exterior derivatives and projecting the result (in the full polynomial space) back onto the trimmed polynomial spaces.

The other application of the potential reconstruction is to construct discrete $L^2$-products $(\cdot,\cdot)_{k,h}$ on each ECDDR space. They are given by
\begin{equation}\label{def:L2p}
	(\underline{\omega}_h, \underline{\mu}_h)_{k,h} \coloneq \sum_{f\in\Delta_n(\Mh)}\int_f\Pec{k}{r,h}\underline{\omega}_h\wedge\chs\Pec{k}{r,h}\underline{\mu}_h + \rho s_{k,h}(\underline{\omega}_h, \underline{\mu}_h)\qquad \forall\underline{\omega}_h,\underline{\mu}_h\in\Xec{k}{r,h},
\end{equation}
where $\chs$ is the canonical Hodge star in $\Real^3$ (see \eqref{def:chs}), $\rho>0$ is the stabilisation parameter, and $s_{k,h}: \Xec{k}{r,h}\times\Xec{k}{r,h}\to\Real$ are suitable stabilisation bilinear forms. For fully discrete methods such as ECDDR, these stabilisation forms are necessary to guarantee the positive definiteness of the $L^2$-product. This is linked to the fact that the discrete spaces $\Xec{k}{r,h}$ have no compatibility condition between the polynomial form on $f$, and the forms associated to its boundary.

\section{Discretisation} \label{sec:disc}

\subsection{Weak formulation}
The weak formulation of \eqref{eq:form.ev1} and \eqref{eq:form.ev2} is obtained by the standard method of wedging both sides by a test form and applying integration by parts where needed. To be able to discretise using the ECDDR complex, we perform the integration with respect to the induced metric of a chart basis $(\ed x^I)_{I=0}^3$. That is, we define the inner product as in \ref{def:hodge.star} with the constant spatial Hodge star $\chs$
\begin{equation}\label{def:chs}
  (v,w)_\Omega \coloneq\int_{\Omega} v\wedge \chs  w=\int_{\Omega} v_{I_1\cdots I_k}w^{I^1\cdots I^k}\ed x^1\wedge\cdots\wedge \ed x^k \qquad \forall v,w\in\kform(\Omega).
\end{equation}
The $\star_c$ does not commute with $\ldm$ due to the spatial component of the derivative. To simplify, we set the shift to $\bvec{\beta}=0$ so that we can write the weak form of the two-field system \eqref{eq:form.ev1} as
\begin{subequations}\label{eq:weak.ev1}
\begin{alignat}{2}
  \Big(\ddt \chs D^i, v\Big)_{\Omega}-\big(N\chs H^i,\sed v\big)_{\Omega}&=\big(N\chs U^i, v\big)_{\Omega}   &\quad\forall v\in\kform[1](\Omega), \label{eq:weak.ev1.1}\\
  \Big(\ddt \theta^i, v\Big)_{\Omega}&=\big(N E^i, v\big)_{\Omega} &\quad \forall v\in\kform[1](\Omega).      \label{eq:weak.ev1.2}
\end{alignat}
\end{subequations}
and the three-field system \eqref{eq:form.ev2} as
\begin{subequations}\label{eq:weak.ev2}
\begin{alignat}{2}
  \Big(\ddt \chs D^i, v\Big)_{\Omega}-\big(N\chs H^i,\sed v\big)_{\Omega}&=\big(N\chs U^i, v\big)_{\Omega}   &\quad\forall v\in\kform[1](\Omega), \label{eq:weak.ev2.1}\\
  \Big(\ddt \chs\theta^i, w\Big)_{\Omega}&=\big(N \chs E^i, w\big)_{\Omega} &\quad \forall w\in\kform[2](\Omega),      \label{eq:weak.ev2.2}\\
  \Big(\ddt \chs B^i, v\Big)_{\Omega}-\big(N\chs E^i,\sed v\big)_{\Omega}&=0   &\quad\forall v\in\kform[1](\Omega). \label{eq:weak.ev2.3}
\end{alignat}
\end{subequations}
%and the weak form of the constraints \eqref{eq:form.cst}
%\begin{subequations}\label{eq:weak.cst}
%\begin{align}
%(\chs D_\alpha,\sed\chs q)&=(\chs V_\alpha,\chs q) &\forall q\in\kform[3](\Omega) \label{eq:weak.cst.1}\\
%(\chs \theta^i,\sed \chs v)&=(\chs B^i,\chs v)     & \forall v\in\kform[2](\Omega)\label{eq:weak.cst.2}\\
%(\chs B^\alpha,\sed\chs p)&=0 &\forall p\in\kform[3](\Omega)       \label{eq:weak.cst.3}
%\end{align}
%\end{subequations}

\subsection{Numerical schemes}
Discretise the temporal domain $[0,T]$ by a sequence of strictly increasing values $(t^n)_{n=0}^{N_T}\subset[0,T]$ where $N\in\mathbb{N}$, $t^0=0$ and $t^{N_T}=T$. Define the timestep as $\delta t^{n+\frac{1}{2}}\coloneq t^{n+1}-t^{n}$, and the discrete time derivative $\delta_t^{n+1}$ applied to a sequence of values $(\bvec{v}(n))_n$ as
\begin{equation*}
  \delta_t^{n+1} \bvec{v}\coloneq \frac{\bvec{v}(n+1)-\bvec{v}(n)}{\delta t^{n+\frac12}}.
\end{equation*}

The discretisations of the weak formulations \eqref{eq:weak.ev1} and \eqref{eq:weak.ev2} do not necessarily use the forms $D^i$, $\theta^i$, $B^i$ directly as unknowns, but instead the (constant) Hodge-star versions $\chs D^i$, $\chs\theta^i$, $\chs B^i$ where appropriate. This choice is delibrate (see Remark \ref{rm:choice.DtB} for details), and of course, once an approximation is found, we can recover easily the approximations of the original forms by taking the inverse Hodge star. We denote in the following the discrete forms using the notation $\underline{Z}_h$ for $Z=\chs D^i,\theta^i,\ldots$ the continuous counterparts, so that it is clear exactly which form we are dealing with.

The discretisation of the two-field equations \eqref{eq:weak.ev1} using semi-implicit time stepping is based on the primary unknowns $(\chs D^i,\theta^i)$ and reads: With lapse $N$ and initial conditions given by $(\underline{\chs D}_h^i(0), \underline{\theta}_h^i(0))_{i\in[1,3]}\in(\Xec{1}{r,h}\times\Xec{1}{r,h})^3$, find for every $n$ a collection of forms $(\underline{\chs D}_h^i(n), \underline{\theta}_h^i(n))_{i\in[1,3]}\in(\Xec{1}{r,h}\times\Xec{1}{r,h})^3$ such that, for all $i\in [1,3]$,
\begin{subequations}\label{eq:scheme1}
\begin{align}\nonumber
  \Big(\delta_t^{n+1}  \underline{\chs D}_h^i,\underline{v}_h\Big)_{1, h}-\Big(N(n)&\underline{\chs H}_h^i\big( \ued{1}{r,h}\underline{\theta}_h(n+1), \underline{\theta}_h(n)\big), \ued{1}{r,h}\underline{v}_h\Big)_{2,h}\\
  &=\Big(N(n)\underline{\chs U}_h^i\big(N(n),\underline{\chs D}_h(n), \underline{\theta}_h(n)\big),\underline{v}_h\Big)_{1,h} & \forall \underline{v}_h\in\Xec{1}{r,h}, \label{eq:scheme1.1}\\
  \Big(\delta_t^{n+1} \underline{\theta}_h^i,\underline{v}_h\Big)_{1,h}&=\Big(N(n)\underline{E}_h^i\big(\underline{\chs D}_h(n+1), \underline{\theta}_h(n)\big),\underline{v}_h\Big)_{1,h} & \forall\underline{v}_h\in\Xec{1}{r,h}, \label{eq:scheme1.2}
\end{align}
\end{subequations}
where, for simplicity, we denote $\underline{\chs D}_h\coloneq(\underline{\chs D}_h^i)_{i\in[1,3]}$ (resp.\ $\underline{\theta}_h$) the collection of forms.
The terms $(\underline{\chs H}_h^i, \underline{\chs U}_h^i, \underline{E}_h^i)\in(\Xec{2}{r,h}\times\Xec{1}{r,h}\times\Xec{1}{r,h})^3$ are functions of the unknowns, defined through the discrete $L^2$-products by: for all $\underline{v}_h\in\Xec{1}{r,h}$ and $\underline{w}_h\in\Xec{2}{r,h}$,
\begin{equation}\label{eq:disc.nl.rel}
\begin{aligned}
\Big(N(n)\underline{\chs H}_h^i\big(\ued{1}{r,h}\underline{\theta}_h(n+1), \underline{\theta}_h(n)\big), &\underline{w}_h\Big)_{2,h}\\
&=\int_\Omega N(n)\chs H^i(\ded{1}{r,h}\underline{\theta}_h(n+1), \Pec{1}{r,h}\underline{\theta}_h(n))\wedge\chs\Pec{2}{r,h}\underline{w}_h,\\
\Big(N(n)\underline{\chs U}_h^i\big(\underline{\chs D}_h(n), \underline{\theta}_h(n)\big),{}&\underline{v}_h\Big)_{1,h}\\
=\int_\Omega N(n{}&)\chs U^i(\Pec{1}{r,h}\underline{\chs D}_h(n), \ded{1}{r,h}\underline{\theta}_h(n), \Pec{1}{r,h}\underline{\theta}_h(n))\wedge\chs \Pec{1}{r,h}\underline{v}_h,\\
\Big(N(n)\underline{E}_h^i\big(\underline{\chs D}_h(n+1), \underline{\theta}_h(n)\big),\underline{v}_h\Big)_{1,h}&=\int_\Omega N(n) E^i(\Pec{1}{r,h}\underline{\chs D}_h(n+1), \Pec{1}{r,h}\underline{\theta}_h(n))\wedge \chs\Pec{1}{r,h}\underline{v}_h,
\end{aligned}
\end{equation}
where, on the right-hand side, $H^i$, $U^i$, $E^i$ are the relations \eqref{eq:rel.Hij}, \eqref{eq:rel.Ui}, and \eqref{eq:rel.Ek}, calculated replacing $(D^i)_{i\in[1,3]}$ by $\chs \Pec{1}{r,h}\underline{\chs D}_h$, $(\theta^i)_{i\in[1,3]}$ by $\Pec{1}{r,h}\underline{\theta}_h$, and $(B^i)_{i\in[1,3]}$ by $\ded{1}{r,h}\underline{\theta}_h$. The $D^0$ that appears in the calculation of $U^i$ is dealt with the same way using relation \eqref{eq:rel.D0}.

For the second three-field system \eqref{eq:weak.ev2}, the primary unknowns are $(\chs D^i,\chs \theta^i,\chs B^i)$ and the scheme reads: With lapse $N$ and initial conditions given by $(\underline{\chs D}_h^i(0), \underline{\chs\theta}_h^i(0), \underline{\chs B}_h^i(0))_{i\in[1,3]}\in(\Xec{1}{r,h}\times\Xec{2}{r,h}\times\Xec{1}{r,h})^3$, find for every $n$ a collection of forms $(\underline{\chs D}_h^i(n), \underline{\chs\theta}_h^i(n), \underline{\chs B}_h^i(n))_{i\in[1,3]}\in(\Xec{1}{r,h}\times\Xec{2}{r,h}\times\Xec{1}{r,h})^3$ such that, for all $i\in [1,3]$,
\begin{subequations}\label{eq:scheme2}
\begin{alignat}{2}\nonumber
  \Big(\delta_t^{n+1}  \underline{\chs D}_h^i,\underline{v}_h\Big)_{1, h}-\Big(N(n)\underline{\chs H}_h^i&\big(\underline{\chs B}_h(n+1), \underline{\chs\theta}_h(n)\big), \ued{1}{r,h}\underline{v}_h\Big)_{2,h}\\
  =\Big(N(n)\underline{\chs U}_h^i&\big(\underline{\chs D}_h(n), \underline{\chs\theta}_h(n), \underline{\chs B}_h(n)\big),\underline{v}_h\Big)_{1,h} & \quad\forall \underline{v}_h\in\Xec{1}{r,h}, \label{eq:scheme2.1}\\
  \Big(\delta_t^{n+1} \underline{\chs\theta}_h^i,\underline{w}_h\Big)_{2,h}=\Big(N(n)\underline{\chs E}_h^i&\big(\underline{\chs D}_h(n+1), \underline{\chs\theta}_h(n)\big),\underline{w}_h\Big)_{2,h} & \quad\forall\underline{w}_h\in\Xec{2}{r,h}, \label{eq:scheme2.2}\\
  \Big(\delta_t^{n+1}  \underline{\chs B}_h^i,\underline{v}_h\Big)_{1, h}-\Big(N(n)\underline{\chs E}_h^i\big(&\underline{\chs D}_h(n+1), \underline{\chs\theta}_h(n)\big), \ued{1}{r,h}\underline{v}_h\Big)_{2,h}=0 & \quad\forall \underline{v}_h\in\Xec{1}{r,h}, \label{eq:scheme2.3}
\end{alignat}
\end{subequations}
The $\underline{\chs H}_h$, $\underline{\chs U}_h$ and $\underline{\chs E}_h$ terms are calculated as for the first scheme \eqref{eq:scheme1}, with the replacements $(D^i)_{i\in[1,3]}$ by $\chs \Pec{1}{r,h}\underline{\chs D}_h$, $(\theta^i)_{i\in[1,3]}$ by $\chs\Pec{2}{r,h}\underline{\chs\theta}_h$, and $(B^i)_{i\in[1,3]}$ by $\chs \Pec{1}{r,h}\underline{\chs B}_h$ instead. See appendix \ref{sec:disc.formulas} for details on the calculation procedure.

\begin{remark}\label{rm:choice.DtB}
In schemes \eqref{eq:scheme1} and \eqref{eq:scheme2}, certain unknowns are taken to be the constant Hodge star of their respective fields. For $\underline{\shs D}_h$, $\underline{\shs B}_h$, this choice is more or less enforced by the position of the exterior derivative in \eqref{eq:scheme1.1}, \eqref{eq:scheme2.1}, and \eqref{eq:scheme2.3}, and we can see it play out in the simple case of Maxwell's equations. Take, for example, the vector proxy evolution equation $\partial_t \bvec{E} - \CURL \bvec{B} =0$ for the electric $\bvec{E}$ and magnetic field $\bvec{B}$, and look at the two possible weak forms
\begin{subequations}\label{eq:mx}
\begin{align}
	(\partial_t \bvec{E},\bvec{v})_{\Omega} - (\CURL \bvec{B},\bvec{v})_{\Omega}=0, \label{eq:mx.1}\\
	(\partial_t \bvec{E},\bvec{v})_{\Omega} - (\bvec{B},\CURL\bvec{v})_{\Omega}=0. \label{eq:mx.2}
\end{align}
\end{subequations}
In \eqref{eq:mx.1}, we have that $\bvec{B}\in\Hcurl{\Omega}$, $\bvec{v}\in\Hdiv{\Omega}$, forcing $\bvec{E}\in\Hdiv{\Omega}$, while in \eqref{eq:mx.2}, $\bvec{B}\in\Hdiv{\Omega}$, $\bvec{v}\in\Hcurl{\Omega}$, resulting in $\bvec{E}\in\Hcurl{\Omega}$. In terms of the canonical $1$-form $E$ and $2$-form $B$, \eqref{eq:mx.1} amounts to having the primary unknowns $(\chs E, \chs B)$, and \eqref{eq:mx.2} to $(E, B)$, that are hidden by the vector proxies. A similar reasoning fixes the discretisation of $\theta^i$: In \eqref{eq:scheme1}, it is natural to prefer $\underline{\theta}^i_h\in\Xec{1}{r,h}$ so that a discrete $\underline{B}^i_h=\ued{1}{r,h}\underline{\theta}^i_h$ can be defined directly (recall that this $\underline{B}^i_h$ is necessary to calculate the relations in \eqref{eq:disc.nl.rel}), while for \eqref{eq:scheme2}, it depends on the form of \eqref{eq:scheme2.3} so that certain cancellations can be made to prove the conservation of a discrete constraint (see the proof of Proposition \ref{prop:const.pres}). Of course, these restrictions are just symptoms of the underlying problem that there is no straightforward discrete analogue of the Hodge star that gives an identification between discrete $k$-forms and $(n-k)$-forms.
\end{remark}

\begin{remark}
Explicit schemes are often subject to CFL conditions, where the timestep $\delta t^{n+\frac12}$ needs to be small enough in proportion to $h$ to ensure stability of the numerical scheme, but not so small that simulations times blow up. In practice, it is often trial and error to approach the optimal value, and we noticed that the explicit variations of \eqref{eq:scheme1} and \eqref{eq:scheme2} diverges in certain situations under the timestep imposed by $\left\lceil 3/h^{r+1}\right\rceil$ number of iterations used in the tests of Sec.~\ref{sec:tests}. The choice of semi-implicit time stepping lets us take larger timesteps, and is in general more stable, at the cost of solving a linear system at each time. The exact choice of time for each function in \eqref{eq:scheme1} and \eqref{eq:scheme2} are made so that the system is never fully nonlinear (see the relations \eqref{eq:rel} and how they are dealt with in Appendix \ref{sec:disc.formulas}), which removes the need for expensive nonlinear methods such as Newton. The fact that $\theta^i$ is always discretised at $n$ also means that it can be decoupled and calculated afterwards, making the system much leaner to solve.
\end{remark}

In the two-field scheme \eqref{eq:scheme1}, the discrete versions of the constraints \eqref{eq:form.cst.2} and \eqref{eq:form.cst.3} are preserved in strong form by taking the definition $\underline{B}^i_h\coloneq\ued{1}{r, h}\underline{\theta}_h$ and using the discrete complex property to see that $\ued{2}{r, h}\underline{B}^i_h=\ued{2}{r, h}\ued{1}{r, h}\underline{\theta}_h=0$. For the three-field scheme \eqref{eq:scheme2}, while the same preservation is possible, it would require to solve the explicit form of $\underline{E}^i_h$ at each step, adding to the cost of the method. With the choice of unknowns in the weak formulation \eqref{eq:weak.ev2}, we can prove the conservation of a weak form of these two constraints, while avoiding the extra complexity.

\begin{proposition}[Preservation of discrete constraints for the three-field scheme \eqref{eq:scheme2}]\label{prop:const.pres}~\\
Let $(\underline{\chs D}_h^i(n), \underline{\chs \theta}_h^i(n), \underline{\chs  B}_h^i(n))_{i\in[1,3]}\in(\Xec{1}{r,h}\times\Xec{2}{r,h}\times\Xec{1}{r,h})^3$ be solutions to \eqref{eq:scheme2}. Then the weak constraints $\mathfrak{C}^i_1(n,\underline{u}_h)$, $\mathfrak{C}^i_2(n,\underline{p}_h)$ defined as
\begin{alignat*}{2}
\mathfrak{C}^i_1(n,\underline{u}_h)\coloneq{}&\big(\underline{\chs \theta}_h^i(n),\ued{1}{r,h}\underline{u}_h\big)_{2,h} - (\underline{\chs B}_h^i(n),\underline{u}_h)_{1,h} &\qquad  \forall \underline{u}_h\in\Xec{1}{r,h},\\
\mathfrak{C}^i_2(n,\underline{p}_h)\coloneq{}&\big(\underline{\chs B}_h^i(n),\ued{0}{r,h}\underline{p}_h\big)_{1,h} &\qquad\forall \underline{p}_h\in\Xec{0}{r,h}.
\end{alignat*}
remains stationary for all $n\in[0,N_T]$.
\end{proposition}

\begin{proof}
We prove that $\mathfrak{C}^i_1(n,\underline{u}_h)$ stays constant by fixing $\underline{u}_h\in\Xec{1}{r,h}$ in time, and taking the discrete time derivative to get 
\begin{align*}
  \delta_t^{n+1}\mathfrak{C}^i_1(\cdot,\underline{u}_h)
  &=\big(\delta_t^{n+1} \underline{\chs \theta}_h^i, \ued{1}{r,h}\underline{u}_h\big)_{2,h} - \big(\delta_t^{n+1}\underline{\chs B}_h^i,\underline{u}_h\big)_{1,h}\\
  &=\Big(N(n)\underline{\chs E}_h^i\big(\underline{\chs D}_h(n+1), \underline{\chs \theta}_h(n)\big),\ued{1}{r,h}\underline{u}_h\Big)_{2,h} - \big(\delta_t^{n+1}\underline{\chs B}_h^i,\underline{u}_h\big)_{1,h} \\
  \overset{\eqref{eq:scheme2.3}}&=\big(\delta_t^{n+1} \underline{\chs B}_h^i,\underline{u}_h\big)_{1,h}- \big(\delta_t^{n+1}\underline{\chs B}_h^i,\underline{u}_h\big)_{1,h}=0,
\end{align*}
where we use \eqref{eq:scheme2.2} with $\underline{u}_h=\ued{1}{r,h}\underline{v}_h$ to get the second line. Since this is true for all $n$, $\mathfrak{C}^i_1(\cdot,\underline{u}_h)$ is stationary in time.

For $\mathfrak{C}^i_2(n,\underline{p}_h)$, we fix $\underline{p}_h\in\Xec{0}{r,h}$ constant in time and take the discrete time derivative
\begin{align*}
  \delta_t^{n+1}\mathfrak{C}^i_2(\cdot,\underline{p}_h)
  &=\big(\delta_t^{n+1}\underline{\chs B}_h^i,\ued{0}{r,h}\underline{p}_h\big)_{1,h}\\
  &=\Big(N(n)\underline{\chs E}_h^i\big(\underline{\chs D}_h(n+1), \underline{\chs \theta}_h(n)\big),\ued{1}{r,h}\ued{0}{r,h}\underline{p}_h\Big)_{2,h}=0,
\end{align*}
where we use the evolution \eqref{eq:scheme2.3} with $\underline{u}_h=\ued{0}{r,h}\underline{p}_h$ to get the second line, and the property $\ued{1}{r,h}\ued{0}{r,h}=0$ to cancel the term. The conclusion follows.
\end{proof}

\begin{remark}
 Proposition \ref{prop:const.pres} applies equally to the explicit variation of the scheme \eqref{eq:scheme2}, since the proof of preservation for the first constraint $\mathfrak{C}^i_1(n,\underline{u}_h)$ depends on a consistent discretisation of the $(N\chs E^i, w)_{\Omega}$ in \eqref{eq:weak.ev2.2} and $(N\chs E^i, \sed v)_{\Omega}$ in \eqref{eq:weak.ev2.3} so that they cancel, which would be true if they were taken at time $n$. The second constraint $\mathfrak{C}^i_2(n,\underline{p}_h)$ simply relies on $\ued{1}{r,h}\ued{0}{r,h}=0$, which is always satisfied.
\end{remark}

\begin{remark}\label{rem:C3}
It is currently unknown if there is any preservation of a discrete version of the Einstein constraint \eqref{eq:form.cst.1} for either scheme. We see in Theorem \ref{thm:einstein1} that for the continuous case, it is a direct result of the contracted Bianchi identity, which is claimed in \cite{Olivares.Peshkov.ea:22} to be equivalent to the evolution of \eqref{eq:3+1.2}. The implication of a discrete version of this equation is unlikely due to the presence of nonlinear terms in the derivatives of \eqref{eq:3+1.2}, which do not expand nicely once the terms are discretised; similar issues have been encountered in the design of DDR schemes for the Yang--Mills equations \cite{Droniou.Oliynyk.ea:23}. Nevertheless, we can define a discrete quantity $\mathfrak{C}^\alpha_3$ reflecting the constraint \eqref{eq:form.cst.1} in the same way as for Proposition \ref{prop:const.pres}. For the two field system \eqref{eq:scheme1}, the nonlinear term is constructed the same way as in the scheme
\begin{equation*}
	(\underline{\chs V}_h^\alpha(n),\underline{p}_h)_{0,h} = \int_\Omega\chs V^\alpha(\Pec{1}{r,h}\underline{\chs D}_h(n), \ded{1}{r,h}\underline{\theta}_h(n), \Pec{1}{r,h}\underline{\theta}_h(n))\wedge\chs\Pec{0}{r,h}\underline{p}_h,
\end{equation*}
and for $\alpha=0$, we also reconstruct the $D^0$ to get
\begin{align*}
	\mathfrak{C}^0_3(n,\underline{p}_h)\coloneq{}&-\int_\Omega D^0(\ded{1}{r,h}\underline{\theta}_h(n), \Pec{1}{r,h}\underline{\theta}_h(n))\wedge\ded{0}{r,h}\underline{p}_h - (\underline{\chs V}_h^0(n),\underline{p}_h)_{0,h} &  \forall \underline{p}_h\in\Xec{0}{r,h},
\end{align*}
while for $\alpha=i$,
\begin{align*}
	\mathfrak{C}^i_3(n,\underline{p}_h)\coloneq{}&-\int_\Omega\Pec{1}{r,h}\underline{\chs D}^i_h(n)\wedge\chs\ded{0}{r,h}\underline{p}_h - (\underline{\chs V}_h^i(n),\underline{p}_h)_{0,h} &  \forall \underline{p}_h\in\Xec{0}{r,h}.
\end{align*}
A similar quantity can be defined for the three-field system \eqref{eq:scheme2}.
\end{remark}

\section{Numerical tests}\label{sec:tests}

\begin{figure}\centering
  \begin{minipage}{0.275\textwidth}
    \includegraphics[width=0.90\textwidth]{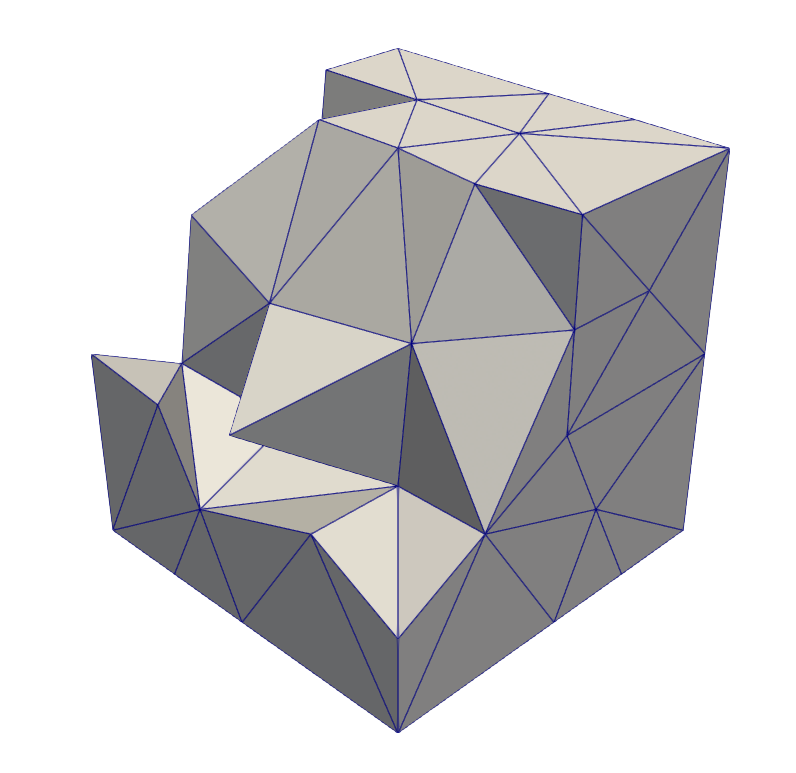}
    \subcaption{Mesh 1}
  \end{minipage}
  \hspace{0.25cm}
  \begin{minipage}{0.275\textwidth}
    \includegraphics[width=0.90\textwidth]{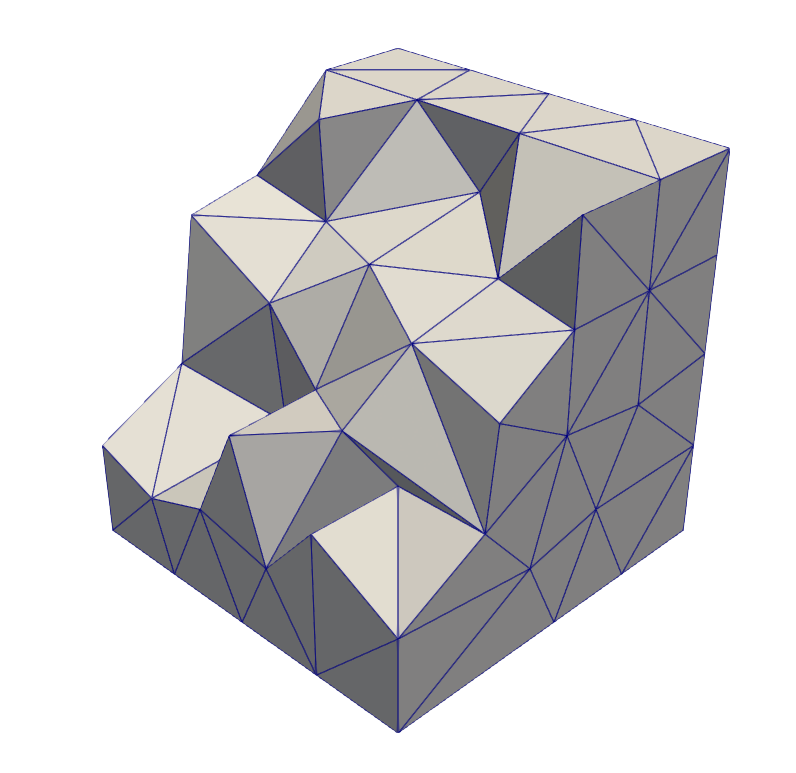}
    \subcaption{Mesh 2}
  \end{minipage}
  \hspace{0.25cm}
  \begin{minipage}{0.275\textwidth}
    \includegraphics[width=0.90\textwidth]{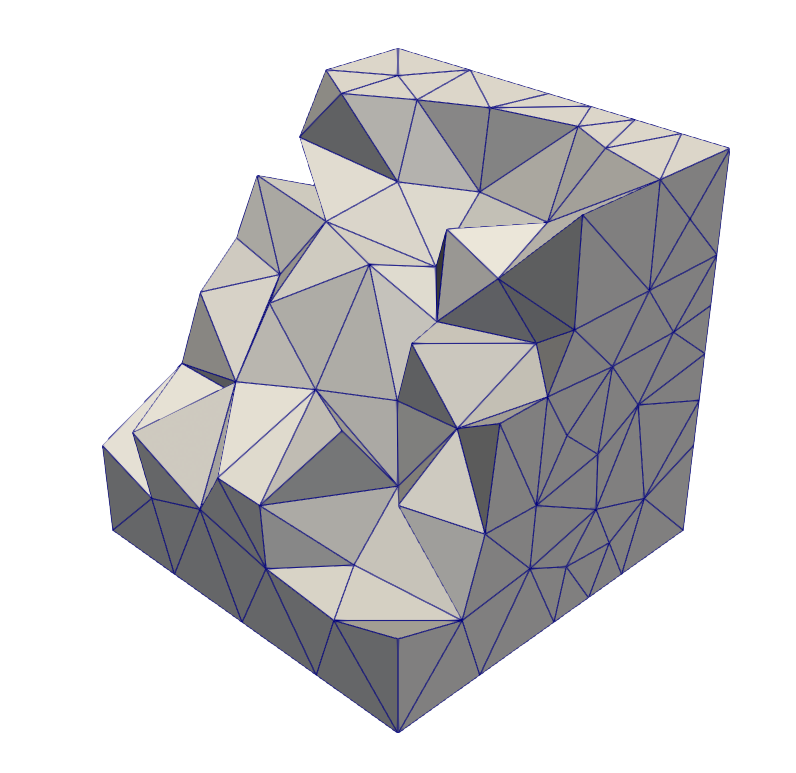}
    \subcaption{Mesh 3}
  \end{minipage} \\
  \begin{minipage}{0.275\textwidth}
    \includegraphics[width=0.90\textwidth]{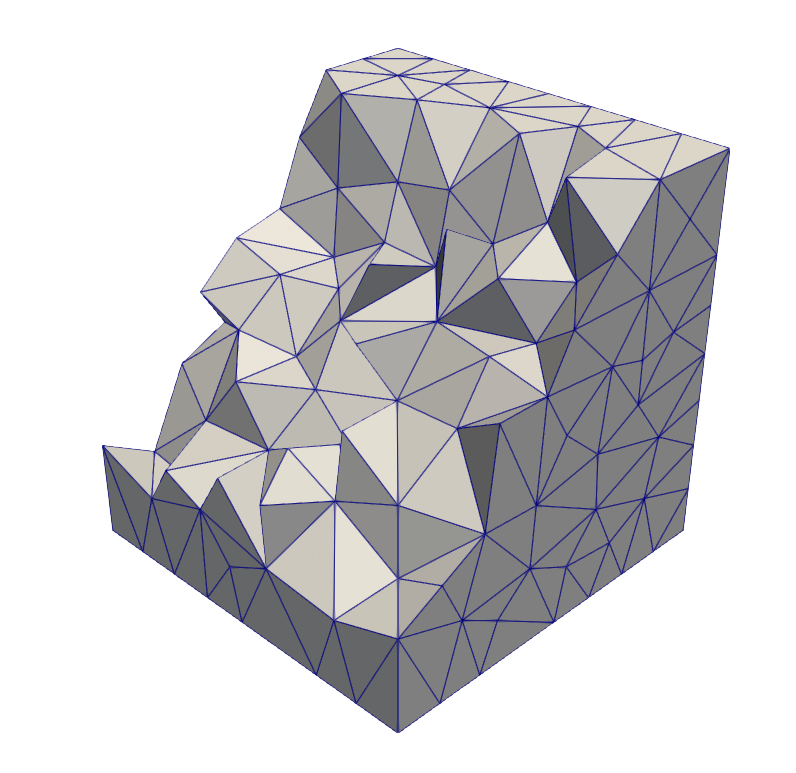}
    \subcaption{Mesh 4}
  \end{minipage}
  \hspace{0.25cm}
  \begin{minipage}{0.275\textwidth}
    \includegraphics[width=0.90\textwidth]{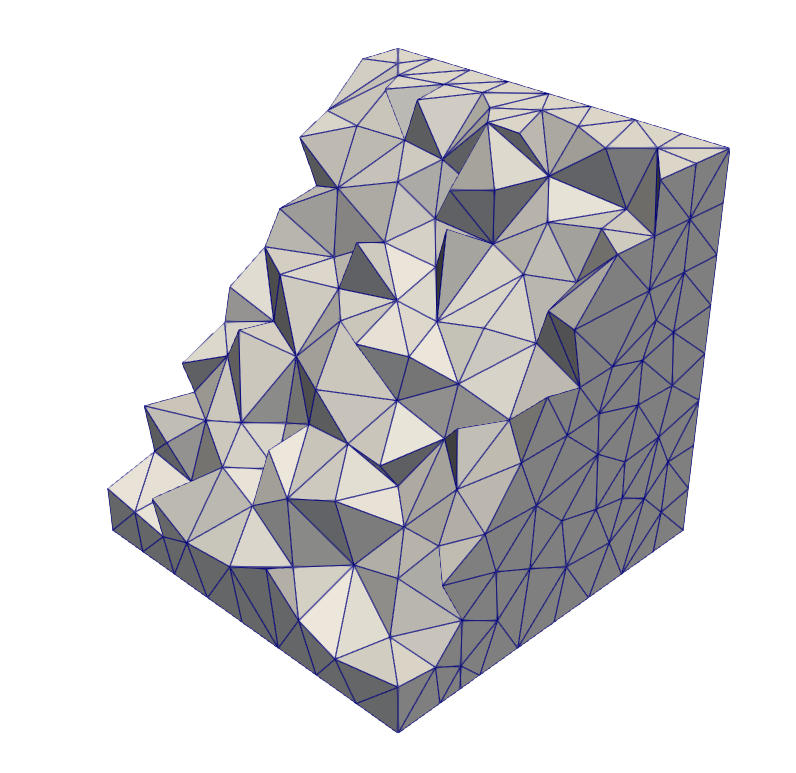}
    \subcaption{Mesh 5}
  \end{minipage}
  \caption{The ``Tetgen-cube-0'' mesh sequence used in the numerical tests}
  \label{fig:mesh.tets}
\end{figure}

\begin{figure}\centering
  \begin{minipage}{0.275\textwidth}
    \includegraphics[width=0.90\textwidth]{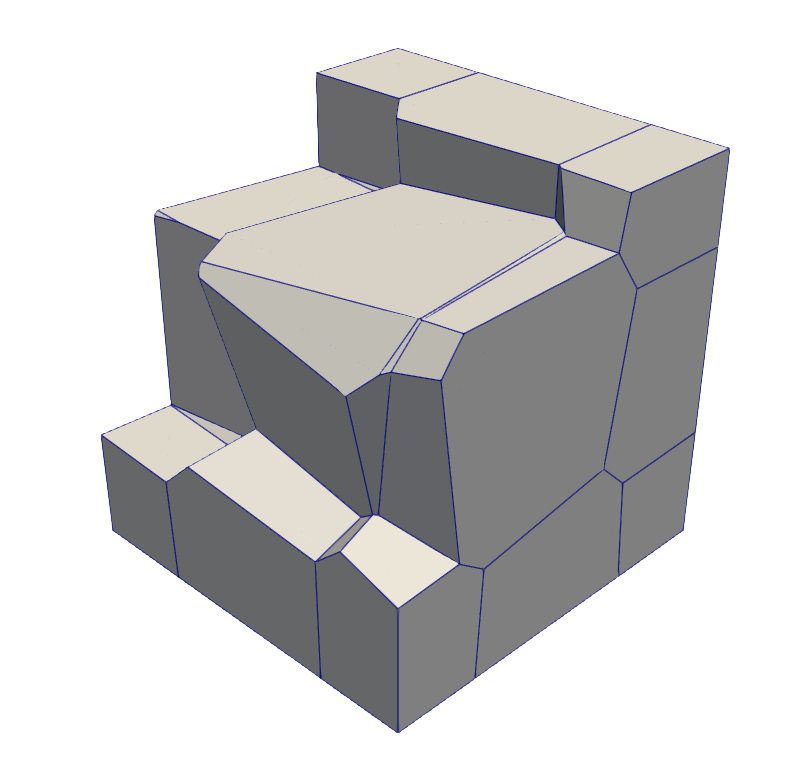}
    \subcaption{Mesh 1}
  \end{minipage}
  \hspace{0.25cm}
  \begin{minipage}{0.275\textwidth}
    \includegraphics[width=0.90\textwidth]{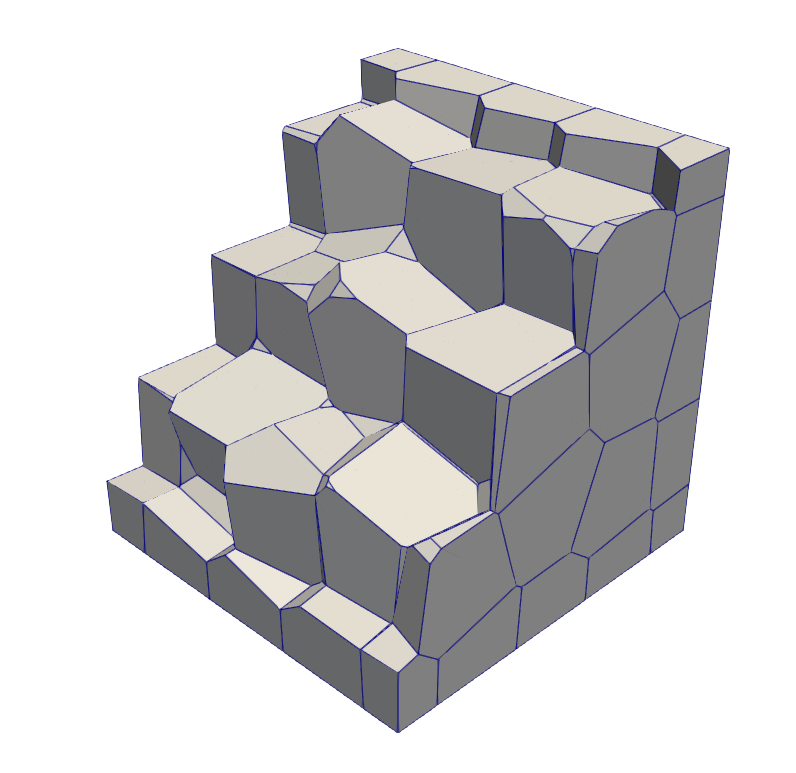}
    \subcaption{Mesh 2}
  \end{minipage}
  \hspace{0.25cm}
  \begin{minipage}{0.275\textwidth}
    \includegraphics[width=0.90\textwidth]{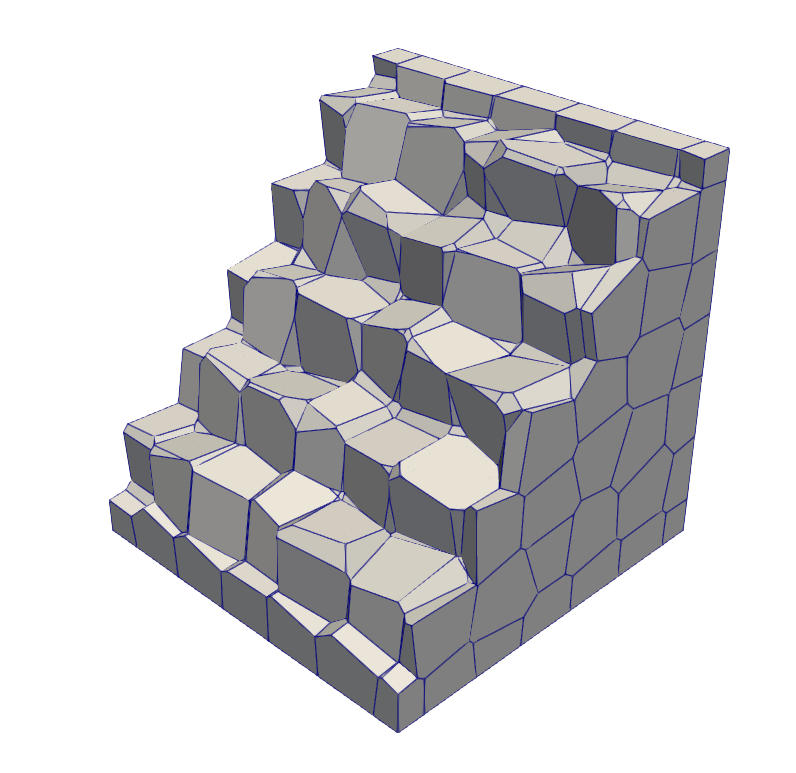}
    \subcaption{Mesh 3}
  \end{minipage} \\
  \begin{minipage}{0.275\textwidth}
    \includegraphics[width=0.90\textwidth]{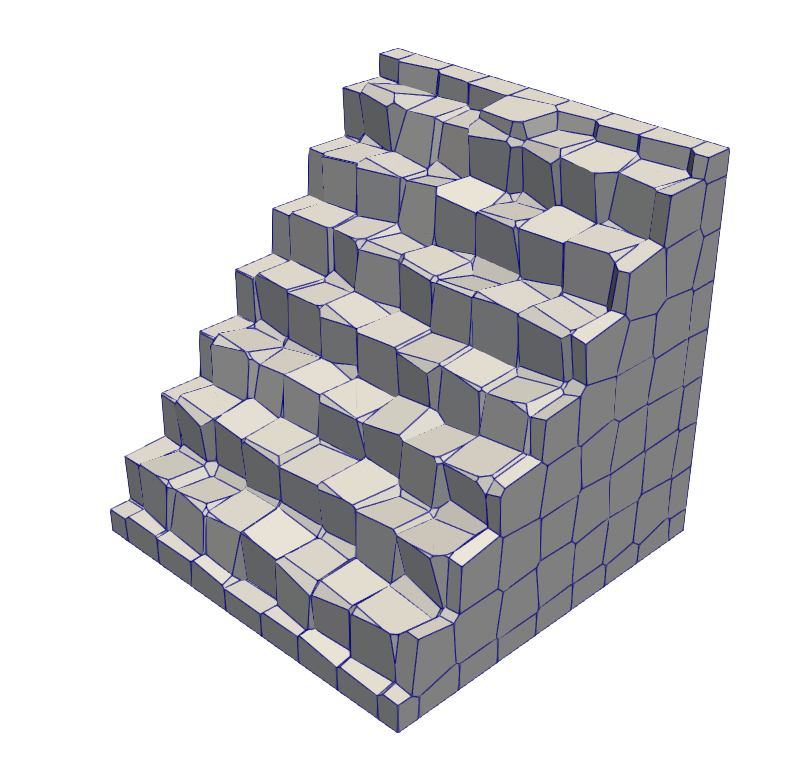}
    \subcaption{Mesh 4}
  \end{minipage}
  \hspace{0.25cm}
  \begin{minipage}{0.275\textwidth}
    \includegraphics[width=0.90\textwidth]{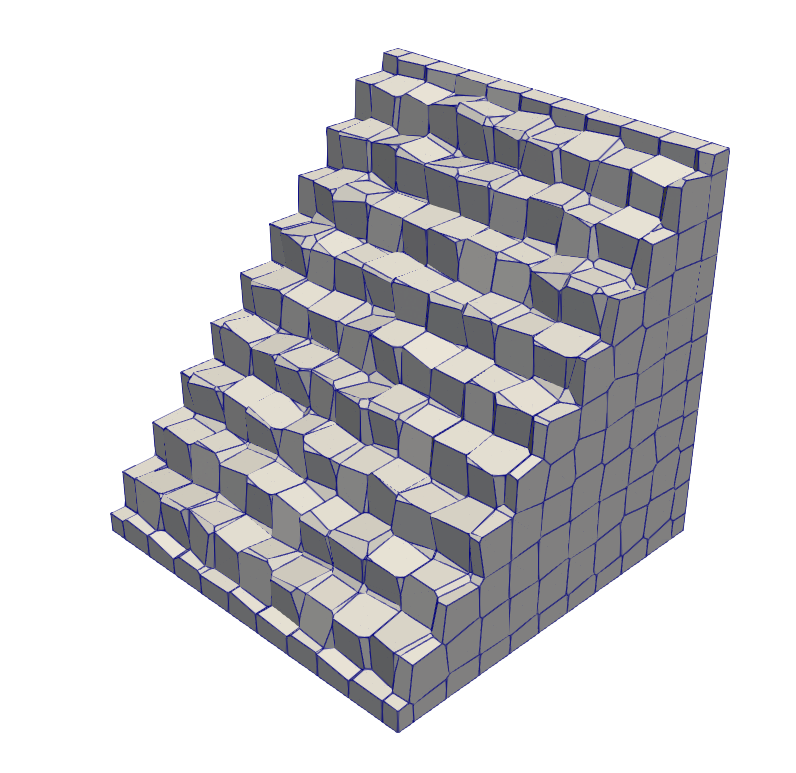}
    \subcaption{Mesh 5}
  \end{minipage}
  \caption{The ``Voro-small-0'' mesh sequence used in the numerical tests}
  \label{fig:mesh.voro}
\end{figure}

\begin{figure}\centering
\begin{tabular}{|c|cllll|}
\toprule
\multicolumn{1}{|c|}{\textbf{Tetgen-cube-0}}&Mesh size $h$& No. Vertices & No. Edges &No. Faces & No. Cells\\
\midrule

Mesh 1 & 0.559 & 75 & 354 & 496 & 216 \\
Mesh 2 & 0.500 & 124 & 628 & 913 & 408 \\
Mesh 3 & 0.392 & 229 & 1217 & 1805 & 816 \\
Mesh 4 & 0.313 & 383 & 2139 & 3261 & 1504 \\
Mesh 5 & 0.257 & 663 & 3965 & 6228 & 2925 \\
\midrule
\multicolumn{1}{|c|}{\textbf{Voro-small-0}}&Mesh size $h$& No. Vertices & No. Edges &No. Faces & No. Cells\\
\midrule
Mesh 1 & 0.827 & 138 & 272 & 162 & 27 \\
Mesh 2 & 0.454 & 678 & 1352 & 800 & 125 \\
Mesh 3 & 0.305 & 2011 & 4018 & 2351 & 343 \\
Mesh 4 & 0.221 & 4370 & 8736 & 5096 & 729 \\
Mesh 5 & 0.177 & 8179 & 16354 & 9507 & 1331 \\
\bottomrule
\end{tabular}
\caption{Mesh statistics for ``Tetgen-cube-0'' (Figure \ref{fig:mesh.tets}) and ``Voro-small-0'' (Figure \ref{fig:mesh.voro})}
\label{tab:mesh.stats}
\end{figure}

We present here some numerical results obtained by both schemes \eqref{eq:scheme1} and \eqref{eq:scheme2}. The numerical code for these schemes can be found in the \texttt{HArDCore3D} repository (see \url{https://github.com/jdroniou/HArDCore}), and were implemented in C++ utilising the exterior calculus DDR module provided by the \texttt{HArDCore3D} library. The matrix operations are facilitated with the \texttt{Eigen3} library (see \url{http:
//eigen.tuxfamily.org}), and linear systems solved using the parallel direct solver \texttt{Intel MKL PARADISO} (see \url{https://software.intel.com/en-us/mkl}).

We test two exact solutions to Einstein's equations, the first being the homogeneous Kasner solution, that describes the expansion of the universe right after the big bang, and the second the Gowdy wave solution in \cite{Alcubierre.Allen.ea:03}, modeling gravitational waves in an expanding universe. The time runs from $1$ to $1.1$ for both tests, avoiding the $t<1$ region where the Kasner solution might be too anisotropic, and the possible complications with long-term simulations since there is no apparent reference for the scale of time that we are dealing with. The size of the timestep is decided by the number of iterations $\left\lceil 3/h^{r+1}\right\rceil$ that depends on the mesh size $h$, and the polynomial degree $r$, so that an $O(h^{r+1})$ convergence on the error can be expected. The spatial domain is the unit cube $\Omega=(0,1)^3$, discretised with tetrahedral ``Tetgen-cube-0'' (Figure \ref{fig:mesh.tets}) and Voronoi ``Voro-small-0'' (Figure \ref{fig:mesh.voro}) mesh sequences  provided in the \texttt{HArDCore3D} library. Both tests are run using the two polynomial degrees $r=0$ and $r=1$ for the ECDDR discretisation, and with natural boundary conditions. Although standard tests for these solutions are run with periodic boundary conditions, this practice is less common in the usual problems tackled by the \texttt{HArDCore3D} library, and the addition of the necessary framework is still a work in progress. Errors are computed at the final time in two different norms: for a solution $Z^i$, and its approximation $\underline{Z}^i_h\in\Xec{k}{r,h}$, we calculate the discrete total relative error as
\begin{equation}\label{eq:dtre}
	{\rm E_{disc}}({Z}) \coloneq \frac{\sum_{i=1}^3||\underline{Z}^i_h-\Iec{k}{r,h}Z^i||_{k, h}}{\sum_{i=1}^3||Z^i||_{\Lpdf{2}{k}{\Omega}}},
\end{equation}
where $\Iec{k}{r,h}$ is the interpolator onto $\Xec{k}{r,h}$, and the continuous total relative error as
\begin{equation}\label{eq:ctre}
	{\rm E_{cont}}({Z}) \coloneq \frac{\sum_{i=1}^3||\Pec{k}{r,h}\underline{Z}^i_h-Z^i||_{\Lpdf{2}{k}{\Omega}}}{\sum_{i=1}^3||Z^i||_{\Lpdf{2}{k}{\Omega}}}.
\end{equation}
The discrete error measures the difference between the discrete solution and the interpolate of the exact solution, while the continuous error measures the difference between the polynomial reconstruction and the exact solution. In particular, ${\rm E_{cont}}$ does not depend on the discrete $L^2$-product, that can vary depending on the chosen stabilisation parameter, making it a better measure of convergence when we test different stabilisation parameters.

\subsection{Convergence tests}
In 3 dimensions, the Kasner line element is
\begin{equation*}
	\ed s^2 = -\ed t^2 + t^{2i_1}\ed x^2 + t^{2i_2}\ed y^2 + t^{2i_3}\ed z^2
\end{equation*}
for real indices meeting the Kasner conditions
\begin{equation*}
i_1+i_2+i_3=1, \qquad
i_1^2+i_2^2+i_3^2=1.
\end{equation*}
We choose the solution with $i_1=\frac12$, $i_2=\frac{1-\sqrt{5}}{4}$, $i_3=\frac{1+\sqrt{5}}{4}$, and run the schemes with lapse $N=1$ to get the results for Scheme \eqref{eq:scheme1} in Figure \ref{fig:kas.sch1} and for Scheme \eqref{eq:scheme2} in Figure \ref{fig:kas.sch2}, measured at the final time $t=1.1$.

%\begin{align*}
%\theta &=
%\begin{bmatrix}
%    t^{0.5} & 0 & 0 \\
%    0 & t^{-0.309} & 0 \\
%    0 & 0 &  t^{0.809}
%\end{bmatrix}
%&
%N &=1\\
%B &=
%\begin{bmatrix}
%    0 & 0 & 0 \\
%    0 & 0 & 0 \\
%    0 & 0 & 0
%\end{bmatrix}
%&
%H &=
%\begin{bmatrix}
%    0 & 0 & 0 \\
%    0 & 0 & 0 \\
%    0 & 0 & 0
%\end{bmatrix} \\
%D &=
%\begin{bmatrix}
%    0 & 0 & 0.5t^{-0.5} \\
%    0 & -1.309t^{0.309} & 0 \\
%    0.190t^{-0.809} & 0 & 0
%\end{bmatrix}
%&
%E &=
%\begin{bmatrix}
%    0.5t^{-0.5} & 0 & 0 \\
%    0 & -0.309t^{-1.309} & 0 \\
%    0 & 0 & 0.809t^{-0.19}
%\end{bmatrix}
%\end{align*}

\begin{figure}\centering
	\ref{kas.sch1}
  \vspace{0.50cm}\\
  \begin{minipage}{0.45\textwidth}
    \begin{tikzpicture}[scale=0.85]
      \begin{loglogaxis} [legend columns=2, legend to name=kas.sch1]
        \logLogSlopeTriangle{0.90}{0.4}{0.1}{1}{black}
        \logLogSlopeTriangle{0.90}{0.4}{0.1}{2}{black}
        \addplot [mark=star, red] table[x=MeshSize,y=relE_L2D] {results/ecddr-einstein/kasner/Tetgen-Cube-0_k0/data_rates-1.dat};
        \addlegendentry{${\rm E_{disc}}(D)$, $r=0$;}
        \addplot [mark=o, blue] table[x=MeshSize,y=relE_L2Theta] {results/ecddr-einstein/kasner/Tetgen-Cube-0_k0/data_rates-1.dat};
         \addlegendentry{${\rm E_{disc}}(\theta)$, $r=0$}
         \addplot [mark=star, mark options=solid, red, dashed] table[x=MeshSize,y=relE_L2D] {results/ecddr-einstein/kasner/Tetgen-Cube-0_k1/data_rates-1.dat};
         \addlegendentry{${\rm E_{disc}}(D)$, $r=1$;}
         \addplot [mark=o, mark options=solid, blue, dashed] table[x=MeshSize,y=relE_L2Theta] {results/ecddr-einstein/kasner/Tetgen-Cube-0_k1/data_rates-1.dat};
         \addlegendentry{${\rm E_{disc}}(\theta)$, $r=1$}
      \end{loglogaxis}            
    \end{tikzpicture}
    \subcaption{Tetrahedral meshes (Figure \ref{fig:mesh.tets})}
  \end{minipage}
  \begin{minipage}{0.45\textwidth}
    \begin{tikzpicture}[scale=0.85] 
      \begin{loglogaxis}
        \logLogSlopeTriangle{0.90}{0.4}{0.1}{1}{black}
        \logLogSlopeTriangle{0.90}{0.4}{0.1}{2}{black}
        \addplot [mark=star, red] table[x=MeshSize,y=relE_L2D] {results/ecddr-einstein/kasner/Voro-small-0_k0/data_rates-1.dat};
        \addplot [mark=o, blue] table[x=MeshSize,y=relE_L2Theta] {results/ecddr-einstein/kasner/Voro-small-0_k0/data_rates-1.dat};
        \addplot [mark=star, mark options=solid, red, dashed] table[x=MeshSize,y=relE_L2D] {results/ecddr-einstein/kasner/Voro-small-0_k1/data_rates-1.dat};
        \addplot [mark=o, mark options=solid, blue, dashed] table[x=MeshSize,y=relE_L2Theta] {results/ecddr-einstein/kasner/Voro-small-0_k1/data_rates-1.dat};
        \end{loglogaxis} 
      \end{tikzpicture}
    \subcaption{Voronoi meshes (Figure \ref{fig:mesh.voro})}
  \end{minipage}
  \captionsetup{width=0.85\textwidth}
  \caption{Kasner solution for the two-field scheme \eqref{eq:scheme1}: Plot of discrete total relative error \eqref{eq:dtre} for $D$, $\theta$, (vertical axis) against mesh size $h$ (horizontal axis) for $r=0$, $1$ polynomial degree ECDDR spaces (stabilisation parameter $\rho=1$ \eqref{def:L2p})} 
  \label{fig:kas.sch1}
\end{figure} 

\begin{figure}\centering
	\ref{kas.sch2}
  \vspace{0.50cm}\\
  \begin{minipage}{0.45\textwidth}
    \begin{tikzpicture}[scale=0.85]
      \begin{loglogaxis} [legend columns=3, legend to name=kas.sch2]
        \logLogSlopeTriangle{0.90}{0.4}{0.1}{1}{black}
        \logLogSlopeTriangle{0.90}{0.4}{0.1}{2}{black}
        \addplot [mark=star, red] table[x=MeshSize,y=relE_L2D] {results/ecddr-einstein-wp/kasner/stab_1/Tetgen-Cube-0_k0/data_rates-1.dat};
        \addlegendentry{${\rm E_{disc}}(D)$, $r=0$;}
        \addplot [mark=o, blue] table[x=MeshSize,y=relE_L2Theta] {results/ecddr-einstein-wp/kasner/stab_1/Tetgen-Cube-0_k0/data_rates-1.dat};
         \addlegendentry{${\rm E_{disc}}(\theta)$, $r=0$;}
         \addplot [mark=square, green] table[x=MeshSize,y=relE_L2B] {results/ecddr-einstein-wp/kasner/stab_1/Tetgen-Cube-0_k0/data_rates-1.dat};
         \addlegendentry{${\rm E_{disc}}(B)$, $r=0$;}
         \addplot [mark=star, mark options=solid, red, dashed] table[x=MeshSize,y=relE_L2D] {results/ecddr-einstein-wp/kasner/stab_1/Tetgen-Cube-0_k1/data_rates-1.dat};
         \addlegendentry{${\rm E_{disc}}(D)$, $r=1$;}
         \addplot [mark=o, mark options=solid, blue, dashed] table[x=MeshSize,y=relE_L2Theta] {results/ecddr-einstein-wp/kasner/stab_1/Tetgen-Cube-0_k1/data_rates-1.dat};
         \addlegendentry{${\rm E_{disc}}(\theta)$, $r=1$;}
         \addplot [mark=square, mark options=solid, green, dashed] table[x=MeshSize,y=relE_L2B] {results/ecddr-einstein-wp/kasner/stab_1/Tetgen-Cube-0_k1/data_rates-1.dat};
         \addlegendentry{${\rm E_{disc}}(B)$, $r=1$;}
      \end{loglogaxis}            
    \end{tikzpicture}
    \subcaption{Tetrahedral meshes (Figure \ref{fig:mesh.tets})}
  \end{minipage}
  \begin{minipage}{0.45\textwidth}
    \begin{tikzpicture}[scale=0.85] 
      \begin{loglogaxis}
        \logLogSlopeTriangle{0.90}{0.4}{0.1}{1}{black}
        \logLogSlopeTriangle{0.90}{0.4}{0.1}{2}{black} 
        \addplot [mark=star, red] table[x=MeshSize,y=relE_L2D] {results/ecddr-einstein-wp/kasner/stab_1/Voro-small-0_k0/data_rates-1.dat};
        \addplot [mark=o, blue] table[x=MeshSize,y=relE_L2Theta] {results/ecddr-einstein-wp/kasner/stab_1/Voro-small-0_k0/data_rates-1.dat};
        \addplot [mark=square, green] table[x=MeshSize,y=relE_L2B] {results/ecddr-einstein-wp/kasner/stab_1/Voro-small-0_k0/data_rates-1.dat};
        \addplot [mark=star, mark options=solid, red, dashed] table[x=MeshSize,y=relE_L2D] {results/ecddr-einstein-wp/kasner/stab_1/Voro-small-0_k1/data_rates-1.dat};
        \addplot [mark=o, mark options=solid, blue, dashed] table[x=MeshSize,y=relE_L2Theta] {results/ecddr-einstein-wp/kasner/stab_1/Voro-small-0_k1/data_rates-1.dat};
        \addplot [mark=square, mark options=solid, green, dashed] table[x=MeshSize,y=relE_L2B] {results/ecddr-einstein-wp/kasner/stab_1/Voro-small-0_k1/data_rates-1.dat};
        \end{loglogaxis} 
      \end{tikzpicture}
    \subcaption{Voronoi meshes (Figure \ref{fig:mesh.voro})}
  \end{minipage}
  \captionsetup{width=0.85\textwidth}
  \caption{Kasner solution for the three-field scheme \eqref{eq:scheme2}: Plot of discrete total relative error \eqref{eq:dtre} for $D$, $\theta$, $B$ (vertical axis) against mesh size $h$ (horizontal axis) for $r=0$, $1$ polynomial degree ECDDR spaces (stabilisation parameter $\rho=1$ \eqref{def:L2p})}
  \label{fig:kas.sch2}
\end{figure}

The Gowdy wave solution is
\begin{equation*}
	\ed s^2 = -t^{-\frac12}\exp^{\frac{\lambda}{2}}\ed t^2 + t\exp^{P}\ed x^2 + t\exp^{-P}\ed y^2 + t^{-\frac12}\exp^{\frac{\lambda}{2}}\ed z^2
\end{equation*}
%\begin{align*}
%\theta &=
%\begin{bmatrix}
%    t^\frac12\exp^{\frac{P}{2}} & 0 & 0 \\
%    0 & t^\frac12\exp^{-\frac{P}{2}} & 0 \\
%    0 & 0 &  t^{-\frac14}\exp^{\frac{\lambda}{4}}
%\end{bmatrix}
%&& \\
%N &=t^{-\frac14}\exp^{\frac{\lambda}{4}}
%&
%E^0 &= \begin{bmatrix} 0 & 0 & -\pi^2 t \sin(4\pi z)J_0(2\pi t)J_1(2\pi t)\end{bmatrix} \\
%B &=
%\begin{bmatrix}
%    0 & B_{21} & 0 \\
%    0 & 0 & -B_{21} \\
%    0 & 0 & 0
%\end{bmatrix}
%&
%H &=
%\begin{bmatrix}
%    0 & H_{21} & 0 \\
%    H_{12} & 0 & 0 \\
%    0 & 0 & 0
%\end{bmatrix} \\
%D &=
%\begin{bmatrix}
%    0 & 0 & D_{31} \\
%    0 & D_{22} & 0 \\
%    D_{13} & 0 & 0
%\end{bmatrix}
%&
%E &=
%\begin{bmatrix}
%    E_{11} & 0 & 0 \\
%    0 & E_{22} & 0 \\
%    0 & 0 & E_{33}
%\end{bmatrix}
%\end{align*}
where
\begin{align*}
P ={}&J_0(2\pi t)\cos(2\pi z), \\
\lambda ={}&-2\pi t J_0(2\pi t)J_1(2\pi t)\cos^2(2\pi z)+2\pi^2 t^2\big(J^2_0(2\pi t)+J^2_1(2\pi t)\big)\\
{}&-\frac12 \Big((2\pi)^2\big(J^2_0(2\pi) +J^2_1(2\pi)\big) - 2\pi J_0(2\pi)J_1(2\pi)\Big),
\end{align*}
and $J_0$, $J_1$ are the Bessel functions of the first kind. The simulations are run with lapse $N=t^{-\frac14}\exp^{\frac{\lambda}{4}}$ to get Figure \ref{fig:gow.sch1} (Scheme \eqref{eq:scheme1}) and Figure \ref{fig:gow.sch2} (Scheme \eqref{eq:scheme2}), measured at the final time. 

The ideal convergence of $r+1$, $r$ being the degree of the discrete spaces, is seen in Figure \ref{fig:kas.sch1}, likely due to the simplicity of the Kasner solution, but not in Figure \ref{fig:kas.sch2} for $D^i$ and $B^i$. The lower convergence of $B^i$ is not unexpected; since $\theta^i$ is theoretically approximated with a rate of $r+1$, its derivative $B^i$ should be approximated at an order $r$; since the evolution of $B^i$ \eqref{eq:form.ev2.3} is just the exterior derivative of the evolution of $\theta^i$ \eqref{eq:form.ev2.2}, the extra equation \eqref{eq:scheme2.3} in the scheme does not in fact bring any new information that is not already informed by the relation $B^i=\sed\theta^i$. The convergence of $D^i$ however is anticipated, but not seen for \eqref{eq:scheme1}. We noticed that this behaviour can be modified by changing the stabilisation parameter $\rho$ in the $L^2$-product (see \eqref{def:L2p}), that regulates the influence of the stabilisation form on the product. Increasing $\rho$ seems to have a positive effect in this case, and indeed past a certain threshold gives convergence to $D^i$ (see Figure \ref{fig:stab.kas2} and \ref{fig:stab.kas2.1} for a comparison of different stabilisations in the case $r=0$ with both the discrete and continuous norms). For the Kasner problem, this convergence can be determined from the results of the coarsest mesh, where there is a noticeable gap between the convergent and non convergent solutions, and this should be considered before running the entire sequence. In general, the ``best'' choice of stabilisation is problem-dependent; Figure \ref{fig:stab.gow2} and \ref{fig:stab.gow2.1} shows a comparison for the Gowdy wave solution in both the discrete and continuous norms, where the the choice of $\rho=1$ is quite acceptable, and larger $\rho$ leads to larger errors, while too small ($\rho=0.1$) is not ideal either. The a priori selection of good values for the stabilisations parameters in polytopal methods is still a relatively unknown area, and our findings show that there is interest in understanding more. For the Gowdy wave solution, we see convergence for both schemes that is more steady for \eqref{eq:scheme1} in Figure \ref{fig:gow.sch1}, and a rougher and lower rate for \eqref{eq:scheme2} on $D^i$ and $B^i$ in Figure \ref{fig:gow.sch2}.

Overall, even with comparatively coarse mesh sequences (Figure \ref{tab:mesh.stats}), and a low number of timesteps, we obtain decently small relative errors for the two-field scheme \eqref{eq:scheme1}. These tests indicate that \eqref{eq:scheme1} is more stable and robust than \eqref{eq:scheme2}, and possibly the addition or treatment of the equation \eqref{eq:scheme2.3} is problematic.  This might indicate that the model \eqref{eq:weak.ev2} is not suitable for direct discretisation, and that a modified version might have to be considered. On the contrary, even though it is formally identical, \eqref{eq:weak.ev1} seems to be a better choice for building a numerical scheme. The analysis of these models and a better understanding of their properties that makes them, or not, suitable for schemes is an interesting avenue for future research.

\subsection{Constraint preservation}
Theorem \ref{prop:const.pres} is tested for the tetrahedral mesh sequence for $r=0$ with initial condition and lapse given by those of the Gowdy wave solution. We impose homogeneous natural boundary conditions on $E^i$, so that following the proof, the discrete constraints $\mathfrak{C}^i_1(n,\underline{u}_h)$, $\mathfrak{C}^i_2(n,\underline{u}_h)$ stay constant in the evolution. The results, visualised in Figure \ref{fig:constraint}, confirms that both constraints are preserved up to machine precision.

We plot the dual norm of the discrete constraint $\mathfrak{C}^\alpha_3(n,\underline{p}_h)$, proposed in Remark \ref{rem:C3} for the two-field scheme \eqref{eq:scheme1}, that measures the satisfaction of the Hamiltonian constraint \eqref{eq:form.cst.1}. Figure \ref{fig:constraint.2} shows the constraint norm at the final time for the Gowdy wave tests (Figure \ref{fig:gow.sch1}). In these results, the quantity decreases as the mesh is refined, or when the degree $r$ is increased, which supports the natural idea that a more accurate approximation should also better satisfy this discrete version of the constraint. The worst performing value, even when increasing the degree, corresponds to the constraint on $D^0$; this is further confirmed by Figure \ref{fig:constraint.3} that plots the evolution of each constraint from initial to final time. This result is, however, not surprising, since $D^0$ is the only field that is not explicitly evolved. Instead, it is reconstructed from the discrete exterior derivative of $\ulh{\theta}$ according to \eqref{eq:rel.D0}. A possible option to control $\mathfrak{C}^0_3(n,\underline{p}_h)$, which is necessary for investigations into long term simulations, is to simply add $D^0$ and its evolution explicity to the scheme, to obtain a better approximation that should theoretically lead to smaller constraint violations.

\begin{figure}\centering
	\ref{D.theta.3}
  \vspace{0.50cm}\\
  \begin{minipage}{0.45\textwidth}
    \begin{tikzpicture}[scale=0.85]
      \begin{loglogaxis} [legend columns=2, legend to name=D.theta.3]
        \logLogSlopeTriangle{0.90}{0.4}{0.1}{1}{black}
        \logLogSlopeTriangle{0.90}{0.4}{0.1}{2}{black}
        \addplot [mark=star, red] table[x=MeshSize,y=relE_L2D] {results/ecddr-einstein/gowdy/Tetgen-Cube-0_k0/data_rates-1.dat};
        \addlegendentry{${\rm E_{disc}}(D)$, $r=0$;}
        \addplot [mark=o, blue] table[x=MeshSize,y=relE_L2Theta] {results/ecddr-einstein/gowdy/Tetgen-Cube-0_k0/data_rates-1.dat};
         \addlegendentry{${\rm E_{disc}}(\theta)$, $r=0$;}
        \addplot [mark=star, mark options=solid, red, dashed] table[x=MeshSize,y=relE_L2D] {results/ecddr-einstein/gowdy/Tetgen-Cube-0_k1/data_rates-1.dat};
        \addlegendentry{${\rm E_{disc}}(D)$, $r=1$;}
        \addplot [mark=o, mark options=solid, blue, dashed] table[x=MeshSize,y=relE_L2Theta] {results/ecddr-einstein/gowdy/Tetgen-Cube-0_k1/data_rates-1.dat};
        \addlegendentry{${\rm E_{disc}}(\theta)$, $r=1$;}
      \end{loglogaxis}            
    \end{tikzpicture}
    \subcaption{Tetrahedral meshes (Figure \ref{fig:mesh.tets})}
  \end{minipage}
  \begin{minipage}{0.45\textwidth}
    \begin{tikzpicture}[scale=0.85] 
      \begin{loglogaxis}
        \logLogSlopeTriangle{0.90}{0.4}{0.1}{1}{black}
        \logLogSlopeTriangle{0.90}{0.4}{0.1}{2}{black}
        \addplot [mark=star, red] table[x=MeshSize,y=relE_L2D] {results/ecddr-einstein/gowdy/Voro-small-0_k0/data_rates-1.dat};
        \addplot [mark=o, blue] table[x=MeshSize,y=relE_L2Theta] {results/ecddr-einstein/gowdy/Voro-small-0_k0/data_rates-1.dat};
        \addplot [mark=star, mark options=solid, red, dashed] table[x=MeshSize,y=relE_L2D] {results/ecddr-einstein/gowdy/Voro-small-0_k1/data_rates-1.dat};
        \addplot [mark=o, mark options=solid, blue, dashed] table[x=MeshSize,y=relE_L2Theta] {results/ecddr-einstein/gowdy/Voro-small-0_k1/data_rates-1.dat};
        \end{loglogaxis}
      \end{tikzpicture}
    \subcaption{Voronoi meshes (Figure \ref{fig:mesh.voro})}
  \end{minipage}
  \captionsetup{width=0.85\textwidth}
  \caption{Gowdy wave solution for the two-field scheme \eqref{eq:scheme1}: Plot of discrete total relative error \eqref{eq:dtre} for $D$, $\theta$, (vertical axis) against mesh size $h$ (horizontal axis) for $r=0$, $1$ polynomial degree ECDDR spaces (stabilisation parameter $\rho=1$ \eqref{def:L2p})}
  \label{fig:gow.sch1}
\end{figure}

\begin{figure}\centering
	\ref{gow.sch2}
  \vspace{0.50cm}\\
  \begin{minipage}{0.45\textwidth}
    \begin{tikzpicture}[scale=0.85]
      \begin{loglogaxis} [legend columns=3, legend to name=gow.sch2]
        \logLogSlopeTriangle{0.90}{0.4}{0.1}{1}{black}
        \logLogSlopeTriangle{0.90}{0.4}{0.1}{2}{black}
        \addplot [mark=star, red] table[x=MeshSize,y=relE_L2D] {results/ecddr-einstein-wp/gowdy/stab_1/Tetgen-Cube-0_k0/data_rates-1.dat};
        \addlegendentry{${\rm E_{disc}}(D)$, $r=0$;}
        \addplot [mark=o, blue] table[x=MeshSize,y=relE_L2Theta] {results/ecddr-einstein-wp/gowdy/stab_1/Tetgen-Cube-0_k0/data_rates-1.dat};
         \addlegendentry{${\rm E_{disc}}(\theta)$, $r=0$;}
         \addplot [mark=square, green] table[x=MeshSize,y=relE_L2B] {results/ecddr-einstein-wp/gowdy/stab_1/Tetgen-Cube-0_k0/data_rates-1.dat};
         \addlegendentry{${\rm E_{disc}}(B)$, $r=0$;}
         \addplot [mark=star, mark options=solid, red, dashed] table[x=MeshSize,y=relE_L2D] {results/ecddr-einstein-wp/gowdy/stab_1/Tetgen-Cube-0_k1/data_rates-1.dat};
         \addlegendentry{${\rm E_{disc}}(D)$, $r=1$;}
         \addplot [mark=o, mark options=solid, blue, dashed] table[x=MeshSize,y=relE_L2Theta] {results/ecddr-einstein-wp/gowdy/stab_1/Tetgen-Cube-0_k1/data_rates-1.dat};
         \addlegendentry{${\rm E_{disc}}(\theta)$, $r=1$;}
         \addplot [mark=square, mark options=solid, green, dashed] table[x=MeshSize,y=relE_L2B] {results/ecddr-einstein-wp/gowdy/stab_1/Tetgen-Cube-0_k1/data_rates-1.dat};
         \addlegendentry{${\rm E_{disc}}(B)$, $r=1$;}
      \end{loglogaxis}            
    \end{tikzpicture}
    \subcaption{Tetrahedral meshes (Figure \ref{fig:mesh.tets})}
  \end{minipage}
  \begin{minipage}{0.45\textwidth}
    \begin{tikzpicture}[scale=0.85] 
      \begin{loglogaxis}
        \logLogSlopeTriangle{0.90}{0.4}{0.1}{1}{black}
        \logLogSlopeTriangle{0.90}{0.4}{0.1}{2}{black} 
        \addplot [mark=star, red] table[x=MeshSize,y=relE_L2D] {results/ecddr-einstein-wp/gowdy/stab_1/Voro-small-0_k0/data_rates-1.dat};
        \addplot [mark=o, blue] table[x=MeshSize,y=relE_L2Theta] {results/ecddr-einstein-wp/gowdy/stab_1/Voro-small-0_k0/data_rates-1.dat};
        \addplot [mark=square, green] table[x=MeshSize,y=relE_L2B] {results/ecddr-einstein-wp/gowdy/stab_1/Voro-small-0_k0/data_rates-1.dat};
        \addplot [mark=star, mark options=solid, red, dashed] table[x=MeshSize,y=relE_L2D] {results/ecddr-einstein-wp/gowdy/stab_1/Voro-small-0_k1/data_rates-1.dat};
        \addplot [mark=o, mark options=solid, blue, dashed] table[x=MeshSize,y=relE_L2Theta] {results/ecddr-einstein-wp/gowdy/stab_1/Voro-small-0_k1/data_rates-1.dat};
        \addplot [mark=square, mark options=solid, green, dashed] table[x=MeshSize,y=relE_L2B] {results/ecddr-einstein-wp/gowdy/stab_1/Voro-small-0_k1/data_rates-1.dat};
        \end{loglogaxis} 
      \end{tikzpicture}
    \subcaption{Voronoi meshes (Figure \ref{fig:mesh.voro})}
  \end{minipage}
  \captionsetup{width=0.85\textwidth}
  \caption{Gowdy wave solution for the three-field scheme \eqref{eq:scheme2}: Plot of discrete total relative error \eqref{eq:dtre} for $D$, $\theta$, $B$, (vertical axis) against mesh size $h$ (horizontal axis) for $r=0$, $1$ polynomial degree ECDDR spaces (stabilisation parameter $\rho=1$ \eqref{def:L2p})}
  \label{fig:gow.sch2}
\end{figure}

\begin{figure}\centering
	\ref{stab_par1}
  \vspace{0.50cm}\\
  \begin{minipage}{0.45\textwidth}
    \begin{tikzpicture}[scale=0.85]
      \begin{loglogaxis} [legend columns=4, legend to name=stab_par1]
        \logLogSlopeTriangle{0.90}{0.4}{0.1}{1}{black}
        \addplot [mark=star, red] table[x=MeshSize,y=relE_L2D] {results/ecddr-einstein-wp/kasner/stab_1/Tetgen-Cube-0_k0/data_rates-1.dat};
        \addlegendentry{$\rho=1$;}
        \addplot [mark=o, blue] table[x=MeshSize,y=relE_L2D] {results/ecddr-einstein-wp/kasner/stab_10/Tetgen-Cube-0_k0/data_rates-1.dat};
        \addlegendentry{$\rho=10$;}
        \addplot [mark=square, green] table[x=MeshSize,y=relE_L2D] {results/ecddr-einstein-wp/kasner/stab_20/Tetgen-Cube-0_k0/data_rates-1.dat};
        \addlegendentry{$\rho=20$;}
        \addplot [mark=triangle, orange] table[x=MeshSize,y=relE_L2D] {results/ecddr-einstein-wp/kasner/stab_30/Tetgen-Cube-0_k0/data_rates-1.dat};
        \addlegendentry{$\rho=30$;}
      \end{loglogaxis}
    \end{tikzpicture}
    \subcaption{Discrete error on $D$, tetrahedral meshes}
  \end{minipage}
  \begin{minipage}{0.45\textwidth}
    \begin{tikzpicture}[scale=0.85] 
      \begin{loglogaxis}
        \logLogSlopeTriangle{0.90}{0.4}{0.1}{1}{black} 
%        \addplot [mark=o, mark options=solid, blue, dashed] table[x=MeshSize,y=relE_L2D] {results/ecddr-einstein-wp/kasner/stab_50/Voro-small-0_k0/data_rates-1.dat};
%        \addlegendentry{$50$;}
				\addplot [mark=star, red] table[x=MeshSize,y=relE_L2D] {results/ecddr-einstein-wp/kasner/stab_1/Voro-small-0_k0/data_rates-1.dat};
        \addplot [mark=o, blue] table[x=MeshSize,y=relE_L2D] {results/ecddr-einstein-wp/kasner/stab_10/Voro-small-0_k0/data_rates-1.dat};
        \addplot [mark=square, green] table[x=MeshSize,y=relE_L2D] {results/ecddr-einstein-wp/kasner/stab_20/Voro-small-0_k0/data_rates-1.dat};
        \addplot [mark=triangle, orange] table[x=MeshSize,y=relE_L2D] {results/ecddr-einstein-wp/kasner/stab_30/Voro-small-0_k0/data_rates-1.dat};
        \end{loglogaxis} 
      \end{tikzpicture}
    \subcaption{Discrete error on $D$, Voronoi meshes}
  \end{minipage}\\[0.5em]
    \begin{minipage}{0.45\textwidth}
    \begin{tikzpicture}[scale=0.85]
      \begin{loglogaxis}
        \logLogSlopeTriangle{0.90}{0.4}{0.1}{1}{black}
        \addplot [mark=star, red] table[x=MeshSize,y=relE_L2Theta] {results/ecddr-einstein-wp/kasner/stab_1/Tetgen-Cube-0_k0/data_rates-1.dat};
        \addplot [mark=o, blue] table[x=MeshSize,y=relE_L2Theta] {results/ecddr-einstein-wp/kasner/stab_10/Tetgen-Cube-0_k0/data_rates-1.dat};
        \addplot [mark=square, green] table[x=MeshSize,y=relE_L2Theta] {results/ecddr-einstein-wp/kasner/stab_20/Tetgen-Cube-0_k0/data_rates-1.dat};
        \addplot [mark=triangle, orange] table[x=MeshSize,y=relE_L2Theta] {results/ecddr-einstein-wp/kasner/stab_30/Tetgen-Cube-0_k0/data_rates-1.dat};
      \end{loglogaxis}
    \end{tikzpicture}
    \subcaption{Discrete error on $\theta$, tetrahedral meshes}
  \end{minipage}
  \begin{minipage}{0.45\textwidth}
    \begin{tikzpicture}[scale=0.85] 
      \begin{loglogaxis}
        \logLogSlopeTriangle{0.90}{0.4}{0.1}{1}{black} 
%        \addplot [mark=o, mark options=solid, blue, dashed] table[x=MeshSize,y=relE_L2Theta] {results/ecddr-einstein-wp/kasner/stab_50/Voro-small-0_k0/data_rates-1.dat};
%        \addlegendentry{$50$;}
				\addplot [mark=star, red] table[x=MeshSize,y=relE_L2Theta] {results/ecddr-einstein-wp/kasner/stab_1/Voro-small-0_k0/data_rates-1.dat};
        \addplot [mark=o, blue] table[x=MeshSize,y=relE_L2Theta] {results/ecddr-einstein-wp/kasner/stab_10/Voro-small-0_k0/data_rates-1.dat};
        \addplot [mark=square, green] table[x=MeshSize,y=relE_L2Theta] {results/ecddr-einstein-wp/kasner/stab_20/Voro-small-0_k0/data_rates-1.dat};
        \addplot [mark=triangle, orange] table[x=MeshSize,y=relE_L2Theta] {results/ecddr-einstein-wp/kasner/stab_30/Voro-small-0_k0/data_rates-1.dat};
        \end{loglogaxis} 
      \end{tikzpicture}
    \subcaption{Discrete error on $\theta$, Voronoi meshes}
  \end{minipage}\\[0.5em]
    \begin{minipage}{0.45\textwidth}
    \begin{tikzpicture}[scale=0.85]
      \begin{loglogaxis}
        \addplot [mark=star, red] table[x=MeshSize,y=relE_L2B] {results/ecddr-einstein-wp/kasner/stab_1/Tetgen-Cube-0_k0/data_rates-1.dat};
        \addplot [mark=o, blue] table[x=MeshSize,y=relE_L2B] {results/ecddr-einstein-wp/kasner/stab_10/Tetgen-Cube-0_k0/data_rates-1.dat};
        \addplot [mark=square, green] table[x=MeshSize,y=relE_L2B] {results/ecddr-einstein-wp/kasner/stab_20/Tetgen-Cube-0_k0/data_rates-1.dat};
        \addplot [mark=triangle, orange] table[x=MeshSize,y=relE_L2B] {results/ecddr-einstein-wp/kasner/stab_30/Tetgen-Cube-0_k0/data_rates-1.dat};
      \end{loglogaxis}
    \end{tikzpicture}
    \subcaption{Discrete error on $B$, tetrahedral meshes}
  \end{minipage}
  \begin{minipage}{0.45\textwidth}
    \begin{tikzpicture}[scale=0.85] 
      \begin{loglogaxis}
%        \addplot [mark=o, mark options=solid, blue, dashed] table[x=MeshSize,y=relE_L2B] {results/ecddr-einstein-wp/kasner/stab_50/Voro-small-0_k0/data_rates-1.dat};
%        \addlegendentry{$50$;}
				\addplot [mark=star, red] table[x=MeshSize,y=relE_L2B] {results/ecddr-einstein-wp/kasner/stab_1/Voro-small-0_k0/data_rates-1.dat};
        \addplot [mark=o, blue] table[x=MeshSize,y=relE_L2B] {results/ecddr-einstein-wp/kasner/stab_10/Voro-small-0_k0/data_rates-1.dat};
        \addplot [mark=square, green] table[x=MeshSize,y=relE_L2B] {results/ecddr-einstein-wp/kasner/stab_20/Voro-small-0_k0/data_rates-1.dat};
        \addplot [mark=triangle, orange] table[x=MeshSize,y=relE_L2B] {results/ecddr-einstein-wp/kasner/stab_30/Voro-small-0_k0/data_rates-1.dat};
        \end{loglogaxis} 
      \end{tikzpicture}
    \subcaption{Discrete error on $B$, Voronoi meshes}
  \end{minipage}
  \captionsetup{width=0.85\textwidth}
  \caption{Kasner solution for the three-field scheme \eqref{eq:scheme2}: Comparison of the discrete total relative error \eqref{eq:dtre} for $D$, $\theta$, $B$, (vertical axis) against mesh size $h$ (horizontal axis) for different values of the stabilisation parameter $\rho$ \eqref{def:L2p} (degree $r=0$ ECDDR spaces)}
  \label{fig:stab.kas2}
\end{figure}

\begin{figure}\centering
	\ref{stab_par1.1}
  \vspace{0.50cm}\\
  \begin{minipage}{0.45\textwidth}
    \begin{tikzpicture}[scale=0.85]
      \begin{loglogaxis} [legend columns=4, legend to name=stab_par1.1]
        \logLogSlopeTriangle{0.90}{0.4}{0.1}{1}{black}
        \addplot [mark=star, red] table[x=MeshSize,y=relContE_L2D] {results/ecddr-einstein-wp/kasner/stab_1/Tetgen-Cube-0_k0/data_rates-1.dat};
        \addlegendentry{$\rho=1$;}
        \addplot [mark=o, blue] table[x=MeshSize,y=relContE_L2D] {results/ecddr-einstein-wp/kasner/stab_10/Tetgen-Cube-0_k0/data_rates-1.dat};
        \addlegendentry{$\rho=10$;}
        \addplot [mark=square, green] table[x=MeshSize,y=relContE_L2D] {results/ecddr-einstein-wp/kasner/stab_20/Tetgen-Cube-0_k0/data_rates-1.dat};
        \addlegendentry{$\rho=20$;}
        \addplot [mark=triangle, orange] table[x=MeshSize,y=relContE_L2D] {results/ecddr-einstein-wp/kasner/stab_30/Tetgen-Cube-0_k0/data_rates-1.dat};
        \addlegendentry{$\rho=30$;}
      \end{loglogaxis}
    \end{tikzpicture}
    \subcaption{Continuous error on $D$, tetrahedral meshes}
  \end{minipage}
  \begin{minipage}{0.45\textwidth}
    \begin{tikzpicture}[scale=0.85] 
      \begin{loglogaxis}
        \logLogSlopeTriangle{0.90}{0.4}{0.1}{1}{black} 
%        \addplot [mark=o, mark options=solid, blue, dashed] table[x=MeshSize,y=relE_L2D] {results/ecddr-einstein-wp/kasner/stab_50/Voro-small-0_k0/data_rates-1.dat};
%        \addlegendentry{$50$;}
				\addplot [mark=star, red] table[x=MeshSize,y=relContE_L2D] {results/ecddr-einstein-wp/kasner/stab_1/Voro-small-0_k0/data_rates-1.dat};
        \addplot [mark=o, blue] table[x=MeshSize,y=relContE_L2D] {results/ecddr-einstein-wp/kasner/stab_10/Voro-small-0_k0/data_rates-1.dat};
        \addplot [mark=square, green] table[x=MeshSize,y=relContE_L2D] {results/ecddr-einstein-wp/kasner/stab_20/Voro-small-0_k0/data_rates-1.dat};
        \addplot [mark=triangle, orange] table[x=MeshSize,y=relContE_L2D] {results/ecddr-einstein-wp/kasner/stab_30/Voro-small-0_k0/data_rates-1.dat};
        \end{loglogaxis} 
      \end{tikzpicture}
    \subcaption{Continuous error on $D$, Voronoi meshes}
  \end{minipage}\\[0.5em]
    \begin{minipage}{0.45\textwidth}
    \begin{tikzpicture}[scale=0.85]
      \begin{loglogaxis}
        \logLogSlopeTriangle{0.90}{0.4}{0.1}{1}{black}
        \addplot [mark=star, red] table[x=MeshSize,y=relContE_L2Theta] {results/ecddr-einstein-wp/kasner/stab_1/Tetgen-Cube-0_k0/data_rates-1.dat};
        \addplot [mark=o, blue] table[x=MeshSize,y=relContE_L2Theta] {results/ecddr-einstein-wp/kasner/stab_10/Tetgen-Cube-0_k0/data_rates-1.dat};
        \addplot [mark=square, green] table[x=MeshSize,y=relContE_L2Theta] {results/ecddr-einstein-wp/kasner/stab_20/Tetgen-Cube-0_k0/data_rates-1.dat};
        \addplot [mark=triangle, orange] table[x=MeshSize,y=relContE_L2Theta] {results/ecddr-einstein-wp/kasner/stab_30/Tetgen-Cube-0_k0/data_rates-1.dat};
      \end{loglogaxis}
    \end{tikzpicture}
    \subcaption{Continuous error on $\theta$, tetrahedral meshes}
  \end{minipage}
  \begin{minipage}{0.45\textwidth}
    \begin{tikzpicture}[scale=0.85] 
      \begin{loglogaxis}
        \logLogSlopeTriangle{0.90}{0.4}{0.1}{1}{black} 
%        \addplot [mark=o, mark options=solid, blue, dashed] table[x=MeshSize,y=relE_L2Theta] {results/ecddr-einstein-wp/kasner/stab_50/Voro-small-0_k0/data_rates-1.dat};
%        \addlegendentry{$50$;}
				\addplot [mark=star, red] table[x=MeshSize,y=relContE_L2Theta] {results/ecddr-einstein-wp/kasner/stab_1/Voro-small-0_k0/data_rates-1.dat};
        \addplot [mark=o, blue] table[x=MeshSize,y=relContE_L2Theta] {results/ecddr-einstein-wp/kasner/stab_10/Voro-small-0_k0/data_rates-1.dat};
        \addplot [mark=square, green] table[x=MeshSize,y=relContE_L2Theta] {results/ecddr-einstein-wp/kasner/stab_20/Voro-small-0_k0/data_rates-1.dat};
        \addplot [mark=triangle, orange] table[x=MeshSize,y=relContE_L2Theta] {results/ecddr-einstein-wp/kasner/stab_30/Voro-small-0_k0/data_rates-1.dat};
        \end{loglogaxis} 
      \end{tikzpicture}
    \subcaption{Continuous error on $\theta$, Voronoi meshes}
  \end{minipage}\\[0.5em]
    \begin{minipage}{0.45\textwidth}
    \begin{tikzpicture}[scale=0.85]
      \begin{loglogaxis}
        \addplot [mark=star, red] table[x=MeshSize,y=relContE_L2B] {results/ecddr-einstein-wp/kasner/stab_1/Tetgen-Cube-0_k0/data_rates-1.dat};
        \addplot [mark=o, blue] table[x=MeshSize,y=relContE_L2B] {results/ecddr-einstein-wp/kasner/stab_10/Tetgen-Cube-0_k0/data_rates-1.dat};
        \addplot [mark=square, green] table[x=MeshSize,y=relContE_L2B] {results/ecddr-einstein-wp/kasner/stab_20/Tetgen-Cube-0_k0/data_rates-1.dat};
        \addplot [mark=triangle, orange] table[x=MeshSize,y=relContE_L2B] {results/ecddr-einstein-wp/kasner/stab_30/Tetgen-Cube-0_k0/data_rates-1.dat};
      \end{loglogaxis}
    \end{tikzpicture}
    \subcaption{Continuous error on $B$, tetrahedral meshes}
  \end{minipage}
  \begin{minipage}{0.45\textwidth}
    \begin{tikzpicture}[scale=0.85] 
      \begin{loglogaxis}
%        \addplot [mark=o, mark options=solid, blue, dashed] table[x=MeshSize,y=relE_L2B] {results/ecddr-einstein-wp/kasner/stab_50/Voro-small-0_k0/data_rates-1.dat};
%        \addlegendentry{$50$;}
        \addplot [mark=star, red] table[x=MeshSize,y=relContE_L2B] {results/ecddr-einstein-wp/kasner/stab_1/Voro-small-0_k0/data_rates-1.dat};
        \addplot [mark=o, blue] table[x=MeshSize,y=relContE_L2B] {results/ecddr-einstein-wp/kasner/stab_10/Voro-small-0_k0/data_rates-1.dat};
        \addplot [mark=square, green] table[x=MeshSize,y=relContE_L2B] {results/ecddr-einstein-wp/kasner/stab_20/Voro-small-0_k0/data_rates-1.dat};
        \addplot [mark=triangle, orange] table[x=MeshSize,y=relContE_L2B] {results/ecddr-einstein-wp/kasner/stab_30/Voro-small-0_k0/data_rates-1.dat};
        \end{loglogaxis} 
      \end{tikzpicture}
    \subcaption{Continuous error on $B$, Voronoi meshes}
  \end{minipage}
  \captionsetup{width=0.85\textwidth}
  \caption{Kasner solution for the three-field scheme \eqref{eq:scheme2}: Comparison of the continuous total relative error \eqref{eq:ctre} for $D$, $\theta$, $B$, (vertical axis) against mesh size $h$ (horizontal axis) for different values of the stabilisation parameter $\rho$ \eqref{def:L2p} (degree $r=0$ ECDDR spaces)}
  \label{fig:stab.kas2.1}
\end{figure}

\begin{figure}\centering
	\ref{stab_par}
  \vspace{0.50cm}\\
  \begin{minipage}{0.45\textwidth}
    \begin{tikzpicture}[scale=0.85]
      \begin{loglogaxis} [legend columns=3, legend to name=stab_par]
        \logLogSlopeTriangle{0.90}{0.4}{0.1}{1}{black}
        \addplot [mark=star, red] table[x=MeshSize,y=relE_L2D] {results/ecddr-einstein-wp/gowdy/stab_0.1/Tetgen-Cube-0_k0/data_rates-1.dat};
        \addlegendentry{$\rho=0.1$;}
        \addplot [mark=o, blue] table[x=MeshSize,y=relE_L2D] {results/ecddr-einstein-wp/gowdy/stab_0.5/Tetgen-Cube-0_k0/data_rates-1.dat};
        \addlegendentry{$\rho=0.5$;}
        \addplot [mark=square, green] table[x=MeshSize,y=relE_L2D] {results/ecddr-einstein-wp/gowdy/stab_1/Tetgen-Cube-0_k0/data_rates-1.dat};
        \addlegendentry{$\rho=1$;}
        \addlegendimage{mark=halfsquare*, magenta}
        \addlegendentry{$\rho=5$}
        \addplot [mark=triangle, orange] table[x=MeshSize,y=relE_L2D] {results/ecddr-einstein-wp/gowdy/stab_10/Tetgen-Cube-0_k0/data_rates-1.dat};
        \addlegendentry{$\rho=10$;}
        \addplot [mark=|, violet] table[x=MeshSize,y=relE_L2D] {results/ecddr-einstein-wp/gowdy/stab_30/Tetgen-Cube-0_k0/data_rates-1.dat};
      	\addlegendentry{$\rho=30$;}
      \end{loglogaxis}
    \end{tikzpicture}
\subcaption{Discrete error on $D$, tetrahedral meshes}
  \end{minipage}
  \begin{minipage}{0.45\textwidth}
    \begin{tikzpicture}[scale=0.85] 
      \begin{loglogaxis}
        \logLogSlopeTriangle{0.90}{0.4}{0.1}{1}{black} 
        \addplot [mark=halfsquare*, magenta] table[x=MeshSize,y=relE_L2D] {results/ecddr-einstein-wp/gowdy/stab_5/Voro-small-0_k0/data_rates-1.dat};
        \addplot [mark=star, red] table[x=MeshSize,y=relE_L2D] {results/ecddr-einstein-wp/gowdy/stab_0.1/Voro-small-0_k0/data_rates-1.dat};
        \addplot [mark=o, blue] table[x=MeshSize,y=relE_L2D] {results/ecddr-einstein-wp/gowdy/stab_0.5/Voro-small-0_k0/data_rates-1.dat};
        \addplot [mark=square, green] table[x=MeshSize,y=relE_L2D] {results/ecddr-einstein-wp/gowdy/stab_1/Voro-small-0_k0/data_rates-1.dat};
        \end{loglogaxis} 
      \end{tikzpicture}
    \subcaption{Discrete error on $D$, Voronoi meshes}
  \end{minipage}\\[0.5em]
    \begin{minipage}{0.45\textwidth}
    \begin{tikzpicture}[scale=0.85]
      \begin{loglogaxis} [legend columns=3]
        \logLogSlopeTriangle{0.90}{0.4}{0.1}{1}{black}
        \addplot [mark=star, red] table[x=MeshSize,y=relE_L2Theta] {results/ecddr-einstein-wp/gowdy/stab_0.1/Tetgen-Cube-0_k0/data_rates-1.dat};
        \addplot [mark=o, blue] table[x=MeshSize,y=relE_L2Theta] {results/ecddr-einstein-wp/gowdy/stab_0.5/Tetgen-Cube-0_k0/data_rates-1.dat};
        \addplot [mark=square, green] table[x=MeshSize,y=relE_L2Theta] {results/ecddr-einstein-wp/gowdy/stab_1/Tetgen-Cube-0_k0/data_rates-1.dat};
        \addplot [mark=triangle, orange] table[x=MeshSize,y=relE_L2Theta] {results/ecddr-einstein-wp/gowdy/stab_10/Tetgen-Cube-0_k0/data_rates-1.dat};
        \addplot [mark=|, violet] table[x=MeshSize,y=relE_L2Theta] {results/ecddr-einstein-wp/gowdy/stab_30/Tetgen-Cube-0_k0/data_rates-1.dat};
      \end{loglogaxis}
    \end{tikzpicture}
    \subcaption{Discrete error on $\theta$, tetrahedral meshes}
  \end{minipage}
  \begin{minipage}{0.45\textwidth}
    \begin{tikzpicture}[scale=0.85] 
      \begin{loglogaxis}
        \logLogSlopeTriangle{0.90}{0.4}{0.1}{1}{black}
        \addplot [mark=halfsquare*, magenta] table[x=MeshSize,y=relE_L2Theta] {results/ecddr-einstein-wp/gowdy/stab_5/Voro-small-0_k0/data_rates-1.dat};
        \addplot [mark=star, red] table[x=MeshSize,y=relE_L2Theta] {results/ecddr-einstein-wp/gowdy/stab_0.1/Voro-small-0_k0/data_rates-1.dat};
        \addplot [mark=o, blue] table[x=MeshSize,y=relE_L2Theta] {results/ecddr-einstein-wp/gowdy/stab_0.5/Voro-small-0_k0/data_rates-1.dat};
        \addplot [mark=square, green] table[x=MeshSize,y=relE_L2Theta] {results/ecddr-einstein-wp/gowdy/stab_1/Voro-small-0_k0/data_rates-1.dat};
        \end{loglogaxis} 
      \end{tikzpicture}
    \subcaption{Discrete error on $\theta$, Voronoi meshes}
  \end{minipage}\\[0.5em]
    \begin{minipage}{0.45\textwidth}
    \begin{tikzpicture}[scale=0.85]
      \begin{loglogaxis} [legend columns=3]
        \logLogSlopeTriangle{0.90}{0.4}{0.1}{1}{black}
        \addplot [mark=star, red] table[x=MeshSize,y=relE_L2B] {results/ecddr-einstein-wp/gowdy/stab_0.1/Tetgen-Cube-0_k0/data_rates-1.dat};
        \addplot [mark=o, blue] table[x=MeshSize,y=relE_L2B] {results/ecddr-einstein-wp/gowdy/stab_0.5/Tetgen-Cube-0_k0/data_rates-1.dat};
        \addplot [mark=square, green] table[x=MeshSize,y=relE_L2B] {results/ecddr-einstein-wp/gowdy/stab_1/Tetgen-Cube-0_k0/data_rates-1.dat};
        \addplot [mark=triangle, orange] table[x=MeshSize,y=relE_L2B] {results/ecddr-einstein-wp/gowdy/stab_10/Tetgen-Cube-0_k0/data_rates-1.dat};
        \addplot [mark=|, violet] table[x=MeshSize,y=relE_L2B] {results/ecddr-einstein-wp/gowdy/stab_30/Tetgen-Cube-0_k0/data_rates-1.dat};
      \end{loglogaxis}
    \end{tikzpicture}
    \subcaption{Discrete error on $B$, tetrahedral meshes}
  \end{minipage}
  \begin{minipage}{0.45\textwidth}
    \begin{tikzpicture}[scale=0.85] 
      \begin{loglogaxis}
        \logLogSlopeTriangle{0.90}{0.4}{0.1}{1}{black}
        \addplot [mark=halfsquare*, magenta] table[x=MeshSize,y=relE_L2B] {results/ecddr-einstein-wp/gowdy/stab_5/Voro-small-0_k0/data_rates-1.dat};
        \addplot [mark=star, red] table[x=MeshSize,y=relE_L2B] {results/ecddr-einstein-wp/gowdy/stab_0.1/Voro-small-0_k0/data_rates-1.dat};
        \addplot [mark=o, blue] table[x=MeshSize,y=relE_L2B] {results/ecddr-einstein-wp/gowdy/stab_0.5/Voro-small-0_k0/data_rates-1.dat};
        \addplot [mark=square, green] table[x=MeshSize,y=relE_L2B] {results/ecddr-einstein-wp/gowdy/stab_1/Voro-small-0_k0/data_rates-1.dat};
        \end{loglogaxis} 
      \end{tikzpicture}
    \subcaption{Discrete error on $B$, tetrahedral meshes}
  \end{minipage}
  \captionsetup{width=0.85\textwidth}
  \caption{Gowdy wave solution for the three-field scheme \eqref{eq:scheme2}: Comparison of the discrete total relative error \eqref{eq:dtre} for $D$, $\theta$, $B$, (vertical axis) against mesh size $h$ (horizontal axis) for different values of the stabilisation parameter $\rho$ \eqref{def:L2p} (degree $r=0$ ECDDR spaces)}
  \label{fig:stab.gow2}
\end{figure}

\begin{figure}\centering
	\ref{stab_par1.2}
  \vspace{0.50cm}\\
  \begin{minipage}{0.45\textwidth}
    \begin{tikzpicture}[scale=0.85]
      \begin{loglogaxis} [legend columns=3, legend to name=stab_par1.2]
        \logLogSlopeTriangle{0.90}{0.4}{0.1}{1}{black}
				\addplot [mark=star, red] table[x=MeshSize,y=relContE_L2D] {results/ecddr-einstein-wp/gowdy/stab_0.1/Tetgen-Cube-0_k0/data_rates-1.dat};
        \addlegendentry{$\rho=0.1$;}
				\addplot [mark=o, blue] table[x=MeshSize,y=relContE_L2D] {results/ecddr-einstein-wp/gowdy/stab_0.5/Tetgen-Cube-0_k0/data_rates-1.dat};
        \addlegendentry{$\rho=0.5$;}
				\addplot [mark=square, green] table[x=MeshSize,y=relContE_L2D] {results/ecddr-einstein-wp/gowdy/stab_1/Tetgen-Cube-0_k0/data_rates-1.dat};
        \addlegendentry{$\rho=1$;}
        \addlegendimage{mark=halfsquare*, magenta}
        \addlegendentry{$\rho=5$}
				\addplot [mark=triangle, orange] table[x=MeshSize,y=relContE_L2D] {results/ecddr-einstein-wp/gowdy/stab_10/Tetgen-Cube-0_k0/data_rates-1.dat};
        \addlegendentry{$\rho=10$;}
				\addplot [mark=|, violet] table[x=MeshSize,y=relContE_L2D] {results/ecddr-einstein-wp/gowdy/stab_30/Tetgen-Cube-0_k0/data_rates-1.dat};
      	\addlegendentry{$\rho=30$;}
      \end{loglogaxis}
    \end{tikzpicture}
    \subcaption{Continuous error on $D$, tetrahedral meshes}
  \end{minipage}
  \begin{minipage}{0.45\textwidth}
    \begin{tikzpicture}[scale=0.85] 
      \begin{loglogaxis}
        \logLogSlopeTriangle{0.90}{0.4}{0.1}{1}{black} 
        \addplot [mark=halfsquare*, magenta] table[x=MeshSize,y=relContE_L2D] {results/ecddr-einstein-wp/gowdy/stab_5/Voro-small-0_k0/data_rates-1.dat};
        \addplot [mark=star, red] table[x=MeshSize,y=relContE_L2D] {results/ecddr-einstein-wp/gowdy/stab_0.1/Voro-small-0_k0/data_rates-1.dat};
        \addplot [mark=o, blue] table[x=MeshSize,y=relContE_L2D] {results/ecddr-einstein-wp/gowdy/stab_0.5/Voro-small-0_k0/data_rates-1.dat};
        \addplot [mark=square, green] table[x=MeshSize,y=relContE_L2D] {results/ecddr-einstein-wp/gowdy/stab_1/Voro-small-0_k0/data_rates-1.dat};
        \end{loglogaxis} 
      \end{tikzpicture}
    \subcaption{Continuous error on $D$, Voronoi meshes}
  \end{minipage}\\[0.5em]
    \begin{minipage}{0.45\textwidth}
    \begin{tikzpicture}[scale=0.85]
      \begin{loglogaxis} [legend columns=3]
        \logLogSlopeTriangle{0.90}{0.4}{0.1}{1}{black}
        \addplot [mark=star, red] table[x=MeshSize,y=relContE_L2Theta] {results/ecddr-einstein-wp/gowdy/stab_0.1/Tetgen-Cube-0_k0/data_rates-1.dat};
        \addplot [mark=o, blue] table[x=MeshSize,y=relContE_L2Theta] {results/ecddr-einstein-wp/gowdy/stab_0.5/Tetgen-Cube-0_k0/data_rates-1.dat};
        \addplot [mark=square, green] table[x=MeshSize,y=relContE_L2Theta] {results/ecddr-einstein-wp/gowdy/stab_1/Tetgen-Cube-0_k0/data_rates-1.dat};
        \addplot [mark=triangle, orange] table[x=MeshSize,y=relContE_L2Theta] {results/ecddr-einstein-wp/gowdy/stab_10/Tetgen-Cube-0_k0/data_rates-1.dat};
        \addplot [mark=|, violet] table[x=MeshSize,y=relContE_L2Theta] {results/ecddr-einstein-wp/gowdy/stab_30/Tetgen-Cube-0_k0/data_rates-1.dat};
      \end{loglogaxis}
    \end{tikzpicture}
    \subcaption{Continuous error on $\theta$, tetrahedral meshes}
  \end{minipage}
  \begin{minipage}{0.45\textwidth}
    \begin{tikzpicture}[scale=0.85] 
      \begin{loglogaxis}
        \logLogSlopeTriangle{0.90}{0.4}{0.1}{1}{black}
        \addplot [mark=halfsquare*, magenta] table[x=MeshSize,y=relContE_L2Theta] {results/ecddr-einstein-wp/gowdy/stab_5/Voro-small-0_k0/data_rates-1.dat};
        \addplot [mark=star, red] table[x=MeshSize,y=relContE_L2Theta] {results/ecddr-einstein-wp/gowdy/stab_0.1/Voro-small-0_k0/data_rates-1.dat};
        \addplot [mark=o, blue] table[x=MeshSize,y=relContE_L2Theta] {results/ecddr-einstein-wp/gowdy/stab_0.5/Voro-small-0_k0/data_rates-1.dat};
        \addplot [mark=square, green] table[x=MeshSize,y=relContE_L2Theta] {results/ecddr-einstein-wp/gowdy/stab_1/Voro-small-0_k0/data_rates-1.dat};
        \end{loglogaxis} 
      \end{tikzpicture}
    \subcaption{Continuous error on $\theta$, Voronoi meshes}
  \end{minipage}\\[0.5em]
    \begin{minipage}{0.45\textwidth}
    \begin{tikzpicture}[scale=0.85]
      \begin{loglogaxis} [legend columns=3]
        \logLogSlopeTriangle{0.90}{0.4}{0.1}{1}{black}
        \addplot [mark=star, red] table[x=MeshSize,y=relContE_L2B] {results/ecddr-einstein-wp/gowdy/stab_0.1/Tetgen-Cube-0_k0/data_rates-1.dat};
        \addplot [mark=o, blue] table[x=MeshSize,y=relContE_L2B] {results/ecddr-einstein-wp/gowdy/stab_0.5/Tetgen-Cube-0_k0/data_rates-1.dat};
        \addplot [mark=square, green] table[x=MeshSize,y=relContE_L2B] {results/ecddr-einstein-wp/gowdy/stab_1/Tetgen-Cube-0_k0/data_rates-1.dat};
        \addplot [mark=triangle, orange] table[x=MeshSize,y=relContE_L2B] {results/ecddr-einstein-wp/gowdy/stab_10/Tetgen-Cube-0_k0/data_rates-1.dat};
        \addplot [mark=|, violet] table[x=MeshSize,y=relContE_L2B] {results/ecddr-einstein-wp/gowdy/stab_30/Tetgen-Cube-0_k0/data_rates-1.dat};
      \end{loglogaxis}
    \end{tikzpicture}
    \subcaption{Continuous error on $B$, tetrahedral meshes}
  \end{minipage}
  \begin{minipage}{0.45\textwidth}
    \begin{tikzpicture}[scale=0.85] 
      \begin{loglogaxis}
        \logLogSlopeTriangle{0.90}{0.4}{0.1}{1}{black}
        \addplot [mark=halfsquare*, magenta] table[x=MeshSize,y=relContE_L2B] {results/ecddr-einstein-wp/gowdy/stab_5/Voro-small-0_k0/data_rates-1.dat};
        \addplot [mark=star, red] table[x=MeshSize,y=relContE_L2B] {results/ecddr-einstein-wp/gowdy/stab_0.1/Voro-small-0_k0/data_rates-1.dat};
        \addplot [mark=o, blue] table[x=MeshSize,y=relContE_L2B] {results/ecddr-einstein-wp/gowdy/stab_0.5/Voro-small-0_k0/data_rates-1.dat};
        \addplot [mark=square, green] table[x=MeshSize,y=relContE_L2B] {results/ecddr-einstein-wp/gowdy/stab_1/Voro-small-0_k0/data_rates-1.dat};
        \end{loglogaxis} 
      \end{tikzpicture}
    \subcaption{Continuous error on $B$, Voronoi meshes}
  \end{minipage}
  \captionsetup{width=0.85\textwidth}
  \caption{Gowdy wave solution for the three-field scheme \eqref{eq:scheme2}: Comparison of the continuous total relative error \eqref{eq:ctre} for $D$, $\theta$, $B$, (vertical axis) against mesh size $h$ (horizontal axis) for different values of the stabilisation parameter $\rho$ \eqref{def:L2p} (degree $r=0$ ECDDR spaces)}
  \label{fig:stab.gow2.1}
\end{figure}

\begin{figure}
\centering
  \ref{constr}
  \vspace{0.50cm}\\
  \begin{minipage}{0.45\textwidth}
    \begin{tikzpicture}[scale=0.85]
      \begin{loglogaxis} [ymin=1e-17, ymax=1e-12, legend columns=3, legend to name=constr]
        \addplot [mark=star, red] table[x=MeshSize,y=cTheta] {results/ecddr-einstein-wp/constraint/Tetgen-Cube-0_k0/data_rates-1.dat};
        \addlegendentry{$i=1$;}
        \addplot [mark=o, blue] table[x=MeshSize,y=cTheta] {results/ecddr-einstein-wp/constraint/Tetgen-Cube-0_k0/data_rates-2.dat};
        \addlegendentry{$i=2$;}
         \addplot [mark=triangle, green] table[x=MeshSize,y=cTheta] {results/ecddr-einstein-wp/constraint/Tetgen-Cube-0_k0/data_rates-3.dat};
         \addlegendentry{$i=3$;}
      \end{loglogaxis}            
    \end{tikzpicture}
    \subcaption{$|\mathfrak{C}^i_1(n,\underline{u}_h)-\mathfrak{C}^i_1(0,\underline{u}_h)|$}
  \end{minipage}
  \begin{minipage}{0.45\textwidth}
    \begin{tikzpicture}[scale=0.85] 
      \begin{loglogaxis} [ymin=1e-17, ymax=1e-12]
        \addplot [mark=star, red] table[x=MeshSize,y=cB] {results/ecddr-einstein-wp/constraint/Tetgen-Cube-0_k0/data_rates-1.dat};
        \addplot [mark=o, blue] table[x=MeshSize,y=cB] {results/ecddr-einstein-wp/constraint/Tetgen-Cube-0_k0/data_rates-2.dat};
        \addplot [mark=triangle, green] table[x=MeshSize,y=cB] {results/ecddr-einstein-wp/constraint/Tetgen-Cube-0_k0/data_rates-3.dat};
        \end{loglogaxis} 
      \end{tikzpicture}
    \subcaption{$|\mathfrak{C}^i_2(n,\underline{u}_h)-\mathfrak{C}^i_2(0,\underline{u}_h)|$}
  \end{minipage}
  \captionsetup{width=0.85\textwidth}
  \caption{Three-field scheme \eqref{eq:scheme2} run with homogeneous natural boundary conditions to verify Proposition \ref{prop:const.pres}: Difference between the discrete constraint at initial and final time (vertical axis) against mesh size $h$ (horizontal axis) on the tetrahedral mesh sequence with $r=0$}
  \label{fig:constraint}
\end{figure}

\begin{figure}
\centering
  \ref{constr.C3}
  \vspace{0.50cm}\\
  \begin{minipage}{0.45\textwidth}
    \begin{tikzpicture}[scale=0.85]
      \begin{loglogaxis} [legend columns=4, legend to name=constr.C3]
        \addplot [mark=star, red] table[x=MeshSize,y=C3_H1norm0] {results/ecddr-einstein/gowdy/Tetgen-Cube-0_k0/constraint.dat};
        \addlegendentry{$\alpha=0$, $r=0$;}
        \addplot [mark=triangle, blue] table[x=MeshSize,y=C3_H1norm1] {results/ecddr-einstein/gowdy/Tetgen-Cube-0_k0/constraint.dat};
        \addlegendentry{$\alpha=1$, $r=0$;}
        \addplot [mark=o, orange] table[x=MeshSize,y=C3_H1norm2] {results/ecddr-einstein/gowdy/Tetgen-Cube-0_k0/constraint.dat};
         \addlegendentry{$\alpha=2$, $r=0$;}
        \addplot [mark=square, green] table[x=MeshSize,y=C3_H1norm3] {results/ecddr-einstein/gowdy/Tetgen-Cube-0_k0/constraint.dat};
        \addlegendentry{$\alpha=3$, $r=0$;}
        \addplot [mark=star, mark options=solid, red, dashed] table[x=MeshSize,y=C3_H1norm0] {results/ecddr-einstein/gowdy/Tetgen-Cube-0_k1/constraint.dat};
        \addlegendentry{$\alpha=0$, $r=1$;}
        \addplot [mark=triangle, mark options=solid, blue, dashed] table[x=MeshSize,y=C3_H1norm1] {results/ecddr-einstein/gowdy/Tetgen-Cube-0_k1/constraint.dat};
        \addlegendentry{$\alpha=1$, $r=1$;}
        \addplot [mark=o, mark options=solid, orange, dashed] table[x=MeshSize,y=C3_H1norm2] {results/ecddr-einstein/gowdy/Tetgen-Cube-0_k1/constraint.dat};
         \addlegendentry{$\alpha=2$, $r=1$;}
        \addplot [mark=square,  mark options=solid, green, dashed] table[x=MeshSize,y=C3_H1norm3] {results/ecddr-einstein/gowdy/Tetgen-Cube-0_k1/constraint.dat};
        \addlegendentry{$\alpha=3$, $r=1$;}
      \end{loglogaxis}            
    \end{tikzpicture}
    \subcaption{Tetrahedral meshes (Figure \ref{fig:mesh.tets})}
  \end{minipage}
  \begin{minipage}{0.45\textwidth}
    \begin{tikzpicture}[scale=0.85]
      \begin{loglogaxis}
        \addplot [mark=star, red] table[x=MeshSize,y=C3_H1norm0] {results/ecddr-einstein/gowdy/Voro-small-0_k0/constraint.dat};
        \addplot [mark=triangle, blue] table[x=MeshSize,y=C3_H1norm1] {results/ecddr-einstein/gowdy/Voro-small-0_k0/constraint.dat};
        \addplot [mark=o, orange] table[x=MeshSize,y=C3_H1norm2] {results/ecddr-einstein/gowdy/Voro-small-0_k0/constraint.dat};
        \addplot [mark=square, green] table[x=MeshSize,y=C3_H1norm3] {results/ecddr-einstein/gowdy/Voro-small-0_k0/constraint.dat};
        \addplot [mark=star, mark options=solid, red, dashed] table[x=MeshSize,y=C3_H1norm0] {results/ecddr-einstein/gowdy/Voro-small-0_k1/constraint.dat};
        \addplot [mark=triangle, mark options=solid, blue, dashed] table[x=MeshSize,y=C3_H1norm1] {results/ecddr-einstein/gowdy/Voro-small-0_k1/constraint.dat};
        \addplot [mark=o, mark options=solid, orange, dashed] table[x=MeshSize,y=C3_H1norm2] {results/ecddr-einstein/gowdy/Voro-small-0_k1/constraint.dat};
        \addplot [mark=square,  mark options=solid, green, dashed] table[x=MeshSize,y=C3_H1norm3] {results/ecddr-einstein/gowdy/Voro-small-0_k1/constraint.dat};
      \end{loglogaxis}            
    \end{tikzpicture}
    \subcaption{Voronoi meshes (Figure \ref{fig:mesh.voro})}
  \end{minipage}
  \captionsetup{width=0.85\textwidth}
  \caption{The Gowdy wave solution for the two-field scheme \eqref{eq:scheme1}: Plot of the discrete constraint $\mathfrak{C}^\alpha_3(n,\cdot)$ (see Remark \ref{rem:C3}) measured in the dual norm (vertical axis) against the mesh size $h$ (horizontal axis) at final time}
  \label{fig:constraint.2}
\end{figure}

\begin{figure}
\centering
  \ref{constr.C3.ev}
  \vspace{0.50cm}\\
  \begin{minipage}{0.45\textwidth}
    \begin{tikzpicture}[scale=0.85]
      \begin{axis} [ymode=log, ymin=1e-3, ymax=1e-1, legend columns=4, legend to name=constr.C3.ev]
        \addplot [mark=star, red] table[x=Timestep,y=C3_H1norm0] {results/ecddr-einstein/gowdy/Tetgen-Cube-0_k0/constraints-5.txt};
        \addlegendentry{$\alpha=0$;}
        \addplot [mark=triangle, blue] table[x=Timestep,y=C3_H1norm1] {results/ecddr-einstein/gowdy/Tetgen-Cube-0_k0/constraints-5.txt};
        \addlegendentry{$\alpha=1$;}
        \addplot [mark=o, orange] table[x=Timestep,y=C3_H1norm2] {results/ecddr-einstein/gowdy/Tetgen-Cube-0_k0/constraints-5.txt};
         \addlegendentry{$\alpha=2$;}
        \addplot [mark=square, green] table[x=Timestep,y=C3_H1norm3] {results/ecddr-einstein/gowdy/Tetgen-Cube-0_k0/constraints-5.txt};
        \addlegendentry{$\alpha=3$;}
      \end{axis}      
    \end{tikzpicture}
    \subcaption{Degree $r=0$ ECDDR scheme}
  \end{minipage}
  \begin{minipage}{0.45\textwidth}
    \begin{tikzpicture}[scale=0.85]
      \begin{axis} [ymode=log, ymin=1e-3, ymax=1e-1, mark repeat = 2]
        \addplot [mark=star, red] table[x=Timestep,y=C3_H1norm0] {results/ecddr-einstein/gowdy/Tetgen-Cube-0_k1/constraints-5.txt};
        \addplot [mark=triangle, blue] table[x=Timestep,y=C3_H1norm1] {results/ecddr-einstein/gowdy/Tetgen-Cube-0_k1/constraints-5.txt};
        \addplot [mark=o, orange] table[x=Timestep,y=C3_H1norm2] {results/ecddr-einstein/gowdy/Tetgen-Cube-0_k1/constraints-5.txt};
        \addplot [mark=square, green] table[x=Timestep,y=C3_H1norm3] {results/ecddr-einstein/gowdy/Tetgen-Cube-0_k1/constraints-5.txt};
      \end{axis}            
    \end{tikzpicture}
    \subcaption{Degree $r=1$ ECDDR scheme}
  \end{minipage}
  \captionsetup{width=0.85\textwidth}
  \caption{The Gowdy wave solution for the two-field scheme \eqref{eq:scheme1}: Evolution of the dual norm of the discrete constraint $\mathfrak{C}^\alpha_3(n,\cdot)$ (vertical axis) against the timestep (horizontal axis) on mesh 5 of the tetrahedral sequence (Figure \ref{fig:mesh.tets})}
  \label{fig:constraint.3}
\end{figure}

\section{Conclusion}

We introduce a exterior calculus approach to view the Einstein equations, building on the works of \cite{Fecko:97, Olivares.Peshkov.ea:22}, and write the decomposed $3+1$ equations in a way that preserves the exterior calculus operators. Then we designed and implemented numerical schemes for two resulting formulations using the exterior calculus discrete de Rham method. The novelty of ECDDR in the context of numerical relativity is that, to our knowledge, it is one of the first polytopal methods to be applied to Einstein's equations, meaning that the resulting discretisation works on meshes made up of general polytopes, opening up possibilities such as interesting mesh refinement techniques to deal with singularities. We prove that the two formulations are indeed equivalent to Einstein's equations, and the constraint propogation in the continuous case, as well as the exact propogation of a discrete version of an auxiliary constraint thanks to the complex properties of the discrete ECDDR complex. Results are provided for a $3$D spatial domain, showing convergence in two standard test cases, and the numerical preservation of the proved constraint quantity.

As a first exploration of polytopal methods for Einstein's equations, the results suggest that there is further work to be done to improve our understanding of numerical schemes based on these equations. Investigations include the implementation of periodic boundary conditions, the design of a more general scheme for non-zero shift, a more expansive panel of numerical tests, analysing the well-posedness of the formulations derived from this method, and the convergence and stability analysis of these schemes. These are the topics of an upcoming work.

\section*{Acknowledgements}
We would like to thank J\'{e}r\^{o}me Droniou for helpful discussions, as well as his advice and suggestions regarding the numerical aspects of the article.

Funded by the European Union (ERC Synergy, NEMESIS, project number 101115663).
Views and opinions expressed are however those of the authors only and do not necessarily reflect those of the European Union or the European Research Council Executive Agency. Neither the European Union nor the granting authority can be held responsible for them.

\appendix
\section{Formulas}
This is a collection of formulas and conventions used throughout the paper with proof or references where necessary. In the following, let $(M,g)$ be a an oriented $n$-manifold equipped with a metric of signature $(n-s,s)$, where $s$ is the number of $-1$'s in the diagonalisation. We denote by $\varepsilon$ the volume form induced by the metric and the orientation, and by $(e_i)_{i=1}^n$ a right-handed orthonormal basis of $TM$ (or $TU$ for some $U\subset M$) and $(\theta_i)_{i=1}^n$ its dual basis. 

\subsection{Orthonormal frame and coordinates}

Let $(e_\alpha)_{i\in[0,n]}$ be a right-handed $g$-orthonormal basis of $TM$ such that $e_0=\bvec{n}$, and $(\theta^\alpha)_{\alpha\in[0,n]}$ the dual basis, which satisfy the duality conditon
\begin{equation}
  \theta^\alpha(e_\beta)=\delta^\alpha_\beta.
\end{equation}
In this basis, the frame components of the metric and its inverse are given by
\begin{equation*}
g_{\alpha\beta}:=g(e_\alpha,e_\beta)= -\delta_{\alpha}^0\delta_{\beta}^0 +\delta_\alpha^i \delta_\beta^j \delta_{ij}
\quad \text{and} \quad 
g^{\alpha\beta}:=g(\theta^\alpha,\theta^\beta)= -\delta^{\alpha}_0\delta^{\beta}_0 +\delta^\alpha_i \delta^\beta_j \delta^{ij},
\end{equation*}
respectively, and the metric can be expresses as
\begin{equation*}
g=g_{\alpha\beta}\theta^\alpha \otimes \theta^\beta.
\end{equation*}
Unless stated otherwise, all formulas and equations are written in this basis.

\subsection{Change of basis}
Given two bases and dual $1$-forms $(e_\alpha)_{i\in[0,n]}$, $(\theta^\alpha)_{i\in[0,n]}$ and $(\tilde{e}_\alpha)_{i\in[0,n]}$, $(\tilde{\theta}^\alpha)_{i\in[0,n]}$, we can express the vectors of one basis in the other by
\begin{equation*}
  e_\alpha = (e_\alpha)^\beta\tilde{e}_\beta,\qquad \tilde{e}_\alpha = (\tilde{e}_\alpha)^\beta e_\beta,
\end{equation*}
where clearly the components $\left[(e_\alpha)^\beta\right]_{\alpha\beta}$ = $\left[(\tilde{e}_\alpha)^\beta\right]^{-1}_{\alpha\beta}$, and for the $1$-forms
\begin{equation*}
  \theta^\alpha = (\theta^\alpha)_\beta\tilde{\theta}^\beta,\qquad \tilde{\theta}^\alpha = (\tilde{\theta}^\alpha)_\beta \theta^\beta.
\end{equation*}
Note that
\begin{equation*}
  (\theta^\alpha)_\beta = \theta^\alpha(\tilde{e}_\beta)=\theta^\alpha((\tilde{e}_\beta)^\gamma e_\gamma)=(\tilde{e}_\beta)^\gamma\delta^\alpha_\gamma=(\tilde{e}_\beta)^\alpha
\end{equation*}
so $\left[(\theta^\alpha)_\beta\right]_{\alpha\beta}=\left[(\tilde{e}_\alpha)^\beta\right]_{\alpha\beta}^T$, where $T$ is the transpose, and vice versa.
Given a $(p,q)$-tensor $\omega$, we can extend the above calculation for the general formula
\begin{align*}
  \omega\indices{^{i_1\cdots i_p}_{j_1\cdots j_q}}
  &=\omega(\theta^{i_1},\cdots, \theta^{i_p} e_{j_1},\cdots,e_{j_q}) \\
  &=(\theta^{i_1})_{k_1}\cdots(\theta^{i_p})_{k_p}(e_{j_1})^{l_1}\cdots(e_{j_q})^{l_q}\tilde{\omega}\indices{^{k_1\cdots k_p}_{l_1\cdots l_q}}.
\end{align*}

\subsection{Differential forms}
Define the exterior (or wedge) product as the map $\wedge:\kform(M)\times\kform[l](M)\to\kform[k+l](M)$, where for all $\omega\in\kform(M),\mu\in\kform[l](M)$,
\begin{equation*}
  (\omega\wedge\mu)(v_1,\cdots,v_{k+l})=\frac{1}{k!l!}\sum_{\sigma\in S_{k+l}}\sign(\sigma)(\omega\otimes\mu)(v_{\sigma(1)},\cdots,v_{\sigma(k+l)}).
\end{equation*}
Let $\omega\in\kform(M)$ be a $k$-form. The components $\omega_{i_1\cdots i_k}$ are those such that
\begin{equation*}
  \omega = \frac{1}{k!}\omega_{i_1\cdots i_k}\theta^{i^1}\wedge \cdots\wedge \theta^{i^k},
\end{equation*}
which coincide with the tensor components of $\omega$, i.e.
\begin{equation*}
  \omega = \omega_{i_1\cdots i_k}\theta^{i^1}\otimes \cdots\otimes \theta^{i^k}.
\end{equation*}

\subsection{Anti-symmetrisation brackets, Levi-Civita symbol, Kronecker delta}
The anti-symmetrisation brackets are square brackets at the index level which anti-symmetrise the indices contained inside. For example,
\begin{equation*}
  \omega_{[i_1\cdots i_k]}=\frac{1}{k!}\sum_{\sigma\in S_k}\sign(\sigma)\omega_{\sigma(i_1)\cdots\sigma(i_k)}.
\end{equation*}
Note that if the indices are already antisymmetric, the brackets have no effect, and can be directly removed.

Let $\epsilon_{i_1i_2\cdots i_k}$ be the Levi-Civita symbol, where
\begin{equation*}
  \epsilon_{12\cdots k}\coloneq 1 \quad\mbox{ and }\quad 
  \epsilon_{i_1i_2\cdots i_k}=
  \begin{cases}
    +1 & \mbox{ if $i_1i_2\cdots i_k$ is an even permutation of $12\cdots k$}, \\
    -1 & \mbox{ if $i_1i_2\cdots i_k$ is an odd permutation of $12\cdots k$}, \\
    0 & \mbox{ otherwise}.
  \end{cases}
\end{equation*}
Raising the indices with the metric, we get the identity
\begin{equation}\label{eq:raised.eps}
  \epsilon^{12\cdots n}=\underbrace{g^{1i_1}g^{2i_2}\cdots g^{ni_n}\epsilon_{i_1i_2\cdots i_n}}_{\det(g^{ab})}=(-1)^s.
\end{equation}
Let $\delta^{i_1\cdots i_k}_{j_1\cdots j_k}$ be the Kronecker delta, defined as
\begin{equation*}
  \delta^{i_1\cdots i_k}_{j_1\cdots j_k}=
  \begin{cases}
    +1 &\mbox{if $j_1\cdots j_k$ distinct is an even permutation of $i_1\cdots i_k$,}\\
    -1 &\mbox{if $j_1\cdots j_k$ distinct is an odd permutation of $i_1\cdots i_k$},\\
    0 &\mbox{otherwise.}
  \end{cases}
\end{equation*}
Then we have the identities
\begin{align}\label{eq:kron.delta}
    \epsilon^{i_1\cdots i_kj_{k+1}\cdots j_n}\epsilon_{l_1\cdots l_kj_{k+1}\cdots j_n}=(-1)^s(n-k)!\delta^{i_1\cdots i_k}_{l_1\cdots l_k}, \\
    \delta^{i_1\cdots i_k}_{j_1\cdots j_k} =k!\delta^{i_1}_{[j_1}\cdots\delta^{i_k}_{j_k]}=k!\delta^{[i_1}_{j_1}\cdots\delta^{i_k]}_{j_k}.
\end{align}
In the right-handed orthonormal basis, the components of the volume form are exactly the Levi-Civita symbols:
\begin{equation}
  \varepsilon = \theta^1\wedge\cdots\wedge \theta^n = \frac{1}{n!}\epsilon_{i_1\cdots i_n}\theta^{i_1} \wedge\cdots\wedge\theta^{i_n},
\end{equation}
thus the above formulas hold when replacing $\epsilon$ by $\varepsilon$.

\subsection{Hodge star operator}\label{def:hodge.star}
Define the Hodge star operator as the unique linear operator $\star:\kform(M)\to\kform[n-k](M)$, such that
\begin{equation*}
  \omega\wedge\star\mu = (\omega|\mu) \varepsilon \qquad \forall\omega,\mu\in\kform(M),
\end{equation*}
where at each $p\in M$, $(\cdot|\cdot)_p$ is the inner product on $\kform(T_pM)$ given by:
\begin{equation*}
  (\omega|\mu)_g = \omega_{i_1\cdots i_k}\mu^{i_k\cdots i_k}\qquad \forall\omega,\mu\in\kform(M).
\end{equation*}
By this definition, we see that (note that $(i_x)_{x=1}^k$ below are fixed indices, so no Einstein summation is performed over them)
\begin{align*}
  \theta^{i_1}\wedge\cdots\theta^{i_k}\wedge\star(\theta^{i_1}\wedge\cdots\theta^{i_k})
  &=g^{i_1i_1}\cdots g^{i_ki_k}\frac{1}{n!}\varepsilon_{j_1\cdots j_n}\theta^{j_1} \wedge\cdots\wedge\theta^{j_n} \\
  &=\theta^{i_1} \wedge\cdots\wedge\theta^{i_k}\wedge(g^{i_1i_1}\cdots g^{i_ki_k}\frac{1}{(n-k)!}\varepsilon_{i_1\cdots i_k j_{k+1}\cdots j_n}\theta^{j_{k+1}} \wedge\cdots\wedge\theta^{j_n})\\
  &=\theta^{i_1} \wedge\cdots\wedge\theta^{i_k}\wedge(\frac{1}{(n-k)!}\varepsilon\indices{^{i_1\cdots i_k}_{j_{k+1}\cdots j_n}}\theta^{j_{k+1}} \wedge\cdots\wedge\theta^{j_n}),
\end{align*}
where the last line uses the diagonality of $g^{ab}$ to conclude 
\begin{equation*}
  g^{i_1i_1}\cdots g^{i_ki_k}\varepsilon_{i_1\cdots i_k j_{k+1}\cdots j_n}=g^{i_1l_1}\cdots g^{i_kl_k}\varepsilon_{l_1\cdots l_k i_2\cdots i_k j_{k+1}\cdots j_n}=\varepsilon\indices{^{i_1\cdots i_k}_{j_{k+1}\cdots j_n}}.
\end{equation*}
Hence
\begin{equation}
  \star(\theta^{i_1}\wedge\cdots\wedge\theta^{i_k})=\frac{1}{(n-k)!}\varepsilon\indices{^{i_1\cdots i_k}_{j_{k+1}\cdots j_n}}\theta^{j_{k+1}} \wedge\cdots\wedge\theta^{j_n},
\end{equation}
and the formula for a general $k$-form $\omega$ is
\begin{equation}
  \star\omega = \frac{1}{k!(n-k)!}\varepsilon\indices{^{i_1\cdots i_k}_{j_{k+1}\cdots j_n}} \omega_{i_1\cdots i_k}\theta^{j_{k+1}} \wedge\cdots\wedge\theta^{j_n},
\end{equation}
or equivalently in components
\begin{equation}
  (\star\omega)_{j_{k+1}\cdots j_n} = \frac{1}{k!}\varepsilon\indices{^{i_1\cdots i_k}_{j_{k+1}\cdots j_n}} \omega_{i_1\cdots i_k}.
\end{equation}
Applying the Hodge star twice leads to
\begin{align*}
  (\star\star\omega)_{l_1\cdots l_k}&=\frac{1}{(n-k)!}\varepsilon\indices{^{j_{k+1}\cdots j_{n}}_{l_{1}\cdots l_k}}(\frac{1}{k!}\varepsilon\indices{^{i_1\cdots i_k}_{j_{k+1}\cdots j_n}} \omega_{i_1\cdots i_k})\\
  &=\frac{(-1)^{k(n-k)}}{(n-k)!k!}\varepsilon\indices{_{l_{1}\cdots l_k j_{k+1}\cdots j_{n}}}\varepsilon\indices{^{i_1\cdots i_k j_{k+1}\cdots j_n}} \omega_{i_1\cdots i_k} \\
  &=(-1)^{k(n-k)+s}\omega_{l_1\cdots l_k}
\end{align*}
where we shuffle indices to get the second line, and use \eqref{eq:kron.delta} for the last line. In other words
\begin{equation}\label{eq:double.hodge}
  \star\star\omega=(-1)^{k(n-k)+s}\omega.
\end{equation}  

\subsection{Components of the Lie derivative}

Let $X$ be a vector field and $T$ a $(p,q)$-tensor. By definition, $\ld{X}e_i = [X,e_i] =X^j[e_j,e_i]-e_i(X^j)e_j=X^jC\indices{^k_j_i}e_k-e_i(X^j)e_j$, where $C\indices{^k_i_j}$ are the commutation coefficients, defined so that $C\indices{^k_i_j}e_k\coloneq [e_i, e_j]$. For general $T$, the Lie derivative satisfies
\begin{align*}
	\ld{X}(T(\theta^{i_1},\cdots,\theta^{i_p},e_{j_1},\cdots,e_{j_q}))
	={}& \ld{X}T(\theta^{i_1},\cdots,\theta^{i_p},e_{j_1},\cdots,e_{j_q}) +T(\ld{X}\theta^{i_1},\cdots,\theta^{i_p},e_{j_1},\cdots,e_{j_q})\\
	& +\cdots +T(\theta^{i_1},\cdots,\theta^{i_p},e_{j_1},\cdots,\ld{X}e_{j_q}).
\end{align*}
Applied to $\theta^i$, we see
\begin{equation*}
	\ld{X}(\theta^i(e_j))
	= \ld{X}\theta^i(e_j)+\theta^i(X^lC\indices{^k_l_j}e_k-e_j(X^l)e_l)=\ld{X}\theta^i(e_j)+X^lC\indices{^i_l_j}-e_j(X^i),
\end{equation*}
or noticing the RHS is 0, 
\begin{equation*}
	\ld{X}\theta^i(e_j)=-X^lC\indices{^i_l_j}+e_j(X^i).
\end{equation*}
The general formula for $T$ is then
\begin{align*}
	\ld{X}T\indices{^{i_1\cdots i_p}_{j_1\cdots j_q}}
	={}& X\big(T\indices{^{i_1\cdots i_p}_{j_1\cdots j_q}}\big)+T\indices{^{ki_2\cdots i_p}_{j_1\cdots j_q}}(X^lC\indices{^{i_1}_l_k}-e_k(X^{i_1}))\\
	& +\cdots +T\indices{^{i_1\cdots i_p}_{j_1\cdots j_{q-1}k}}(-X^l C\indices{^k_l_{j_q}}+e_{j_q}(X^k)).
\end{align*}

\subsection{Symmetries of the Riemann curvature tensor}
\begin{subequations}
The Riemann curvature tensor is antisymmetric in the first and last pairs of indices:
\begin{equation}\label{eq:reim.sym.1}
  R_{\alpha\beta\mu\nu}=-R_{\beta\alpha\mu\nu}=-R_{\alpha\beta\nu\mu}
\end{equation}
and also symmetric after swapping the first and last pairs of indices
\begin{equation}\label{eq:reim.sym.2}
  R_{\alpha\beta\mu\nu}=R_{\mu\nu\alpha\beta}.
\end{equation}
\end{subequations}

\subsection{Ricci tensor and scalar curvature}
The Ricci tensor is defined from the contraction of the first and third index of the Riemann curvature tensor:
\begin{equation}\label{def:ricci}
  R_{\alpha\beta}\coloneq R\indices{^\mu_{\alpha\mu\beta}}.
\end{equation}
It is symmetric due to \eqref{eq:reim.sym.2}. Contracting the two indices of the Ricci tensor with the metric gives the scalar curvature
\begin{equation}\label{def:sc}
  R\coloneq R\indices{^\alpha_\alpha}.
\end{equation}

\section{Calculations}\label{app:calc}
A derivation of the non-differential form relations and the nonlinear relations in $3+1$ Einstein is recorded here.
\subsection{Non-differential form relations}\label{sec:nondiff.relations}

From the definition of $\mathcal{L}_\alpha$ \eqref{def:L}, \eqref{eq:wedge.hs} and \eqref{eq:kron.delta}, we get
\begin{align*}
  \mathcal{L}_\alpha&=-\frac{1}{2}\varepsilon_{\alpha\beta\mu\nu}\omega\indices{^{\beta\mu}_\gamma}\theta^\gamma\wedge\theta^\nu \\
  &=\frac{1}{2}\varepsilon_{\alpha\beta\mu\nu}\omega\indices{^{\beta\mu}_\gamma}(\frac{1}{2}\varepsilon^{\rho\tau\gamma\nu}\Sigma_{\rho\tau}) \\
  &=-\frac{1}{4}\delta^{\rho\tau\gamma}_{\alpha\beta\mu}\omega\indices{^{\beta\mu}_\gamma}\Sigma_{\rho\tau} \\
  &=-\frac{1}{2}(\delta^\rho_\alpha\omega\indices{^{\tau\gamma}_\gamma}-\delta^\tau_\alpha\omega\indices{^{\rho\gamma}_\gamma}+\omega\indices{^{\rho\tau}_\alpha})\Sigma_{\rho\tau},
\end{align*}
where we use the antisymmetry of $\omega$ in the first two indices \eqref{eq:metric.comp} for the last line. Using definition \eqref{def:hsBasis}, 
\begin{equation*}
\mathcal{L}_\alpha = -\frac{1}{2}(g_{\rho\alpha}\omega\indices{_\tau^\gamma_\gamma}-g_{\tau\alpha}\omega\indices{_\rho^\gamma_\gamma}+\omega\indices{_{\rho\tau\alpha}})\star(\theta^\rho\wedge \theta^\tau)
\end{equation*} 
whereby applying the Hodge star operator to this 2-form and recalling that $\star \star = -1$ on spacetime 2-forms \eqref{eq:double.hodge}, we see immediately that 
\begin{equation*}
\star\mathcal{L}_\alpha = \frac{1}{2}(g_{\rho\alpha}\omega\indices{_\tau^\gamma_\gamma}-g_{\tau\alpha}\omega\indices{_\rho^\gamma_\gamma}+\omega\indices{_{\rho\tau\alpha}})\theta^\rho\wedge \theta^\tau,
\end{equation*}
 or equivalently in component notation
 \begin{equation}\label{eq:starL}
\star\mathcal{L}\indices{_\alpha^\mu^\nu} = \delta^\mu_\alpha\omega\indices{^{\nu\gamma}_\gamma}-\delta^\nu_\alpha\omega\indices{^{\mu\gamma}_\gamma}+\omega\indices{^{\mu\nu}_\alpha}.
 \end{equation}
From \eqref{eq:star.decomp}, with $\mathcal{L}_\alpha$ a $2$-form, we get $\star \mathcal{L}_\alpha=-\bvec{n}^\flat\wedge\shs D_\alpha-\shs H_\alpha$; in components,
\begin{subequations}\label{eq:starL.decomp}
\begin{align}\label{eq:starL.decomp.1}
  \star\mathcal{L}\indices{_\alpha^\mu^\nu}
  &=-n^\mu\shs D\indices{_\alpha^\nu}+n^\nu\shs D\indices{_\alpha^\mu}-\shs H\indices{_\alpha^\mu^\nu}.
\end{align}
\end{subequations}

By \eqref{eq:omega.C}, $\omega\indices{^{\mu\nu}_\alpha}$ relates to $E_\alpha$ and $B_\alpha$ by 
\begin{align}\nonumber
  \omega\indices{^{\mu\nu}_\alpha}
  ={}&\frac{1}{2}
    (-\mathcal{C}\indices{^{\nu\mu}_{\alpha}}
     +\mathcal{C}\indices{^{\mu\nu}_{\alpha}}
     -\mathcal{C}\indices{_{\alpha}^{\mu\nu}})\\\label{eq:omegaEB}
  ={}&\frac{1}{2}\big[
    n^\mu E\indices{_\nu_\alpha}-n_\alpha E\indices{^\nu^\mu}- B\indices{^\nu^\mu_\alpha}\\\nonumber
    &-n^\nu E\indices{^\mu_\alpha}+n_\alpha E\indices{^\mu^\nu}+ B\indices{^\mu^\nu_\alpha}\\\nonumber
    &+n^\mu E\indices{_\alpha^\nu}-n^\nu E\indices{_\alpha^\mu}- B\indices{_\alpha^\mu^\nu}\big],
\end{align}
and the contraction $\omega\indices{^{\mu\alpha}_\alpha}$
\begin{equation}\label{eq:omegaEBcont}
  \omega\indices{^{\mu\alpha}_\alpha}
  =n^\mu E\indices{^\alpha_\alpha}+ E\indices{^0^\mu}- B\indices{^\alpha^\mu_\alpha}.
\end{equation}
Put \eqref{eq:omegaEB} and \eqref{eq:omegaEBcont} in \eqref{eq:starL} for
 \begin{align*}
\star\mathcal{L}\indices{_\alpha^\mu^\nu} 
  ={}&\delta^\mu_\alpha\omega\indices{^{\nu\gamma}_\gamma}-\delta^\nu_\alpha\omega\indices{^{\mu\gamma}_\gamma}+\omega\indices{^{\mu\nu}_\alpha} \\
  ={}&\delta^\mu_\alpha(n^\nu E\indices{^\gamma_\gamma}+ E\indices{^0^\nu}- B\indices{^\gamma^\nu_\gamma}) \\
  &-\delta^\nu_\alpha(n^\mu E\indices{^\gamma_\gamma}+ E\indices{^0^\mu}- B\indices{^\gamma^\mu_\gamma})\\ 
  &+\frac{1}{2}\big(
    n^\mu E\indices{_\nu_\alpha}-n_\alpha E\indices{^\nu^\mu}- B\indices{^\nu^\mu_\alpha}\\
    &-n^\nu E\indices{^\mu_\alpha}+n_\alpha E\indices{^\mu^\nu}+ B\indices{^\mu^\nu_\alpha}\\
    &+n^\mu E\indices{_\alpha^\nu}-n^\nu E\indices{_\alpha^\mu}- B\indices{_\alpha^\mu^\nu}\big).
 \end{align*}
We now take normal and tangential projections to get the relations between $\shs D$, $\shs H$ and $E$, $B$.
\begin{itemize}
  \item[Case 1]Contract with $n^\alpha$, $n_\mu$ and $h\indices{_\nu^j}$ (set $\nu$ spatial index), this extracts $\shs D_0$ from $\star\mathcal{L}_\alpha$.\\
    \begin{align*}
       \shs D\indices{_0^j}={}&-  E\indices{^0^j}+ B\indices{^\gamma^j_\gamma}+\frac{1}{2}\big(E\indices{^0^j}- E\indices{_0^j}\big) \\
      ={}& B\indices{^\gamma^j_\gamma}
    \end{align*}
  \item[Case 2]Contract with $h\indices{^\alpha_k}$, $n_\mu$ and $h\indices{_\nu^j}$ (set $\alpha$, $\nu$ spatial indices), this extracts $\shs D_k$ from $\hs\mathcal{L}_\alpha$:\\
  \begin{align*}
    \shs D\indices{_k^j} 
    ={}&\delta^j_k E\indices{^\gamma_\gamma}+\frac{1}{2}\big(-E\indices{^j_k}- B\indices{^0^j_k}- E\indices{_k^j}\big) \\
    ={}&\delta^j_k E\indices{^i_i}-\frac{1}{2}\big(E\indices{^j_k}+ E\indices{_k^j}\big)-\frac{1}{2} B\indices{^0^j_k}.
   \end{align*}
 
  \item[Case 3]Contract with $n^\alpha$, $h\indices{_\mu^i}$, $h\indices{_\nu^j}$ (set $\mu$, $\nu$ spatial indices), this extracts $-{\hs}H_0$ from $\star\mathcal{L}_\alpha$:\\
  \begin{align*}
  -{ \hs}H\indices{_0^i^j} 
    ={}&\frac{1}{2}\big(E^{ji}- E^{ij}- B\indices{_0^{ij}}\big).
  \end{align*}
        
  \item[Case 4]Contract with $h\indices{^\alpha_k}$, $h\indices{_\mu^i}$, $h\indices{_\nu^j}$ (set $\alpha$, $\mu$, $\nu$ spatial indices), this extracts $-\star H_k$ from $\star\mathcal{L}_\alpha$:\\
  \begin{align*}
  -(\star H_k)^{ij} 
    ={}&\delta^i_k(E\indices{^0^j}- B\indices{^\gamma^j_\gamma})-\delta^j_k( E\indices{^0^i}- B\indices{^\gamma^i_\gamma})+\frac{1}{2}\big(
    - B\indices{^j^i_k}+ B\indices{^i^j_k}- B\indices{_k^{ij}}\big) \\
    ={}&(\delta^i_k  E\indices{^0^j}-\delta^j_k E\indices{^0^i})-(\delta^i_k B\indices{^s^j_s}-\delta^j_k B\indices{^s^i_s})+\frac{1}{2}\big(
    - B\indices{^j^i_k}+ B\indices{^i^j_k}- B\indices{_k^{ij}}\big).
  \end{align*}
\end{itemize}

Now to invert the relations. For the inverse relations we use 
\begin{equation}\label{eq:C.omega}
  \mathcal{C}\indices{^\alpha_{\mu\nu}}=\omega\indices{^\alpha_{\mu\nu}}-\omega\indices{^\alpha_{\nu\mu}}
\end{equation}
from \eqref{eq:comm.coeff}. Contracting $\alpha$ and $\nu$ in \eqref{eq:starL} shows
\begin{equation}\label{eq:L.contract}
  \star\mathcal{L}\indices{_\alpha^\mu^\alpha}=(\omega\indices{^{\mu\gamma}_\gamma}-4\omega\indices{^{\mu\gamma}_\gamma}+\omega\indices{^{\mu\alpha}_\alpha})=-2\omega\indices{^{\mu\alpha}_\alpha},
\end{equation}
then this lets us invert \eqref{eq:starL} for an expression for $\omega\indices{^{\mu\nu}_\alpha}$ 
\begin{equation*}
  \omega\indices{^{\mu\nu}_\alpha}=\star\mathcal{L}\indices{_\alpha^\mu^\nu}+\frac{1}{2}\delta^\mu_\alpha{\hs}\mathcal{L}\indices{_\gamma^\nu^\gamma}-\frac{1}{2}\delta^\nu_\alpha{\hs}\mathcal{L}\indices{_\gamma^\mu^\gamma},
\end{equation*}
where the contraction
\begin{equation*}\label{eq:starL.decomp.2}
    \star\mathcal{L}\indices{_\gamma^\mu^\gamma}=-n^\mu \shs D\indices{_\gamma^\gamma}+ \shs D\indices{_0^\mu}- \shs H\indices{_\gamma^\mu^\gamma}.
\end{equation*}
Let us write
\begin{align}
  \omega\indices{^{\mu\nu}_\alpha}
  ={}&-n^\mu\shs D\indices{_\alpha^\nu}+n^\nu\shs D\indices{_\alpha^\mu}-\shs H\indices{_\alpha^\mu^\nu} \nonumber\\
  &+\frac{1}{2}\delta^\mu_\alpha(-n^\nu \shs D\indices{_\gamma^\gamma}+ \shs D\indices{_0^\nu}- \shs H\indices{_\gamma^\nu^\gamma}) \nonumber\\
  &-\frac{1}{2}\delta^\nu_\alpha(-n^\mu \shs D\indices{_\gamma^\gamma}+ \shs D\indices{_0^\mu}- \shs H\indices{_\gamma^\mu^\gamma}). \label{eq:omegaDH}
\end{align}
Thus \eqref{eq:C.omega} is
\begin{align*}
  \mathcal{C}\indices{^\alpha_{\mu\nu}}
  ={}&\omega\indices{^\alpha_{\mu\nu}}-\omega\indices{^\alpha_{\nu\mu}} \\
  ={}&-n^\alpha\shs D\indices{_\nu_\mu}+n_\mu\shs D\indices{_\nu^\alpha}-\shs H\indices{_\nu^\alpha_\mu}+\frac{1}{2}\delta^\alpha_\nu(-n_\mu \shs D\indices{_\gamma^\gamma}+\shs D\indices{_0_\mu}-\shs H\indices{_\gamma_\mu^\gamma})\\
  &\cancel{-\frac{1}{2}g_{\mu\nu}(-n^\alpha \shs D\indices{_\gamma^\gamma}+\shs D\indices{_0^\alpha}-\shs H\indices{_\gamma^\alpha^\gamma})} \\
  &-\big[-n^\alpha\shs D\indices{_\mu_\nu}+n_\nu\shs D\indices{_\mu^\alpha}-\shs H\indices{_\mu^\alpha_\nu}+\frac{1}{2}\delta^\alpha_\mu(-n_\nu \shs D\indices{_\gamma^\gamma}+\shs D\indices{_0_\nu}-\shs H\indices{_\gamma_\nu^\gamma})\\
  &\cancel{-\frac{1}{2}g_{\nu\mu}(-n^\alpha \shs D\indices{_\gamma^\gamma}+\shs D\indices{_0^\alpha}-\shs H\indices{_\gamma^\alpha^\gamma})}\big] \\
  ={}&-n^\alpha\shs D\indices{_\nu_\mu}+n_\mu\shs D\indices{_\nu^\alpha}-\shs H\indices{_\nu^\alpha_\mu}+\frac{1}{2}\delta^\alpha_\nu(-n_\mu \shs D\indices{_\gamma^\gamma}+\shs D\indices{_0_\mu}-\shs H\indices{_\gamma_\mu^\gamma})\\
  &n^\alpha\shs D\indices{_\mu_\nu}-n_\nu\shs D\indices{_\mu^\alpha}+\shs H\indices{_\mu^\alpha_\nu}-\frac{1}{2}\delta^\alpha_\mu(-n_\nu \shs D\indices{_\gamma^\gamma}+\shs D\indices{_0_\nu}-\shs H\indices{_\gamma_\nu^\gamma}).
\end{align*}
Now taking the same contractions as above to get $E_\alpha$, $B_\alpha$ in terms of $\shs D_\alpha$ and $\shs H_\alpha$, we obtain the following.
\begin{itemize}
  \item[Case 1] Contract with $n_\alpha$, $n^\mu$ and $h\indices{^\nu_j}$ (set $\nu$ spatial index) to extract $E_0$:\\
        \begin{align*}
          E\indices{_0_j}=-\frac{1}{2}\shs D\indices{_0_j}-\frac{1}{2}\shs H\indices{_\gamma_j^\gamma}.
        \end{align*}
        
  \item[Case 2] Contract with $h\indices{_\alpha^k}$, $n^\mu$ and $h\indices{^\nu_j}$ (set $\alpha$, $\nu$ spatial index) to extract $E^k$:\\
\begin{align*}
           E\indices{^k_j}={}&-\shs D\indices{_j^k}+\frac{1}{2}\delta^k_j \shs D\indices{_\gamma^\gamma}+\shs H\indices{_0^k_j}.
        \end{align*}
        
  \item[Case 3] Contract with $n_\alpha$, $h\indices{^\mu_i}$ and $h\indices{^\nu_j}$ (set $\mu$, $\nu$ spatial index) to extract $B_0$:\\
\begin{align*}
           B\indices{_0_i_j}={}& \shs D_{ji}- \shs D_{ij}.
        \end{align*}
        
  \item[Case 4] Contract with $h\indices{_\alpha^k}$, $h\indices{^\mu_i}$ and $h\indices{^\nu_j}$ (set $\alpha$, $\mu$, $\nu$ spatial index) to extract $B^k$:\\
\begin{align*}
           B\indices{^k_i_j}={}&-\shs H\indices{_j^k_i}+\frac{1}{2}\delta^k_j(\shs D_{0i}-\shs H\indices{_\gamma_i^\gamma})+\shs H\indices{_i^k_j}-\frac{1}{2}\delta^k_i(\shs D\indices{_0_j}-\shs H\indices{_\gamma_j^\gamma}) \\
          =&{}\shs H\indices{_i^k_j}-\shs H\indices{_j^k_i}+\frac{1}{2}(\delta^k_j\shs D_{0i}-\delta^k_i\shs D\indices{_0_j})-\frac{1}{2}(\delta^k_j\shs H\indices{_s_i^s}-\delta^k_i\shs H\indices{_s_j^s}).
        \end{align*}
\end{itemize}
\subsection{Nonlinear relations} \label{sec:nonlinear.relations}
By definition \eqref{eq:sparling}, $\mathcal{S}_\alpha$ can be written (we shuffle a bit the indices in the second term compared to the definition)
\begin{align*}
  \mathcal{S}_\alpha
  ={}&\frac{1}{2}\omega\indices{^\beta_\gamma}\wedge\omega^{\gamma\mu}\wedge\Sigma_{\alpha\beta\mu}
    -\frac{1}{2}\varepsilon_{\alpha\beta\mu\gamma}\omega\indices{^\beta_\rho}\wedge\omega^{\gamma\mu}\wedge\theta^\rho \\
  ={}&\frac{1}{2}(\varepsilon_{\alpha\beta\mu\rho}\omega\indices{^\beta_{\gamma\zeta}}\omega\indices{^{\gamma\mu}_\xi}
    -\varepsilon_{\alpha\beta\mu\gamma}\omega\indices{^\beta_{\rho\zeta}}\omega\indices{^{\gamma\mu}_\xi})\theta^\zeta\wedge\theta^\xi\wedge\theta^\rho \\
  ={}&-\frac{1}{2}(\varepsilon_{\alpha\beta\mu\rho}\omega\indices{^\beta_{\gamma\zeta}}\omega\indices{^{\gamma\mu}_\xi}
    -\varepsilon_{\alpha\beta\mu\gamma}\omega\indices{^\beta_{\rho\zeta}}\omega\indices{^{\gamma\mu}_\xi})\varepsilon^{\pi\zeta\xi\rho}\Sigma_\pi \\
  ={}&\frac{1}{2}(\delta_{\alpha\beta\mu}^{\pi\zeta\xi}\omega\indices{^\beta_{\gamma\zeta}}\omega\indices{^{\gamma\mu}_\xi}-\delta^{\pi\zeta\xi\rho}_{\alpha\beta\mu\gamma}\omega\indices{^\beta_{\rho\zeta}}\omega\indices{^{\gamma\mu}_\xi})\Sigma_\pi.
\end{align*}
Notice that in the second term, we can rename indices and expand $\delta$  to cancel the first term above with the first term below, and the third term is $0$ because $\omega$ is antisymmetric in the first two indices,
\begin{equation*}
  \delta^{\pi\zeta\xi\rho}_{\alpha\beta\mu\gamma}\omega\indices{^\beta_{\rho\zeta}}\omega\indices{^{\gamma\mu}_\xi}
  =\delta^{\pi\zeta\xi}_{\alpha\beta\mu}\omega\indices{^\beta_{\gamma\zeta}}\omega\indices{^{\gamma\mu}_\xi}
  -\delta^{\pi\zeta\xi}_{\alpha\beta\gamma}\omega\indices{^\beta_{\mu\zeta}}\omega\indices{^{\gamma\mu}_\xi}
  +\cancel{\delta^{\pi\zeta\xi}_{\alpha\mu\gamma}\omega\indices{^\beta_{\beta\zeta}}\omega\indices{^{\gamma\mu}_\xi}}
  -\delta^{\pi\zeta\xi}_{\beta\mu\gamma}\omega\indices{^\beta_{\alpha\zeta}}\omega\indices{^{\gamma\mu}_\xi},
\end{equation*}
so
\begin{align*}
  \mathcal{S}_\alpha
  =\frac12(\delta^{\pi\zeta\xi}_{\alpha\beta\gamma}\omega\indices{^\beta_{\mu\zeta}}\omega\indices{^{\gamma\mu}_\xi}
  +\delta^{\pi\zeta\xi}_{\beta\mu\gamma}\omega\indices{^\beta_{\alpha\zeta}}\omega\indices{^{\gamma\mu}_\xi})\Sigma_\pi.
\end{align*}
Expand the Kronecker deltas, and collect the terms antisymmetric in the last two indices to get 
\begin{align*}
  \delta^{\pi}_{\alpha}\omega\indices{^\beta_{\mu[\beta}}\omega\indices{^{\gamma\mu}_{\gamma]}}
  {}&-\omega\indices{^\pi_{\mu[\alpha}}\omega\indices{^{\gamma\mu}_{\gamma]}}+\omega\indices{^\beta_{\mu[\alpha}}\omega\indices{^{\pi\mu}_{\beta]}}
  +\omega\indices{^\pi_{\alpha\mu}}\omega\indices{^{\gamma\mu}_\gamma}
  +\omega\indices{^\beta_{\alpha\beta}}\omega\indices{^{\pi\gamma}_\gamma}
  +\omega\indices{^\beta_{\alpha\gamma}}\omega\indices{^{\gamma\pi}_\beta} \\
  ={}& -\delta^{\pi}_{\alpha}\omega\indices{^\beta_{\mu[\beta}}\omega\indices{^{\mu\gamma}_{\gamma]}}
  +\omega\indices{^\beta_{\mu\alpha}}\omega\indices{^{\pi\mu}_\beta}
  -\omega\indices{^\pi_{\mu\alpha}}\omega\indices{^{\gamma\mu}_\gamma}
  +\omega\indices{^\pi_{\alpha\mu}}\omega\indices{^{\gamma\mu}_\gamma}
  +\omega\indices{^\beta_{\alpha\beta}}\omega\indices{^{\pi\gamma}_\gamma}
  -\omega\indices{^\beta_{\alpha\mu}}\omega\indices{^{\pi\mu}_\beta}\\
  ={}& -\delta^{\pi}_{\alpha}\omega\indices{^\beta_{\mu[\beta}}\omega\indices{^{\mu\gamma}_{\gamma]}}
  +2\omega\indices{^\beta_{[\mu\alpha]}}\omega\indices{^{\pi\mu}_\beta}
  +2\omega\indices{^\pi_{[\mu\alpha]}}\omega\indices{^{\mu\gamma}_\gamma}
  +\omega\indices{^\beta_{\alpha\beta}}\omega\indices{^{\pi\gamma}_\gamma}  
\end{align*}
Using the relations of $\star\mathcal{L}_\alpha$ \eqref{eq:starL} and $\mathcal{C}^\alpha$ \eqref{eq:C.omega} with the connection $1$-forms, we see
\begin{align*}
  \mathcal{C}\indices{^\alpha_{\mu\nu}}{\hs}\mathcal{L}\indices{_\alpha^\mu^\nu}   
  ={}&2\omega\indices{^\alpha_{\mu\nu}}(\delta^\mu_\alpha\omega\indices{^{\nu\gamma}_\gamma}-\delta^\nu_\alpha\omega\indices{^{\mu\gamma}_\gamma}+\omega\indices{^{\mu\nu}_\alpha})\\
  ={}&2(-\omega\indices{^\nu_{\mu\nu}}\omega\indices{^{\mu\gamma}_\gamma}+\omega\indices{^\alpha_{\mu\nu}}\omega\indices{^{\mu\nu}_\alpha}) \\
  ={}&-4\omega\indices{^\nu_{\mu[\nu}}\omega\indices{^{\mu\gamma}_{\gamma]}}, \\
  \mathcal{C}\indices{^\alpha_{\mu\nu}}{\hs}\mathcal{L}\indices{_\alpha^\mu^\beta}
  ={}&2\omega\indices{^\alpha_{[\mu\nu]}}(\delta^\mu_\alpha\omega\indices{^{\beta\gamma}_\gamma}-\delta^\beta_\alpha\omega\indices{^{\mu\gamma}_\gamma}+\omega\indices{^{\mu\beta}_\alpha}) \\
  ={}&-\omega\indices{^\mu_{\nu\mu}}\omega\indices{^{\beta\gamma}_\gamma}
  		-2\omega\indices{^\beta_{[\mu\nu]}}\omega\indices{^{\mu\gamma}_\gamma}
  		+2\omega\indices{^\alpha_{[\mu\nu]}}\omega\indices{^{\mu\beta}_\alpha},
\end{align*}
and putting this into the expansion of $\mathcal{S}_\alpha$ gives
\begin{equation}\label{eq:sparl.comp}
  \mathcal{S}_\alpha=\underbrace{(\frac14\delta^\pi_\alpha\mathcal{C}\indices{^{\beta}_{\mu\nu}}{\hs}\mathcal{L}\indices{_\beta^\mu^\nu}-\mathcal{C}\indices{^\beta_{\mu\alpha}}{\hs}\mathcal{L}\indices{_\beta^\mu^\pi})}_{{\hs}\mathcal{S}\indices{_\alpha^\pi}}\frac{1}{3!}\varepsilon_{\pi\gamma\rho\tau}\theta^{\gamma}\wedge\theta^{\rho}\wedge\theta^{\tau}.
\end{equation}
By the normal/tangential decompositions of $\mathcal{C}^\beta$ \eqref{eq:C}, $\star\mathcal{L}_{\beta}$ \eqref{eq:starL.decomp.1}, we can work out the normal/tangential components of the formula \eqref{eq:star.decomp}
\begin{align*}
  \star\mathcal{S}_\alpha&=\bvec{n}^\flat\wedge\shs  V_\alpha - \shs  U_\alpha, \\
  {\hs}\mathcal{S}\indices{_\alpha^\pi}&=n^\pi\shs  V_\alpha-{\hs}U\indices{_\alpha^\pi},
\end{align*}
and expanding \eqref{eq:sparl.comp}
\begin{align*}
  {\hs}\mathcal{S}\indices{_\alpha^\pi}
  ={}&\frac14\delta^\pi_\alpha(-n_\mu E\indices{^\beta_\nu}+n_\nu E\indices{^\beta_\mu}+ B\indices{^\beta_\mu_\nu})(-n^\mu\shs D\indices{_\beta^\nu}+n^\nu\shs D\indices{_\beta^\mu}-\shs H\indices{_\beta^\mu^\nu})\\
  &-(-n_\mu E\indices{^\beta_\alpha}+n_\alpha E\indices{^\beta_\mu}+ B\indices{^\beta_\mu_\alpha})(-n^\mu\shs D\indices{_\beta^\pi}+n^\pi\shs D\indices{_\beta^\mu}-\shs H\indices{_\beta^\mu^\pi}) \\
  ={}&\frac14\delta^\pi_\alpha(-2 E\indices{^\beta_\nu}\shs D\indices{_\beta^\nu}-2\shs B\indices{^\beta_\mu} H\indices{_\beta^\mu})\\
  &+ E\indices{^\beta_\alpha}\shs D\indices{_\beta^\pi}-n_\alpha n^\pi E\indices{^\beta_\mu}\shs D\indices{_\beta^\mu}+n_\alpha E\indices{^\beta_\mu}\shs H\indices{_\beta^\mu^\pi}-n^\pi B\indices{^\beta_\mu_\alpha}\shs D\indices{_\beta^\mu}\\
  &+ B\indices{^\beta_\mu_\alpha}\shs H\indices{_\beta^\mu^\pi}.
\end{align*}
Taking projections, we obtain the following.
\begin{itemize}
  \item[Case 1] Contract ${\hs}\mathcal{S}\indices{_\alpha^\pi}$ with $n^\alpha$, $n_\pi$ to get $-\shs  V_0$:
    \begin{align*}
    -\shs  V_0
    ={}&-\frac14(-2 E\indices{^\beta_\nu}\shs D\indices{_\beta^\nu}-2\shs B\indices{^\beta_\mu} H\indices{_\beta^\mu})- E\indices{^\beta_\mu}\shs D\indices{_\beta^\mu} \\
    ={}&-\frac12 E\indices{^\beta_\mu}\shs D\indices{_\beta^\mu}+\frac12 H\indices{_\beta^\mu}\shs B\indices{^\beta_\mu}.
    \end{align*}
    
  \item[Case 2] Contract ${\hs}\mathcal{S}\indices{_\alpha^\pi}$ with $h\indices{^\alpha_i}$, $n_\pi$ to get $-\shs  V_i$ ($i$ spatial index):
  \begin{align*}
    -\shs  V_i
    &= B\indices{^\beta_\mu_i}\shs D\indices{_\beta^\mu}.
   \end{align*}
   
  \item[Case 3] Contract ${\hs}\mathcal{S}\indices{_\alpha^\pi}$ with $n^\alpha$, $h\indices{^j_\pi}$ to get $-\shs U\indices{_0^j}$:
  \begin{align*}
    -\shs U\indices{_0^j}
    &=- E\indices{^\beta_\mu}\shs H\indices{_\beta^\mu^i}.
  \end{align*}
  
  \item[Case 4] Contract ${\hs}\mathcal{S}\indices{_\alpha^\pi}$ with $h\indices{^\alpha_i}$, $h\indices{^j_\pi}$ to get $-(\shs  U_i)^j$ ($i,j$ spatial indices):
  \begin{align*}
    -(\shs  U_i)^j
    &=\delta^j_i(-\frac12 E\indices{^\beta_\nu}\shs D\indices{_\beta^\nu}-\frac12\shs B\indices{^\beta_\mu} H\indices{_\beta^\mu})+ E\indices{^\beta_i}\shs D\indices{_\beta^j}+ B\indices{^\beta_\mu_i}\shs H\indices{_\beta^\mu^j} \\
    &=\delta^j_i(-\frac12 E\indices{^\beta_k}\shs D\indices{_\beta^k}+\frac12 H\indices{^\beta^l}\shs B\indices{_\beta_l})+ E\indices{^\beta_i}\shs D\indices{_\beta^j}- H\indices{^\beta_i}\shs B\indices{_\beta^j}.
  \end{align*}
\end{itemize}

\section{Discrete nonlinear formulas}\label{sec:disc.formulas}

In ECDDR, we work on forms that have polynomial components with respect to the canonical basis $(\ed x^I)_{I=1}^3$. When applying the formulas in Section \ref{sec:non.diff.rel}, that only hold on components in the $(\theta^i)_{i=1}^3$ basis, an added change of basis must be performed. We list here the process and full formulas used.

First we use the approximation
\begin{equation*}
	\theta^i \approx [\Pec{1}{r,h}\underline{\theta}_h^i]_J\ed x^J\,,\qquad \ed x^I = [\ed x^I]_j\theta^j
\end{equation*}
where $[\ed x^I]_j\approx[(\Pec{1}{r,h}\underline{\theta}_h)^{-1}]\indices{^I_j}$, which is justified since $\delta^I_K = [\ed x^I]_j[\theta^j]_K\approx [\ed x^I]_j[\Pec{1}{r,h}\underline{\theta}_h^j]_K$. Using these in the change of basis, we calculate
\begin{align*}
	B\indices{^i}&\approx[\ued{1}{r,h}\underline{\theta}_h^i]_{IJ}\ed x^I\wedge\ed x^J \approx\underbrace{[\ued{1}{r,h}\underline{\theta}_h^i]_{IJ}[\Pec{1}{r,h}\underline{\theta}_h^{-1}]\indices{^I_j}[\Pec{1}{r,h}\underline{\theta}_h^{-1}]\indices{^J_k}}_{\approx B\indices{^i_{jk}}} \theta^j\wedge\theta^k,  \\
	D\indices{^i}&\approx[\chs \Pec{1}{r,h}\underline{\chs D}_h^i]_{IJ}\ed x^I\wedge\ed x^J\approx\underbrace{[\chs \Pec{1}{r,h}\underline{\chs D}_h^i]_{IJ}[\Pec{1}{r,h}\underline{\theta}_h^{-1}]\indices{^I_j}[\Pec{1}{r,h}\underline{\theta}_h^{-1}]\indices{^J_k}}_{\approx D\indices{^i_{jk}}} \theta^j\wedge\theta^k.
\end{align*}
Taking then the $\shs$, which is just a rearrangement of the coefficients $[(\shs\omega)_{1}, (\shs\omega)_{2}, (\shs\omega)_{3}]=[\omega_{23}, -\omega_{13}, \omega_{12}]$, we get the approximations of $\shs B^i$, $\shs D^i$, that are plugged into the formulas
\begin{equation*}
	H_{kl} =\shs B_{lk}-\frac{1}{2}\delta_{kl}\shs B\indices{_i^i}, \qquad
	E_{kl} =-\big(\shs D_{lk}-\frac{1}{2}\delta_{kl}\shs D\indices{_i^i}\big).
\end{equation*}
The final step is to return to the $\ed x^I$ basis by
\begin{align*}
	H^i(\ued{1}{r,h}\underline{\theta}_h, \underline{\theta}_h)\approx H\indices{^i_j}\theta^j\approx \underbrace{H\indices{^i_j}[\Pec{1}{r,h}\underline{\theta}_h^j]_I}_{H\indices{^i_I}} \ed x^I, \\
	E^i(\Pec{1}{r,h}\underline{\chs D}_h(n), \underline{\theta}_h(n))\approx \underbrace{E\indices{^i_j}\theta^j\approx E\indices{^i_j}[\Pec{1}{r,h}\underline{\theta}_h^j]_I}_{E\indices{^i_I}} \ed x^I.
\end{align*}

%------------------------------------------------------------------------------%
% Bibliography
%------------------------------------------------------------------------------%

\printbibliography

\end{document}